\documentclass[aps,prd,twocolumn,superscriptaddress,nofootinbib,floatfix]{revtex4-2}
\usepackage{amsmath,amssymb,bm}
\usepackage{graphicx}
\usepackage{hyperref}
\usepackage{orcidlink}
\usepackage{booktabs}

\newtheorem{theorem}{Theorem}
\newtheorem{proposition}{Proposition}
\newenvironment{proof}{\par\noindent\textit{Proof.}\ }{\hfill$\square$\par}

\newcommand{\Pc}{P_c}
\newcommand{\Z}{\mathbb{Z}}

\begin{document}

\title{Mass as pairing: an explicit Majorana mass on half of a generation does not seesaw in symmetric mass generation}

\author{Piotr S. Topa\,\orcidlink{0009-0005-2870-2463}}
\affiliation{Independent researcher}

\date{October 6, 2026}

\begin{abstract}
A fermion mass is a pairing between a field and a partner of conjugate charge: an elementary fermion (the Higgs--Dirac mass), a self-conjugate partner with a Majorana mass (the seesaw), or a three-fermion composite (symmetric mass generation, SMG). We give half of a Spin(10)-like generation an explicit Majorana mass in a four-dimensional SMG model (four reduced staggered flavours with an on-site vertex) and ask what the other half does in a phase with no condensate. The statement of the title is a symmetry statement. Theorem: with the source on, every mass-type term of the unsourced half---the light part of each of the sixteen lattice strings that anticommute with the kinetic operator, flavour-symmetric and flavour-sextet---averages to zero under the symmetry group of the sourced action, so a light mass term, and with it the light--heavy mixing a seesaw needs, requires a spontaneously broken symmetry; this is an exact lattice statement on every box carrying the symmetry and a general proposition. The simulated source is colour-democratic: a $\Delta_R$-type Majorana mass on the right-handed taste doublet of every flavour, breaking $SU(4)_{\rm PS}$ to $SO(4)$; the single-$\nu_R$ split lies outside the $\tau_2K$ (quaternion-linear) positivity class used here, and the symmetry statement covers it. The measured content, on lattices up to $8^4$, is (i) the unsourced half in the symmetric-gap phase under a source whose own linear response is screened to at most $1.1\,\%$ of free: its single-fermion readout at $p_{\min}$ stays at $0.3\,\%$ of free and its odd-part exponent at its SMG-phase value, the same for both doublets, while its pair susceptibility relative to free rises from $0.2\,\%$ to $1.2$--$1.7\,\%$ at $h=2$ on $8^4$; on like mode sets that susceptibility falls with the volume on both available pairs at $h=2$, on the $(4^4,8^4)$ pair at $h=2$ no inequivalent light mass channel grows faster than free, and symmetry breaking at that point is not resolved at $L\le8$ beyond this two-volume statement; and (ii) the detuning at the critical point, where the explicit mass collapses the critical $\sigma$ channel, raises the light pair-channel effective mass monotonically toward free as a saturating crossover rather than a power law, and moves the light half toward its symmetric-phase values. The fermion-loop curvature and the direction tests on two volumes, MIXED by their pre-registered letter, are consistent with a critical coupling rising with the mass; $y_c(h)$ is not located. The pole-versus-zero character of the light propagator at $p\to0$ is not measured at $L\le8$. The reported outcomes are readings of the data after the pre-registered tests, most of which returned no verdict, MIXED or VOID; the criteria that returned a verdict as written are named in Sec.~\ref{sec:methods}. Beside this: no first-order signature of the merged SMG transition at $L\le8$, the hand-over of a $\Z_4$-odd response from the elementary fermion to a lattice composite, and the flavour-democracy corollary for Majorana sources in the $\tau_2K$ class. Every claim is listed with its evidence class (proven, machine-checked, or measured) and a runnable check in the companion repository.
\end{abstract}

\maketitle

\section{Introduction}
\label{sec:intro}

In the Standard Model every fermion mass is a pairing: a left-handed field is paired with a right-handed partner of conjugate charge, and the Higgs condensate supplies the charge that the pairing would otherwise violate. The seesaw~\cite{Minkowski:1977sc,Gell-Mann:1979vob,Yanagida:1979as,Mohapatra:1979ia} pairs the neutrino with a partner that is its own conjugate, a right-handed neutrino with a Majorana mass. Symmetric mass generation (SMG)~\cite{Eichten:1985ft,Wang:2022ucy,Tong:2021phe,Razamat:2020kyf}, reviewed in Refs.~\cite{Hasenfratz:2025lti,Hasenfratz:2026orf}, pairs a fermion with a composite built from the fermions themselves, so that a gap opens with no condensate. This paper asks what happens when two of these meet in one multiplet.

The organising object is a $2\times2$ partner matrix. For a fermion $\chi$ with Dirac mixing $\Delta$ to a partner $\Psi$ that may carry a Majorana mass $M_p$,
\begin{equation}
G^{-1}=\begin{pmatrix} ip & \Delta\\ \Delta & ip+M_p\end{pmatrix},
\label{eq:partner}
\end{equation}
the light propagator obeys $G_{\chi\chi}(p\to0)=-M_p/\Delta^2$. For $M_p=0$ it has a zero at $p=0$ with gap $\Delta$, $|G_{\chi\chi}|=|p|/(p^2+\Delta^2)$; for $M_p\gg\Delta$ the light eigenvalue is the seesaw mass $\Delta^2/M_p$. The Higgs--Dirac mass (elementary partner, $M_p=0$), the seesaw (elementary partner, $M_p\gg\Delta$) and SMG (composite partner, $M_p=0$ forced by symmetry) are the three corners of this one matrix, and all three are known. The SMG corner, in which a Weyl fermion of the Spin(10) $\mathbf{16}$ pairs with the three-fermion composite $\bar\psi\bar\psi\psi\in\overline{\mathbf{16}}$, is the strong-coupling picture of Eichten and Preskill~\cite{Eichten:1985ft,Golterman:1992yha}, restated by Golterman and Shamir~\cite{Golterman:2023zqf,Golterman:2025boq} and by Wang and You~\cite{Wang:2022ucy,You:2017ltx}. The anomaly-free remnant of the fermion-number $U(1)$ of one generation is a $\Z_4$~\cite{Garcia-Etxebarria:2018ajm,Hsieh:2018ifc,Davighi:2019rcd,Wan:2019gqr,Razamat:2020kyf,Wang:2020xyo}, whose generator squares to fermion parity, $\Z_4^2=(-1)^F$. Every fermionic operator, elementary or composite, has odd $\Z_4$ charge, and every Majorana bilinear has charge $2$~\cite{Tong:2021phe,Razamat:2020kyf,Wang:2020xyo}. A $\Z_4$-symmetric phase therefore has no Majorana (self-conjugate) pole, while a Dirac pole pairing $\psi$ (charge $1$) with $\bar\psi\bar\psi\psi$ (charge $3\equiv-1$) is allowed. In the matrix, SMG is the composite-partner corner with $M_p=0$ forced by the $\Z_4$. The matrix is bookkeeping, the frame for the question below.

The question put to the lattice here is: give half of a Spin(10)-like generation an explicit Majorana mass and ask what the other half does in a phase with no condensate. We found no earlier lattice study of it (Sec.~\ref{sec:discussion}). In the Higgs language the answer is a selection rule---a Majorana mass for the left-handed neutrino needs $\langle H\rangle^2$~\cite{Weinberg:1979sa}---and in Spin-$\Z_4$ language a Majorana neutrino mass breaks the $\Z_4$~\cite{Razamat:2020kyf,Tong:2021phe,Wang:2020xyo,Wang:2020gqr}. The content of this paper is the statement and its test inside an un-Higgsed phase in which composite partners exist. The source we can simulate is fixed by positivity. The $\tau_2K$ (quaternion-linear) positivity class of Majorana sources, the class the RHMC uses, is flavour-democratic (Sec.~\ref{sec:model}), so the ``half'' is the right-handed taste doublet of every flavour, and the source is a $\Delta_R$-type Majorana mass with the same value for all four Pati--Salam colours, breaking $SU(4)_{\rm PS}$ to $SO(4)$. The single-$\nu_R$ split lies outside that class; the symmetry statement (Theorem~\ref{thm:zeros}, Proposition~\ref{prop:noseesaw}) covers it. We prove that no mass term can appear on the unsourced half unless a symmetry breaks spontaneously, and we measure on lattices up to $8^4$. In the symmetric-gap phase the unsourced half keeps its gap, read by its single-fermion momentum readout and odd-part exponent, with no inequivalent light mass channel growing faster than free on the $(4^4,8^4)$ pair and symmetry breaking not resolved at $L\le8$. At the critical point the explicit mass moves the unsourced half toward its symmetric-phase values, and the fermion-loop curvature together with the direction tests are consistent with a critical coupling that rises with the mass. No light mass-type term can appear without spontaneous symmetry breaking (Theorem~\ref{thm:zeros}), and the single-fermion light readouts do not move under the source at $L\le8$.

The lattice model is the four-flavour reduced-staggered model of Butt, Catterall, Hartshaw and Hasenfratz (BCHH)~\cite{Butt:2026lkt}, whose $SU(4)$-invariant on-site vertex decouples into a $\sigma$ triplet and whose four flavours carry the Pati--Salam content $(\mathbf4,\mathbf2,\mathbf1)\oplus(\bar{\mathbf4},\mathbf1,\mathbf2)$ of one $\mathbf{16}$~\cite{Pati:1974yy,Butt:2021brl,Catterall:2020fep}. Around its merged critical point $\Pc$ we also report three results of independent interest: the merged transition shows no first-order signature at $L\le8$ (a test BCHH do not report); the response to a $\Z_4$-odd field moves from the elementary fermion to its composite partner across the transition; and the $\tau_2K$ class of Majorana sources is flavour-democratic.

Every claim in the text corresponds to an entry of the claim map (Appendix~\ref{app:claims}), which gives its identifier, its evidence class---proven (symbolic or algebraic), machine-checked (exact finite verification with quoted tolerances), or measured (always with a known-answer gate run through the identical pipeline)---and its runnable check in the companion repository~\cite{repo}. Secs.~\ref{sec:model}--\ref{sec:nomass} contain the structure and the symmetry statements; Sec.~\ref{sec:methods} the instruments and the verdict ledger; Secs.~\ref{sec:order}--\ref{sec:halfmass} the measurements; Sec.~\ref{sec:discussion} the neutrino selection rules, the limits of $L\le8$ and the prior art.

\subsection{The physical picture in plain terms}
\label{sec:picture}

\paragraph{The question.} A Weyl $\mathbf{16}$ of Spin(10) can pair with itself only in the $\mathbf{10}$ or the $\mathbf{126}$; a $\mathbf{126}$ condensate in its $SU(5)$-singlet direction gives exactly one Majorana mass, the right-handed neutrino's, which is the input of the seesaw~\cite{Babu:1992ia}. In the gauged Spin(10) theory only a $\Z_4$ of the fermion-number $U(1)$ survives the anomalies; in the lattice model, which has no gauge field, the $\varepsilon$ vertex breaks the lattice $U(1)_\varepsilon$ explicitly to a $\Z_4$ (row K1.6 of Table~\ref{tab:algebra}). Its generator acts as $i$ on every field of the $\mathbf{16}$ and as $-1$ on any fermion pair. A phase that keeps the $\Z_4$ can give no fermion, elementary or composite, a self-conjugate mass, but it can still pair a fermion with a partner of opposite charge, and the composite $\bar\psi\bar\psi\psi$ is such a partner: that pairing is SMG. The question is what an explicit Majorana mass on half of the multiplet does to the other half in such a phase, where the partners are composites and there is no Higgs field.

\paragraph{The idea.} A mass term for the light half is a bilinear with definite quantum numbers under the symmetries that survive the explicit mass. When every such bilinear transforms non-trivially under an exact symmetry, its expectation value is zero unless that symmetry breaks spontaneously; a seesaw block built from symmetric couplings is such a bilinear. The light half can then be gapped only symmetrically, as a Green's-function zero, or by spontaneous breaking. The discriminator between them is the mass-type part of the light propagator at $p\to0$: pole or zero.

\paragraph{What was done.} The lattice model has four real staggered fields per site, coupled by a four-fermion vertex decoupled into a three-component boson $\sigma$. At weak coupling the fermions are massless (SYM); at strong coupling they are gapped with no condensate (SMG); the two phases meet at the critical point $\Pc$. A taste-chiral Majorana mass $hC_\chi$ gaps exactly one of the two taste doublets of every flavour in the free theory (the heavy half) and leaves the other (the light half) massless: this is model N1. The mass-type strings of the light half, those anticommuting with the kinetic operator, are enumerated and their group averages under the surviving symmetries computed exactly. The light half is then measured in the SMG phase and at $\Pc$ on boxes up to $8^4$ with $h$ up to $3$, every number read against the free theory on the same box.

\paragraph{What was found.} With $hC_\chi$ on, every mass-type term of the light half averages to zero under an exact symmetry of the sourced action (Theorem~\ref{thm:zeros}). In the SMG phase, where the source's own linear response is screened to at most $1.1\,\%$ of free, the light half keeps its gap as $h$ grows, with a small drift toward free; on the $(4^4,8^4)$ pair no inequivalent light mass channel grows faster than free, and whether a symmetry breaks there is not resolved at $L\le8$. At $\Pc$ the heavy half's mass removes it from the fermion loop that drives the critical $\sigma$ channel; the channel collapses, and the light half at fixed coupling becomes more free-like, never more gapped than free. The response is a saturating crossover with scale $h\approx2$, where the heavy mass reaches the lattice bandwidth, not a power law; the data are consistent with a critical coupling rising with $h$, and the critical line $y_c(h)$ itself is not located.

\section{One generation and its lattice proxy}
\label{sec:model}

Table~\ref{tab:algebra} collects the exact facts on which the rest stands. The group theory is classical~\cite{Slansky:1981yr,Georgi:1974my,Fritzsch:1974nn}; the anomaly computations reproduce the published conditions~\cite{Garcia-Etxebarria:2018ajm,Hsieh:2018ifc,Ibanez:1991hv}; the certified content is the enumeration, the identities and the channel fractions. The identities are standard: $\Phi_{126}\!\cdot\!\Phi_{126}\equiv0$ because $\mathbf{126}\otimes\mathbf{126}$ contains no singlet (the invariant pairs the $\mathbf{126}$ with its conjugate $\overline{\mathbf{126}}$), and the uniqueness of the quartic singlet is a plethysm count. The consequence used here is that the only $U(1)_\psi$-breaking quartic of a $\mathbf{16}$ lives in the $\mathbf{10}$ channel, $(\Phi_{126}\Phi_{126})_{54}=-10\,(\Phi_{10}\Phi_{10})_{54}$; together with the $\mathbf{10}/\mathbf{126}$ content of the lattice bilinears this fixes which continuum channel each lattice operator probes.

\begin{table*}[!t]
\caption{Exact structure of one Spin(10) generation and of its four-flavour lattice proxy. Class: P proven, M machine-checked (tolerances as asserted by the check scripts), X measured.}
\label{tab:algebra}
\begin{ruledtabular}
\begin{tabular}{p{0.75\textwidth}cl}
statement & class & claim\\
\colrule
$\mathrm{Sym}^2(\mathbf{16})=\mathbf{10}\oplus\mathbf{126}$, $\Lambda^2(\mathbf{16})=\mathbf{120}$; a generic $\mathbf{126}$ vev has rank 16, the $SU(5)$ singlet rank 1 & M & K1.1\\
Unique holomorphic quartic $\Phi_{10}\!\cdot\!\Phi_{10}$; $\Phi_{126}\!\cdot\!\Phi_{126}\equiv0$; no charge-6 invariant, one charge-8 invariant; $(\Phi_{126}\Phi_{126})_{54}=-10(\Phi_{10}\Phi_{10})_{54}$ & M & K1.2\\
Fierz basis of $\psi\psi\bar\psi\bar\psi$: $\{|\Phi_{10}|^2,|\Phi_{126}|^2\}$; $J_1\!\cdot\!J_1=-\tfrac{1}{256}|\Phi_{10}|^2-\tfrac{1}{512}|\Phi_{126}|^2$ (to $10^{-14}$) & M & K1.3\\
Dai--Freed: $\xi(\mathbb{RP}^5,\mathrm{Spin}\text{-}\Z_4)=1/16$, so $\nu\equiv0\bmod16$; Spin-$\Z_8$, $\Z_{16}$ need $\nu\equiv0\bmod32,64$; 't~Hooft vertex $\psi^4$, so $\Z_n\subset U(1)_\psi$ is anomaly-free iff $n\,|\,4$ & M & K1.4\\
Pati--Salam: $\mathbf{16}=(\mathbf4,\mathbf2,\mathbf1)\oplus(\bar{\mathbf4},\mathbf1,\mathbf2)$; the on-site $\Lambda^2(\mathbf4)$ (lattice $\varepsilon$) channel is $0.500000$ in the $\mathbf{10}$ and in the $\mathbf{126}$; the two-site $\mathrm{Sym}^2(\mathbf4)$ channel is entirely in the $\mathbf{126}$ ($\le6\times10^{-31}$ in the $\mathbf{10}$) & M & K1.5\\
Lattice $U(1)_\varepsilon$ is orthogonal to $\mathrm{spin}(10)\oplus u(1)_\psi$; the $\varepsilon$ vertex breaks it to $\Z_4$; the lattice $\Z_4$ equals the $U(1)_\psi$ $\Z_4$ up to a Pati--Salam-central Spin(10) element & M & K1.6\\
$\Z_4^2=(-1)^F$: every Majorana bilinear has charge 2; partner matrix~(\ref{eq:partner}); $\bar\psi\bar\psi\psi$ contains the $\overline{\mathbf{16}}$ partner 3 times; $\psi\cdot(\bar\psi\bar\psi\psi)$ has 2 Grassmann invariants, the quartics of K1.3 & P, M & K1.7\\
\colrule
$\tau_2K$: $D_c=K\otimes1+iy\,\sigma\!\cdot\!\tau$ anti-hermitian, spectrum paired $\pm i\mu$ ($1.1\times10^{-13}$), $\det D_c>0$; on-site block a quaternion, $\mathrm{tr}\,S(x,x)=0$ & P, M & K2.1\\
Pure $\mathrm{Sym}^2$ ($\mathbf{10}$-channel) term: no sign-free route found (HS schemes 47/43/17/45\,\% negative weights); its sum-of-squares completion is sign-free & M & K2.2, K2.3\\
Flavour democracy (corollary): the $\tau_2K$ class is $\mathbb{H}^{\mathrm{op}}\otimes\mathbb{R}^{V\times V}$ and commutes with a hidden $SU(2)$ transitive on flavour directions, so every Majorana source in that class gives all four flavours the same $|m|$; class distance of a 1-of-4 split $0.8660$, of a 2+2 split $0.7071$ & P, M & K2.4\\
$C_0$ is a nodal pairing (16 corner zero modes at every $h$, no induced corner mass to $O(h^5)$); $C_\chi$ is a Majorana mass of gap $h$ for exactly one taste doublet (8 zero + 8 at $h$) & M & K2.5\\
2+2 split: every fermion-bag weight a perfect square; a conjugate-pair HS gives $|\mathrm{Pf}|^2$ & P, M & K2.6\\
Calibration to Ref.~\cite{Butt:2026lkt}: $y=\sqrt2\,y_{\rm ref}$ (0 of 15 $L=8$ points misclassified); $\kappa_c=-0.01\pm0.006$ transferred, $\Pc=(y,\kappa)=(2.41,-0.01)$; $\kappa_c\in(-0.03,0.05)$ from our scan & X & I.3\\
\end{tabular}
\end{ruledtabular}
\end{table*}

\paragraph{The model.} Four reduced staggered flavours $\chi^a(x)$, one real Grassmann variable per site and flavour, with the on-site vertex $U\sum_x\chi^1\chi^2\chi^3\chi^4$ decoupled by a real $\sigma$ triplet (the $\varepsilon$ model); even sites carry the $\mathbf4$ of $SU(4)$, odd sites the $\bar{\mathbf4}$. The action is that of BCHH~\cite{Butt:2026lkt},
\begin{align}
S&=\tfrac12\sum_{x,y}\chi^T(x)\bigl[K_{xy}-h\,(C_\chi)_{xy}\bigr]\chi(y)\nonumber\\
&\quad+\sum_x\Bigl[\tfrac{y}{2}\,\sigma_a(x)\,\chi^T(x)\Gamma_a\chi(x)+\tfrac12|\sigma|^2+\tfrac{\lambda}{4}|\sigma|^4\Bigr]\nonumber\\
&\quad-\tfrac{\kappa}{2}\sum_x\sigma\!\cdot\!\Box\sigma,
\label{eq:action}
\end{align}
with $K$ the reduced staggered kinetic matrix and $C_\chi$ the source, both flavour-blind, $\Gamma_a$ the antisymmetric Yukawa triplet in flavour space (on the complex doublet the Yukawa term is $iy\,\sigma\!\cdot\!\tau$, row K2.1 of Table~\ref{tab:algebra}), $\Box\sigma(x)=\sum_\mu[\sigma(x+2\hat\mu)+\sigma(x-2\hat\mu)-2\sigma(x)]$ the block-lattice Laplacian of BCHH, and the quartic coupling $\lambda=1$ throughout, as in BCHH; the $hC_\chi$ term is the source of model N1 below ($h=0$ for the $\varepsilon$ model). Optionally the sign-free completion of the $\mathbf{10}$-channel plaquette term is added through link fields $(g_{10},g_6)$ on the wedge $|g_6|\le g_{10}$ (Sec.~\ref{sec:channel}). The merged point of BCHH, $\Pc=(y,\kappa)=(2.41,-0.01)$ in our units, is transferred: $\kappa_c=-0.01\pm0.006$ is BCHH's value, and our own lines bracket $\kappa_c$ in $(-0.04,+0.02)$ and $(-0.03,0.05)$ (Sec.~\ref{sec:order}).

\paragraph{Calibration.} The map between the couplings of Eq.~(\ref{eq:action}) and those of BCHH is fixed empirically (claim I.3): on the $\kappa=0.2$ line at $L=8$ the AFM onset is $y_1=2.405^{+0.004}_{-0.003}$ against $\sqrt2\times1.706=2.413$, the exit $y_2\in(3.6,3.9)$ against $\sqrt2\times2.69=3.80$, and the $\kappa$ scan of $y_2$ matches the rescaled window with $\kappa$ unchanged (no scored point misclassified); of the 15 $L=8$ points of the line, $y=\sqrt2\,y_{\rm ref}$ misclassifies none, $y=y_{\rm ref}$ six and $y=y_{\rm ref}/\sqrt2$ eight. Eq.~(1) of BCHH as printed maps onto Eq.~(\ref{eq:action}) with $\kappa$ and $\lambda$ unchanged and $y=y_{\rm ref}/\sqrt2$: the fermion sectors coincide term for term at $y=y_{\rm ref}$ (unit-entry derivative, Pauli-matrix Yukawa), the scalar norm $\tfrac12\mathrm{tr}\,\sigma^\dagger\sigma$ against our $\tfrac12|\sigma|^2$ gives $\sigma\to\sqrt2\,\sigma$, and the normalisation of the $SU(2)$ generators cancels between the Yukawa term and the norm. The data select $y=\sqrt2\,y_{\rm ref}$: a factor $2$ in the Yukawa coupling relative to the derivative that the printed action does not fix. $\kappa_c$ transfers unchanged under either map. The $y$ coordinate of $\Pc$ is located by our own measurements, the $\chi_{\Sigma,\rm stag}$ peak at $\kappa=0.2$ and the $\xi_2/L$ maximum at $y=2.41$ on $6^4$ and $8^4$ along $\kappa=-0.01$ (Sec.~\ref{sec:order}); only $\kappa_c$ is transferred.

\paragraph{Model N1.} Model N1 adds the flavour-blind taste-chiral Majorana mass $hC_\chi$ on one taste doublet (``R'', heavy) of every flavour and reads the other (``L'', light). In the free theory $K-hC_\chi$ has exactly 8 zero singular values in $P_L$ and 8 at $h$ in $P_R$, with $P_L$, $P_R$ the taste projectors of Sec.~\ref{sec:methods}. In Pati--Salam terms the two taste doublets of a flavour are the $(\mathbf4,\mathbf2,\mathbf1)$ and $(\bar{\mathbf4},\mathbf1,\mathbf2)$. $C_\chi$ with the flavour-blind $\delta_{ab}$ pairs $\bar{\mathbf4}\otimes\bar{\mathbf4}$ in the $\overline{\mathbf{10}}$ and is symmetric in taste: a $(\overline{\mathbf{10}},\mathbf1,\mathbf3)$ $\Delta_R$-type Majorana mass of equal size for $u^c$, $d^c$, $e^c$ and $\nu^c$, which breaks $SU(4)_{\rm PS}$ to $SO(4)$ and keeps $SU(2)_L$-taste. The weight stays in the $\tau_2K$ class for every flavour-blind real antisymmetric addition to $K$, so the RHMC needs a single square root and no sign.

\paragraph{Flavour democracy.} The price of positivity in the $\tau_2K$ class is flavour democracy. Operators commuting with an antiunitary $J$ with $J^2=-1$ are quaternion-linear, and the quaternion structure commutes with a hidden $SU(2)$ that acts transitively on the flavour directions, so every Majorana source in the $\tau_2K$ class gives all four flavours the same $|m|$ (row K2.4 of Table~\ref{tab:algebra}). A literal single-$\nu_R$ seesaw (1-of-4) therefore lies outside the $\tau_2K$ class, and no positive representation of it is known. The 2+2 split lies outside the class as well, but it is sign-free in the fermion-bag~\cite{Chandrasekharan:2009wc,Chandrasekharan:2013rpa} and conjugate-pair representations (row K2.6), not in the $\sigma$-triplet RHMC used here: democracy is a property of the $\tau_2K$ construction, not of sign-free simulability. Positivity itself is standard~\cite{Wu:2005zzb,Wei:2016sgb,Li:2016gte,Catterall:2015zua}; democracy is a corollary of that structure.

\section{No light mass term without symmetry breaking}
\label{sec:nomass}

\begin{theorem}[Symmetry zeros]\label{thm:zeros}
Let $hC_\chi$ act on the heavy doublet. A mass-type term is a bilinear whose corner-space string anticommutes with every kinetic $\Gamma_\mu$ (it opens a gap on the massless corner modes); there are sixteen such strings. The mass-type terms of the light doublet are: in the flavour-symmetric $\mathbf{10}$ of $SU(4)$, the three anti-self-dual plane (Majorana, $LL$) masses and the four one-link (Dirac, $LR$) masses; in the flavour sextet, $\Sigma\in\Lambda^2(\mathbf4)$, the on-site string $\varepsilon(x)\chi^a\Sigma_{ab}\chi^b$ (of $\gamma_5\otimes\xi_5$ type in spin$\,\otimes\,$taste language: a pseudoscalar, taste-non-singlet string), the taste-diagonal four-link string and the four three-link ($LR$) strings. The symmetry group of the sourced action is the stabiliser of $C_\chi$ in the group generated by the signed hypercubic maps and the one-site shifts, together with the flavour $SO(4)$ left exact by the flavour-blind source. Every one of these mass terms averages to zero under that group. The flavour-symmetric terms vanish under the hypercubic point group alone, on every box carrying it ($L_t=L$, equal boundary conditions on all axes). The sextet terms vanish under the shifts, and under $SO(4)$, on every box, and not under the point group alone. Hence the light one-point function of each term vanishes at every coupling, on every box carrying the symmetry and in the infinite-volume limit taken through such boxes, unless the corresponding symmetry (point group, shift, or $SO(4)$) breaks spontaneously.
\end{theorem}

\begin{proof}
The symmetry group of the free reduced staggered action is known in closed form~\cite{Golterman:1984cy,Kilcup:1986dg}: it is generated by the one-site shifts $\chi(x)\to\zeta_\mu(x)\chi(x+\hat\mu)$ with $\zeta_\mu(x)=(-1)^{\sum_{\nu>\mu}x_\nu}$, by the hypercubic rotations and axis reversals $\chi(x)\to s_g(x)\chi(gx)$ with their standard sign fields $s_g$ built from the phases $\eta_\mu$ and $\zeta_\mu$, and by $U(1)_\varepsilon$. Every element is a site map with a flavour-blind sign field $s(x)=\pm1$; it leaves the $\varepsilon$ vertex ($s^4=1$), the Yukawa and $\sigma$ terms ($s^2=1$ on every on-site bilinear) and the $\kappa$ term invariant, so the symmetry group of the sourced action is the subgroup fixing $C_\chi$, extended by the flavour $SO(4)$. On periodic boxes the sign fields and the phases $\zeta_S(x)$ of the strings are functions of the coordinates modulo 2; on an antiperiodic axis the sign field of an element reflecting that axis carries in addition the factor $-1$ on the boundary slice $x_\mu=0$, the boundary-crossing sign already contained in the covariant definition of the strings (every one of the 384 point-group sign fields and the four shifts is of this form on the periodic and the all-antiperiodic $4^4$, $6^4$ and $8^4$ boxes). An element therefore acts on a string in the same way on every even-$L$ box that carries the element with the same boundary conditions: it maps the translation-covariant bilinear with phase $\zeta_S$ and displacement $D$ to $\pm$ the bilinear with displacement $gD$, with an $L$-independent sign. The taste projectors $P_L$, $P_R$ are built from $K$ and $C_\chi$ and commute with every element of the stabiliser, so the light part of a string transforms as the string. If an element $g$ of the stabiliser maps a string to its own negative, the average over the coset $\{e,g\}$ vanishes, hence the average over the group, hence the one-point function in every state invariant under $g$. Such an element exists for each class. (a) Flavour-symmetric strings, point group: with $C_\chi$ in the planes $(0,3)$ and $(1,2)$, the anti-self-dual component $C_{(03)}-C_{(12)}$ is reversed by the axis exchange $(0\leftrightarrow1,\,2\leftrightarrow3)$, which exchanges the two planes of $C_\chi$ and fixes it; the components $C_{(01)}-C_{(23)}$ and $C_{(02)}-C_{(13)}$ are reversed by the reversal of both axes of one plane of $C_\chi$ ($x_{1,2}\to-x_{1,2}$ or $x_{0,3}\to-x_{0,3}$), which fixes $C_\chi$; the one-link string along $\mu$ is reversed by the reversal of both axes of the plane of $C_\chi$ that contains $\mu$; the four one-link and the four three-link strings are all reversed by the inversion $x\to-x$. (b) The on-site sextet string: the one-site shift along any axis composed with the reversal of one axis of the other plane of $C_\chi$ fixes $C_\chi$ and has odd total shift; $\varepsilon(x+\hat\mu)=-\varepsilon(x)$ reverses the string while the sign field squares to one on an on-site bilinear. (c) Every sextet string: for $\Sigma=E_{ab}-E_{ba}$ the flavour rotation $\mathrm{diag}(\pm1)$ with $-1$ on $a$ and on one index $c\notin\{a,b\}$ lies in $SO(4)$, leaves the flavour-blind action and source invariant and maps $\Sigma\to-\Sigma$; $\Lambda^2(\mathbf4)$ contains no $SO(4)$ invariant. The elements of (a) exist on every box with $L_t=L$ and equal boundary conditions on all axes, those of (b) and (c) on every box.
\end{proof}

The group averages are computed exactly on $4^4$ as the machine check of the proof (Table~\ref{tab:strings}, claim K5.12; the elements named in (a) and (b) are verified on $4^4$, $6^4$ and $8^4$). There the point group has 384 elements, of which 32 fix $C_\chi$, and each anti-self-dual component is flipped by a 16-element coset of that stabiliser. With the one-site shifts the stabiliser has 8192 elements, of which the 4096 with an odd total shift flip the on-site sextet string.

With $\Sigma$ in the Yukawa triplet $\Gamma_a$ it is $\sum_x\varepsilon(x)\,\tfrac14\mathrm{tr}[\Gamma_aG(x,x)]$, the fermionic staggered-$\sigma$ density conjugate to $\Sigma_{\rm stag}$: the light sextet Majorana mass in the $(\mathbf3,\mathbf1)$ of $SO(4)=SU(2)_\sigma\times SU(2)_{\rm hidden}$ is the antiferromagnetic (AFM) order parameter itself. Inside the AFM phase the light doublet therefore carries a taste-diagonal sextet Majorana mass $\propto\langle\Sigma_{\rm stag}\rangle$, bounded by the measured $|\Sigma_{\rm stag}|$ and $S(\pi)$ (Sec.~\ref{sec:order}); with $h$ on, $S(\pi)$ at $\Pc$ falls to the SMG value by $h=2$ (Sec.~\ref{sec:critical}).

Domain of the box statement: on $L^3\times L_t$ boxes and on boxes whose boundary conditions single out the time axis, every element flipping the $(0,3)$ anti-self-dual component exchanges a spatial axis with time and is absent. On such boxes that one component is an allowed shape-induced one-point function, present in the free theory ($-0.00085$ per flavour on $6^3\times12$ at $h=2$) and read against it; the other fifteen strings remain zeros there. The infinite-volume limit through boxes carrying the full point group inherits all sixteen zeros.

\begin{table*}[!t]
\caption{The sixteen mass-type strings of the light doublet (the corner-space strings anticommuting with every kinetic $\Gamma_\mu$) and the symmetry that kills each: group averages of the light part on $4^4$, $h=0.5$, under the 32-element point-group stabiliser of $C_\chi$ (pg) and the 8192-element stabiliser with one-site shifts (shift); $SO(4)$ by representation theory ($\mathbf4\otimes\mathbf4$ has one invariant, $\delta_{ab}$; $\Lambda^2(\mathbf4)$ none). The self-dual rows are the heavy ($RR$) directions, averaged as whole strings: $C_\chi$ itself is the source and is invariant; its light block on each box is given in Table~\ref{tab:lightblock}. On the periodic $4^4$ box 240 of the 256 modes are zone-boundary modes.}
\label{tab:strings}
\begin{ruledtabular}
\begin{tabular}{lccccl}
strings & fl. & taste & pg & shift & killed by\\
\colrule
3 anti-self-dual planes & $\mathbf{10}$ & $LL$ & $0$ & $0$ & point group\\
4 one-link & $\mathbf{10}$ & $LR$ & $0$ & $0$ & point group\\
self-dual $(03)+(12)=C_\chi$ & $\mathbf{10}$ & $RR$ & $1$ & $1$ & allowed (the source)\\
2 other self-dual planes & $\mathbf{10}$ & $RR$ & $0$ & $0$ & point group\\
on-site $\varepsilon\chi\Sigma\chi$ & $\mathbf6$ & $LL{+}RR$ & $1.000$ & $0$ & shift; $\Z_4$; $SO(4)$\\
four-link $((0,2),(0,1,2,3))$ & $\mathbf6$ & $LL{+}RR$ & $0.354^{\rm a}$ & $0$ & shift; $SO(4)$\\
4 three-link & $\mathbf6$ & $LR$ & $0$ & $0$ & point group; $SO(4)$\\
\multicolumn{6}{l}{$^{\rm a}$ all-antiperiodic box; $1.000$ on the periodic $4^4$ box, whose light sector is its 16 corner modes.}\\
\end{tabular}
\end{ruledtabular}
\end{table*}

\begin{proposition}[No seesaw without SSB]\label{prop:noseesaw}
Let $K_L$ be the free light kinetic operator, $G_M$ the symmetry group left exact by the explicit perturbation (acting linearly on the fields and commuting with $K_L$), and $\mathcal M_L$ the space of mass-type light bilinears, those whose low-energy restriction anticommutes with $K_L$. Assume (A1) $\mathcal M_L$ contains no $G_M$-invariant element; (A2) the state is $G_M$-invariant (a unique ground state, or a finite-volume Gibbs state with $G_M$-symmetric boundary conditions and no symmetry-breaking field, the infinite-volume limit taken afterwards). Then (i) the mass-type projection of every light two-point function vanishes, and likewise for the light 1PI function and for the pairing of any operator carrying the light quantum numbers with itself or with $\psi_L$; (ii) the mass-type projection of the seesaw block $\Delta M_P^{-1}\Delta^T$ of any partner matrix with light block $K_L$, mixing $\Delta$, partner block $K_P+M_P$ and $G_M$-invariant couplings vanishes (the block may carry $G_M$-invariant $K$-type terms); (iii) if in addition (A3) any light gap is a gap of the full $G_M$-symmetric spectrum in the light channel, and (A4) $G_M$ contains the lattice rotations and reflections, so that every $G_M$-invariant element of the light corner algebra analytic at $p=0$ is mass-type or odd in $p$, then a gapped light channel has $G_{LL}(0)=0$: a Green's-function zero, never a mass-type pole.
\end{proposition}

\begin{proof}
(i) For $g\in G_M$ with representation $U_g$, (A2) gives $\langle\psi_L\psi_L^T\rangle=U_g\langle\psi_L\psi_L^T\rangle U_g^T$ (with complex conjugation for antiunitary $g$), so the two-point function is $G_M$-invariant; the projection onto $\mathcal M_L$ commutes with the $G_M$ action because $G_M$ commutes with $K_L$, so it is a $G_M$-invariant element of $\mathcal M_L$, zero by (A1). Inversion and Dyson resummation preserve invariance, which covers the 1PI function; for an operator in the light representation the argument is identical. (ii) is (i) applied to the effective light action after integrating out $P$, a $G_M$-invariant operation. (iii) With (A3) the light two-point function is analytic at $p=0$. For real Grassmann fields $G(p)^T=-G(-p)$, so the even-in-$p$ part of $G_L$ is an antisymmetric matrix on the light corner space tensored with flavour, and its value at $p=0$ is a $G_M$-invariant antisymmetric element of that algebra; by (A4) it is mass-type, by (i) it vanishes, and the remainder is odd in $p$. On the lattice (A4) is a corner-space statement under a finite group. In N1 it holds with $G_M$ the stabiliser of $C_\chi$ with the one-site shifts, times $SO(4)$, and not under the point group alone; the statement is the machine-checked enumeration of all 256 corner strings as lattice bilinears on $4^4$ (claim K5.14). Under the 32-element point-group stabiliser alone 20 light invariants with antisymmetric kernel, even in $p$ and not mass-type, survive on the periodic $4^4$ box and 36 on the all-antiperiodic one; under its 64-element extension by the $\Z_4$-composed $C_\chi$-flipping elements twelve remain on the periodic box (four taste-diagonal self-dual tensor strings of $\Gamma_\mu\Gamma_\nu$ type and eight $LR$ strings), so (A4) fails without the shifts. Under the 8192-element stabiliser with shifts all twelve vanish, each flipped by half the group, and the light invariants that remain are the identity and the taste-chirality string $X_{13}Z_{0123}$, both with symmetric kernel (excluded from the even-in-$p$ part of the flavour-symmetric channel by $G(p)^T=-G(-p)$ and from the sextet by $SO(4)$), together with eight invariants with antisymmetric kernel, the four kinetic strings $\Gamma_\mu$ and four three-link strings, all odd in $p$ and so allowed by (A4). The mass-type strings themselves are removed by Theorem~\ref{thm:zeros}, which alone does not give (A4): it concerns only the strings anticommuting with every $\Gamma_\mu$.
\end{proof}

The proposition has two instances here. In N1, (A1) is Theorem~\ref{thm:zeros} with $G_M$ the lattice stabiliser with shifts times $SO(4)$; at $h=0$ it is the charge-2 statement of Sec.~\ref{sec:intro}, every Majorana bilinear carrying $\Z_4$ charge $2$ (row K1.7 of Table~\ref{tab:algebra}). The second instance is a $0{+}1$-dimensional model with 8 light and 8 heavy Majoranas, Fidkowski--Kitaev quartics~\cite{Fidkowski:2009dba,Fidkowski:2010jmn,Tong:2019bbk} and a heavy mass $M$, solved by exact diagonalisation: the light mass-type Green's function is $\le5.1\times10^{-15}$ for every $M\le100$ and coupling while the residual antiunitary symmetry $T'$ is exact, and $\ge0.16$ as soon as one $T'$-odd bilinear is added. The pre-registered verdict of that check is void by its letter (its sabotage threshold assumed an $O(U)$ light gap; the Fidkowski--Kitaev gap is $14U$); the check is cited for these exact zeros only.

What the proposition does not fix is the size or the $M$-dependence of the symmetric gap. In the same $0{+}1$D model a Higgs-like light--heavy pairing gives the seesaw slope $d\ln\Delta_L/d\ln M=-0.881$ at $M\le100$ (asymptotically $-1$). A light half gapped only through virtual heavy excitations also has a gap scaling as $1/M$ (slopes $-1.09$, $-1.13$), but with mass-type part zero. The two are told apart by the mass-type part of the propagator at $p\to0$, pole versus zero, not by the size of the gap. Whether the gapped symmetric option exists at all is decided by the 't~Hooft anomaly of $G_M$ on the light representation~\cite{Catterall:2022jky}; in N1 the residual symmetry acts on the 8 light Weyl fermions as a real representation, which excludes a perturbative anomaly. The exact residual group contains order-4 elements, $C_\chi$-flipping point-group elements composed with the $\Z_4$ generator $\chi\to i^{\varepsilon(x)}\chi$, that act on the 8 light zero modes of the periodic $4^4$ box with eigenvalues $\pm i$ only (20 of the 32 flipping elements; charges $1,1,1,1,3,3,3,3$, with $U^2=-1$ on the light doublet): a Spin-$\Z_4$-type action under which every light mode carries odd charge. Its discrete anomaly is not computed here.

Remark: for the seesaw corner of Eq.~(\ref{eq:partner}) the light eigen-gap $\lambda=\sqrt{\Delta^2+M_p^2/4}-|M_p|/2\le\Delta$ decreases strictly with $|M_p|$ (eigenvalues of a symmetric $2\times2$ matrix). Its premise, a tree-level Dirac mixing of the light field with a Majorana-massive partner, is absent in N1, so the remark bounds the seesaw corner only.

\section{Methods}
\label{sec:methods}

\paragraph{Instruments.} Every number comes from the rational hybrid Monte Carlo (RHMC) of the companion package and from exact (dense) or stochastic fermion measurements on stored configurations. Each instrument is certified against an independent exact computation with a sabotage control (rows I.1--I.11 of the claim map, Appendix~\ref{app:claims}): the RHMC and its Hasenbusch split against an exact-determinant Metropolis (largest pull $1.36\sigma$; sabotage rejected at $6.9\sigma$); the link-field HMC against the all-orders fermion-bag expansion~\cite{Chandrasekharan:2009wc,Chandrasekharan:2013rpa} ($Z$ to $9.4\times10^{-15}$); the $\mathbf{10}/\mathbf{6}$-channel estimator; the pairing sources and the Pfaffian sign by three routes; the N1 RHMC against exact-determinant Metropolis on all-antiperiodic $4^4$ (20 observables, largest pull $2.08\sigma$; the phase-flipped pattern rejected at $107\sigma$) and the taste observables.

\paragraph{Observables.} Table~\ref{tab:obs} defines the observables used in Secs.~\ref{sec:order}--\ref{sec:halfmass}. $P_L$ and $P_R$ are the taste projectors onto the light and heavy doublets, and $G_L$ is the light-projected fermion propagator. ``Free'' always means the same observable of the free theory on the same $L$, $L_t$, boundary condition and $h$, and every interacting light-sector number is quoted as a ratio to it. $C_L(t)$ has no plateau at $L_t\le12$ (Fig.~\ref{fig:noplateau}): $m$ is an effective mass at the stated $t$, never a single-fermion mass. The single-fermion readouts are $g$ and $\alpha_L$, at $p_{\min}$ only.

\begin{table}[!htb]
\caption{Observables. Every light-sector quantity is read as a ratio to the free theory on the same $L$, $L_t$, boundary condition and $h$.}
\label{tab:obs}
\begin{ruledtabular}
\begin{tabular}{p{0.17\columnwidth}p{0.40\columnwidth}p{0.35\columnwidth}}
symbol & definition & reads\\
\colrule
$C_L(t)$ & $\sum_{\vec x}G_L^2$, zero-momentum light pair channel & positive two-fermion correlator; no plateau at $L_t\le12$\\
$m$ & cosh effective mass of $C_L$ at the midpoint $t=L_t/2-1$ & distance from criticality; not a single-fermion mass\\
$g$ & $G_L(p_{\min})/{\rm free}$ at the smallest twisted momentum & single-fermion readout: $1$ free massless, small when gapped\\
$c$ & $\chi_L/{\rm free}$, $\chi_L=V\langle\phi_L^2\rangle$ & light pair susceptibility\\
$e$ & $d\ln\chi_L/d\ln V$ on a volume pair & free $0.39$--$0.45$ on $(6^4,8^4)$, $0.12$--$0.14$ on $(4^4,8^4)$; $1$ at the SSB scale\\
$R_{\rm peak}$ & $[\chi_L(0)/\chi_L(p_{\min})]/{\rm free}$ & $p=0$ peak of the light pair susceptibility\\
$\phi_R$, $\phi_{T_R}$ & elementary and composite amplitudes in the $C_\chi$ direction & composite takeover, $\rho=|\phi_{T_R}/\phi_R|$\\
$S(\pi)$, $\xi_{2,\rm stag}$ & $V\langle m_{\rm stag}^2\rangle$; second-moment length of the staggered $\sigma$ & the AFM (critical) $\sigma$ channel\\
$O_4$ & $\varepsilon$-vertex density & weight of the four-fermion vertex\\
$\alpha_L$ & odd-part frequency exponent of the light corner block & $-1$ for a pole at $\vec p=0$, $\to+1$ for a Luttinger zero\\
\end{tabular}
\end{ruledtabular}
\end{table}

\paragraph{Taste coverage and the light block of the source.} The light/heavy split exists only at momenta with every $\cos p_\mu\neq0$, which on all-antiperiodic boxes fails iff $L\equiv2\bmod4$: it exists on every momentum of $4^4$, $4^3\times8$, $8^4$ and $12^4$, on $(4/6)^4=19.8\,\%$ of the momenta of $6^4$ (1040 of 1296 modes are zone-boundary modes) and on $29.6\,\%$ of those of $6^3\times12$. On $L=6$ the surviving momenta are the four per axis nearest a corner, so every light-sector number on $6^4$ is built from the IR-most fifth of the fermion modes, on $8^4$ from all of them. A second, independent classification is whether the source has a light block. In the corner block of each momentum $C_\chi=\cos p_0\cos p_3\,A+\cos p_1\cos p_2\,B$ with constant corner strings $A$, $B$, and $\|P_LC_\chi P_Lv_p\|=0.354\,h\,\bigl||\cos p_0\cos p_3|-|\cos p_1\cos p_2|\bigr|$ on every momentum of every box (claim K5.15): the light block vanishes wherever all $|\cos p_\mu|$ are equal. Table~\ref{tab:lightblock} gives its weight per box. On $4^4$ and $6^4$ the light doublet is source-free on every momentum; on $8^4$ the source dresses 2560 of the 4096 light momenta, with $14\,\%$ of its Frobenius weight in the light doublet, and on $4^3\times8$ and $6^3\times12$ every light momentum. This is why the free light pair-channel mass depends on $h$ on $8^4$ (Table~\ref{tab:pc}: $m_{\rm free}$ from $1.519$ to $1.731$ for $h=0\to3$) and not on $6^4$ ($0.962$ at every $h$). The free normalisation removes the free counting of a box's mode set, not the interacting weight of the modes it lacks. Every $(6^4,8^4)$ light-sector ratio carries the mode-set systematic, marked $\dagger$ in the tables; the $(4^4,8^4)$ pair has the split on every momentum of both boxes but compares a source-free light doublet with a partly dressed one. The like-for-like comparisons are those of Sec.~\ref{sec:smg}, with the $8^4$ light projector restricted to its source-free or to its IR-most momenta.

\begin{table}[!htb]
\caption{Taste split and light block of the source per all-antiperiodic box (claims K5.13, K5.15): fraction of the momenta on which the light/heavy split exists; light block $\|P_LC_\chi P_L\|/\|C_\chi\|$, the Frobenius fraction of the source in the light doublet; number of light momenta on which that block vanishes (source-free).}
\label{tab:lightblock}
\begin{ruledtabular}
\begin{tabular}{lccc}
box & split on & light block & source-free light momenta\\
\colrule
$4^4$ & $100\,\%$ & $0$ & 256 of 256\\
$6^4$ & $19.8\,\%$ & $0$ & 256 of 256\\
$8^4$ & $100\,\%$ & $0.368$ & 1536 of 4096\\
$4^3\times8$ & $100\,\%$ & $0.195$ & 0 of 512\\
$6^3\times12$ & $29.6\,\%$ & $0.173$ & 0 of 768\\
\end{tabular}
\end{ruledtabular}
\end{table}

\begin{figure}[!htb]
\includegraphics[width=\columnwidth]{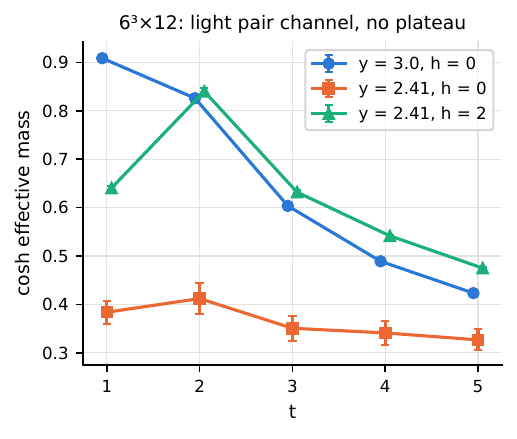}
\caption{Cosh effective mass of the light pair channel $C_L(t)$ on $6^3\times12$: no plateau in the SMG phase ($y=3.0$, falling at every step) or at $(\Pc,h{=}2)$; 14 of 15 single-cosh windows at $(3.0,0)$ have $\chi^2/\mathrm{dof}>3$. The $(\Pc,h{=}0)$ curve is flat from $t=3$ within $0.7\sigma$ ($0.351\pm0.026$, $0.341\pm0.025$, $0.327\pm0.022$ at $t=3$--$5$; the $t=2$ value $0.412\pm0.031$ lies $2.2\sigma$ above $t=5$). Errors: blocked jackknife with the error-model inflation.}
\label{fig:noplateau}
\end{figure}

\paragraph{Error model and samples.} Errors are blocked jackknives, $\sigma=\max\{\sigma_{\rm block},\sigma_{\rm naive}\sqrt{2\tau_{\rm eff}}\}$ with $\tau_{\rm eff}=\max(\tau_{\rm int},\tau_B/\Delta)$, $\tau_B$ the Madras--Sokal time of the per-trajectory bosonic series and $\Delta$ the measurement spacing; pooled rows carry each member's inflation. Dense fermion measurements on $8^4$ are taken every 20 trajectories (45 per 1000-trajectory chain, 100 per extension); the $8^4$ $\Pc$, $h=0$ chain (45 measurements, acceptance $0.88$, $\tau_B=6.3$) is the denominator of every $8^4$ ratio. Replica consistency on independent seeds gives $p\ge0.16$ at $8^4$ and $\ge0.07$ at $6^4$. Every volume comparison is a $(6^4,8^4)$ pair except the light susceptibility, which has the $(4^4,8^4)$ pair at $y=3.0$ and at $\Pc$ and the box $4^3\times8$ at $y=3.0$ (Secs.~\ref{sec:smg}, \ref{sec:critical}); every quoted ``$\gamma/\nu_{\rm eff}$'', ``$e$'' or ``$d\ln\chi/d\ln V$'' is a two-volume finite-size ratio, not a scaling dimension.

\paragraph{Protocol.} Each verdict has a written criterion, recorded in the ledger (Table~\ref{tab:prereg}) with the claim that applies it. The first-round criteria of stages 0 and 1, of K8 tier 1 and of the K7 controls proved mis-specified or were changed after a first look, as the table states row by row; the N1 outcomes (SMG class (a) by the ledger, K8 tier 2, the K8.3 onset, the direction tests) are second-round criteria, fixed before the trajectories they were applied to but after the first-round data at the same points had been read. The stage-1c replication on independent seeds reproduces the stage-1b numbers and class (a); it does not test the criterion, whose SSB threshold was fixed after the stage-1 row had been read. Criteria are applied as written, including the mis-specified ones. The outcomes reported in Secs.~\ref{sec:order}--\ref{sec:halfmass} are readings of the data after these tests, most of which returned no verdict, MIXED or VOID. The pre-registered criteria that returned a verdict as written in the direction of a conclusion of this paper are K3.1 and K3.3 (no $(\mathbf{10},\mathbf3,\mathbf1)$ condensate; the pure $\mathbf{10}$ term irrelevant), K4.1--K4.3 (composite takeover: SMG $\rho>4$ with $\phi_R/h\le0.05\times$free, $\Pc$ $\rho<0.2$), M.1 (the moonshot trigger not met, the fall alert silent), K8.2 T1, T4 and T5 (the light gap rises monotonically with $h$; no SSB at $\Pc$ in the $(0,3)$ channel of the clauses; the volume check), K8.5 P3 ($(3.0,3)$ deep SMG) and the $8^4$ rule of K7.2 (continuous, consistent with). The SSB clauses at $y=3.0$, the onset exponent, the direction tests and the single-power tests returned no verdict or a mixed one; the drift toward free, the crossover reading and the large-$N$ comparison are readings.

\begin{table*}[!t]
\caption{Verdict ledger. ``Claim'' is the companion-repository check that applies the criterion; the last column records every criterion applied literally, with each mis-specification disclosed.}
\label{tab:prereg}
\begin{ruledtabular}
\begin{tabular}{p{0.11\textwidth}p{0.10\textwidth}p{0.65\textwidth}}
verdict & claim & as written, and what proved mis-specified\\
\colrule
K3 null, K3.3, K4, M & K3.1, K3.3, K4.1--K4.3, M.1 & Applied as written ($\chi_{10}$ (6,8) ratio $+2\sigma<0.5$ with $\xi_{10}/L$ not rising; $q=R(8)/R(6)\le1+2\sigma_q$, $q+2\sigma_q<1.78$; SMG $\rho>4$ with $\phi_R/h\le0.05\times$free, $\Pc$ $\rho<0.2$; moonshot trigger and fall alert). The $8^4$ chain at $(2.41,0.02,0)$ dropped its first 600 trajectories (acceptance 0.04--0.32) before any readout.\\
K5 stage 0 gates & K5.3 & G1 pass; G2 (SYM control) fails at $h=2$ on $r_L=1.020\pm0.006$ ($+3.35\sigma$, a 2\,\% shift); G3 presumed a free-like heavy response and is ill-posed at $h=0.5$ and at $y=3.0$. Not a verdict.\\
K5 stage 1 & K5.5 & NO VERDICT at $L\le8$: the gap floor was placed on the $t^*$ log-ratio proxy ($\approx m^2/2$, not the gap); the SSB threshold of clause (b), $e-2\sigma>0.5$, sat $0.1$ above the free growth ratio; the free normalisation of $r_L$ is void on $6^3\times12$. Clause (b) fired by its letter at $(3.0,h{=}2)$: $e=0.786\pm0.129$, $e-2\sigma=0.527$ on $(6^4,8^4)$.\\
K5 stage 1b & K5.8 & Midpoint cosh mass as estimator; thresholds $0.7$, $0.1$; SSB clause re-specified relative to free, (b1) $e-e_{\rm free}-2\sigma>0.5$ with (b2), (b3) on $R_{\rm peak}$ and $\chi_L/$free; fixed before the new trajectories but after the stage-1 row that fired (b) had been read: a replication on extended chains. Outcome (a); partial outcome ``growth relative to free'' recorded. S7$'$ flagged $3\sigma$ per-mille drifts on two chains.\\
K5 stage 1c & K5.11, K8.3 & Independent $8^4$ seeds (two per $h$); replicates (a) and its numbers, not the criterion; the (6,8) growth clause compares unlike taste coverage (Sec.~\ref{sec:smg}). The onset bands $[1.6,2.4]$ (quadratic) and $[0.7,1.3]$ (linear) for $p$ in $m^2-m_0^2=(ch^p)^2$ were transplanted from the $8^4$ regime $m\gg m_0$; on $6^4$, where $\delta m\ll m_0$, they test $\delta m\propto h^4$ against $h^2$; the quadratic control also returns UNDECIDED. Restated in Sec.~\ref{sec:critical}.\\
K6 readout & K6.2, K6.3 & Label-swap control fails as written wherever SMG screens chirality (every $y=3.0$ row). K6 is an instrument.\\
K7 order & K7.1, K7.2 & After the controls were run: the injected-jump control had the wrong sign (fixed); the bimodality bandwidths over-smoothed the $4\sigma$ control mixture ($5\sigma$ mixture substituted); H3 has no power on this grid. The $8^4$ rule first returned UNDECIDED at 1000 trajectories per start; the starts were extended to 2000 and the unchanged rule applied.\\
K8 tier 1 & K8.1 & Two $\sigma$-channel thresholds widened after a first look; disclosed. The $h$-scaling was registered separately.\\
K8 tier 2 & K8.2 & T1 PASS; T2: no acceptable single power on any lattice (by the letter FAIL on $6^3\times12$, UNDECIDED on $6^4$ and $8^4$); T3 UNDECIDED; T4 PASS (reads the $(0,3)$ channel; every channel in Sec.~\ref{sec:smg}); T5 PASS (T4 and T5 are $(6^4,8^4)$ light-sector ratios, $\dagger$); T3's pass clause lacks a resolution requirement (affects no outcome); the $8^4$ $h=3$ chain lies outside the letter.\\
K8.5 direction & K8.5 & P1 ($6^4$, three clauses): MIXED by the letter, the $S(\pi)$ clause at $1.9\sigma$ and direction-blind by construction; P3 ($(3.0,3)$ critical-like?): DEEP SMG.\\
K8.6 large-$N$ fit & K8.6, K8.8 & Post hoc (shape parameter-free, amplitude and offset fitted); weighted accordingly (Sec.~\ref{sec:critical}).\\
K5.10 $0{+}1$D & K5.10 & Verdict VOID: the zero-sabotage threshold assumed an $O(U)$ light gap (Fidkowski--Kitaev gap $14U$). Cited for its exact zeros only.\\
K8.9 $8^4$ direction & K8.9 & P1 clauses unchanged on $8^4$ rows, fixed before the $8^4$ pair was run; MIXED as on $6^4$: $g$, $c$ toward SYM, $S(\pi)$ clause $1.5\sigma$ (Sec.~\ref{sec:critical}).\\
\end{tabular}
\end{ruledtabular}
\end{table*}

\section{No first-order signature at \texorpdfstring{$\Pc$}{Pc} for \texorpdfstring{$L\le8$}{L<=8}}
\label{sec:order}

BCHH read the merged transition as continuous from susceptibility powers on $8^4$; earlier work in the same line leaned the same way~\cite{Butt:2018nkn,Catterall:2016dzf,Schaich:2017czc,Ayyar:2016lxq,Catterall:2017ogi,Butt:2024kxi}. We add the explicit first-order test (Table~\ref{tab:order}, Fig.~\ref{fig:xi}): 44 chains of the $\varepsilon$ model at $L=6$ and $8$ around $\Pc$ (the 39-chain pilot set of claim K7.1 and the five $8^4$ chains of K7.2), including antiferromagnetic, cold and hot starts at $\Pc$ on both volumes, 2000 trajectories per $8^4$ start. The pre-registered rule returns ``continuous (consistent with, $L\le8$)'', and three sabotaged copies (an ordered-start memory, a two-state telegraph signal, a late drift) never return it. The afm-start $8^4$ chain carries a slow mode: $m^2=(5.69\pm0.77)\times10^{-3}$ over trajectories 100--1000 against $(8.38\pm0.73)\times10^{-3}$ over 1000--2000 ($2.5\sigma$; 200-trajectory windows between $4.0$ and $10.4\times10^{-3}$), inside the rule's $3\sigma$ stationarity threshold and read as an excursion that ended, not an ordered-start memory; the injected ordered-start memory of the sabotage is flagged at $7.4\sigma$ by the same rule. Two volumes cannot establish continuity; a weak first-order transition with a correlation length far above 8 lattice spacings is not excluded; and the $\xi_2/L$ crossing is a $+2.4\sigma$ statement on two volumes (the jackknife of claim K7.2; the rounded pooled values of Table~\ref{tab:order} give $2.3\sigma$). Walking versus a power law is not decidable at $L\le12$: with $\xi\le L_{\max}/3$ and a free $y_c$, the expected $\Delta\chi^2$ between the two scenarios~\cite{Gorbenko:2018ncu,Kaplan:2009kr} is $\le1.2$ for $L_{\max}\le12$ even with seven points at 1\,\%.

\paragraph{The transferred $\kappa_c$.} $\kappa_c=-0.01\pm0.006$ is BCHH's value. BCHH state that their ``simulations suffer from extremely long autocorrelation times near $P_c$'' and that it ``has been very difficult to extract reliable results there''~\cite{Butt:2026lkt}; they read $\Pc$ as a multicritical point and conjecture, without confirmation, a description by condensed Hopf defects. The transferred point carries that uncertainty. Our own lines bracket $\kappa_c$ in $(-0.04,+0.02)$. At $\kappa=-0.04$ no $y$ shows order: on the scored chains $|\Sigma_{\rm stag}|_6/|\Sigma_{\rm stag}|_8=1.75$ at $y=2.028$ and $1.78$ at $y=2.452$, against the pure-fluctuation value $\sqrt{V_8/V_6}=1.78$, and $\xi_2/L$ falls with $L$ at every $y$ of the line; the $6^4$ chains of this line at $y=2.283$, $2.368$ and $2.41$ are unscored by the rule of claim K7.1 (non-stationary or $\tau_{\rm int}>N/50$, low acceptance) and enter no number here. At $\kappa=+0.02$ the scored $8^4$ chain at $y=2.452$ has $R_4=1.085\pm0.008$ (AFM $1.09$); the $6^4$ member of this pair is unscored ($\tau_{\rm int}=100>N/50$, non-stationary, low acceptance), so the ratio $1.01$ and $\gamma/\nu_{\rm eff}=3.89\pm0.51$ formed with it carry no weight, and the line is AFM-edged at $y=2.41$ on scored chains (ratio $1.30$, $\Delta(\xi_2/L)=+3.6\sigma$). The upper edge of the bracket $(-0.04,+0.02)$ rests on the scored $8^4$ $R_4$ at $y=2.452$ and on the $y=2.41$ pair. The $\kappa$ scan of the calibration (claim I.3) brackets it in $(-0.03,0.05)$. At $(2.41,-0.01)$ itself the $L\le8$ indicators exclude developed AFM order: $|\Sigma_{\rm stag}|_6/|\Sigma_{\rm stag}|_8=1.37$ unpooled and $1.42$ pooled (ordered $1.01$, disordered $1.78$), $\gamma/\nu_{\rm eff}=1.58\pm0.16$ (AFM $3.9$), $R_4=1.42$--$1.51$ (AFM $1.09$), $\xi_2/L$ growth $1.07\pm0.03$ (ordered $\ge1.33$), and no start dependence. A thin AFM sliver at $(2.41,-0.01)$, ordered only beyond $L=8$, is not excluded: along $\kappa=-0.01$, $\Delta(\xi_2/L)=(L{=}8)-(L{=}6)$ is $\ge0$ over $y\in[2.368,2.452]$ ($+0.016\pm0.015$ at $2.368$; $+0.028\pm0.042$ at $2.452$, where the $8^4$ chain is unscored), which is consistent with a finite-size-broadened crossing and with a thin ordered band alike. Every result ``at $\Pc$'' is a result at the transferred point. Inside a sliver the $h=0$ anchor would be a weakly ordered state at $L\le8$, and the light doublet would carry the sextet Majorana mass of Theorem~\ref{thm:zeros}, bounded by $|\Sigma_{\rm stag}|_8=0.082$. The fall of $S(\pi)$ with $h$ and the moves of $g$, $c$ toward free (Sec.~\ref{sec:critical}) are measured at the transferred point; the $y=3.0$ results do not depend on $\Pc$.

\begin{table*}[!t]
\caption{First-order signatures at $\Pc=(2.41,-0.01)$. First 100 trajectories cut; blocked jackknife with block $2\tau_{\rm int}(m^2)$; $\tau_{\rm int}(m^2)=10$--$17$ at $8^4$ (110--180 independent samples per start).}
\label{tab:order}
\begin{ruledtabular}
\begin{tabular}{ll}
signal & result\\
\colrule
start dependence: 15 pulls per volume on $\{m^2,\sigma^2,O_4,S(\pi),S(\pi{+}p_{\min})\}$ & largest $2.66\sigma$ ($6^4$), $2.42\sigma$ ($8^4$)\\
two-peak histograms, 132 series of 44 chains & none\\
energy-cumulant deficit ratio $(2/3-V_e)_8/(2/3-V_e)_6$ & $0.28\pm0.02$ ($1/V$ law $0.32$; latent heat: constant)\\
$\gamma/\nu_{\rm eff}$ from $S(\pi)$ on (6,8), pooled replicas & $1.58\pm0.16$, $16\sigma$ below the first-order value 4\\
$\xi_{2,\rm stag}/L$ at $y=2.41$, $L=6$ / $8$, pooled & $0.373\pm0.007$ / $0.399\pm0.009$ ($+2.4\sigma$)\\
same, unpooled & $0.377\pm0.010$ / $0.415\pm0.019$ ($+1.8\sigma$)\\
growth of $\xi_2/L$ at $\Pc$, $8^4/6^4$ & $1.07\pm0.03$ (first-order peak $\propto L$: $1.33$)\\
Binder $\Delta R_4$ at $y=2.41$ & $+0.052\pm0.035$\\
\end{tabular}
\end{ruledtabular}
\end{table*}

\begin{figure*}[!t]
\includegraphics[width=\textwidth]{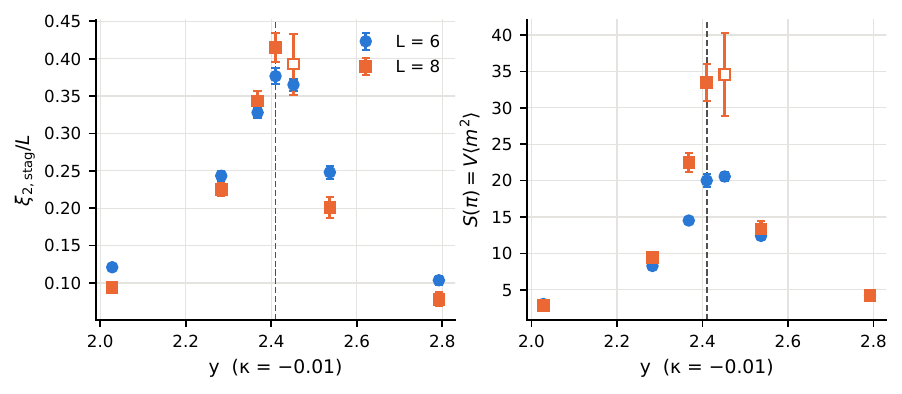}
\caption{$\xi_{2,\rm stag}/L$ (left) and $S(\pi)$ (right) along $\kappa=-0.01$ on $6^4$ and $8^4$; the dashed line is $y=2.41$. Open symbols: the unscored $8^4$ point at $y=2.452$ ($\tau_{\rm int}>N/50$). $\Delta(\xi_2/L)$, $8^4$ minus $6^4$, is $+2.4\sigma$ at $y=2.41$, $+0.016\pm0.015$ ($+1.0\sigma$) at $y=2.368$ and negative ($-1.7\sigma$ to $-3.3\sigma$) at the four outer $y$: the $(6,8)$ crossing lies in $[2.368,2.452]$, slightly above $2.41$. Errors: blocked jackknives.}
\label{fig:xi}
\end{figure*}

\section{The \texorpdfstring{$\Z_4$}{Z4}-odd channel and the composite takeover}
\label{sec:channel}

\paragraph{The elementary $\mathbf{126}$ channel does not condense.} The two-site $\mathrm{Sym}^2(\mathbf4)$ bilinear, the $(\mathbf{10},\mathbf3,\mathbf1)\subset\mathbf{126}$ channel of the lattice model, is the only $\Z_4$-odd channel beside the on-site one. On every sign-free model that carries its interaction (Table~\ref{tab:k3}) the complete $SU(4)$-invariant correlator $\chi_{10}$ does not grow with the volume up to $L=8$, and at $\Pc$ it is free-like and suppressed. The finite-$g$ response to the pure plaquette term does not grow from $L=6$ to 8 either, while the columnar dimer channel of the positive control grows on the same pipeline. A flavour-blind same-parity pairing source (a nodal pairing field, not a Majorana mass) finds a finite, volume-independent susceptibility at $0.3$--$0.4\times$ free and no spontaneous breaking. Where the $U(1)_\varepsilon$-exact family orders, it is a bond crystal, not the $\mathbf{126}$. All of these are finite-size statements with (4,6) and (6,8) ratios only. Earlier four-dimensional studies~\cite{Ayyar:2016lxq,Catterall:2016dzf,Butt:2026lkt} measured four-fermion condensates and found bilinears to vanish; the two-site channel had not been examined.

\begin{table*}[!t]
\caption{The $\Z_4$-odd elementary channel at $L\le8$ (blocked jackknife; $8^4$ from the stochastic estimator with 8 noise vectors). Ratios are $d\ln X/d\ln V$ on (6,8).}
\label{tab:k3}
\begin{ruledtabular}
\begin{tabular}{p{0.45\textwidth}p{0.38\textwidth}l}
quantity & value & claim\\
\colrule
$\chi_{10}$ (6,8) ratio at five wedge points & $-0.12\ldots-0.03$, $+2\sigma\le0.10$ & K3.1\\
positive control, dimer $V\langle D_\parallel^2\rangle$ at $(2.41,0.1,+0.1)$, (4,6) / (6,8) & $0.71\pm0.01$ / $0.80\pm0.01$ & K3.1\\
$(\chi_{10}+2\sigma)/\chi_{10}^{\rm free}$ at $\Pc$, largest & $0.54$ (free $1.6766$, $1.0875$, $0.9365$ at $L=4,6,8$) & K3.2\\
response to the pure $\mathbf{10}$ term, $R(8)/R(6)$ & $1.022\pm0.045$ ($g=0.01$), $0.973\pm0.023$ ($g=0.02$) & K3.3\\
pairing source: (6,8) ratio of $\langle\Phi_{\rm src}\rangle/h$ & $-0.055\pm0.055$, $-0.030\pm0.048$ & K3.4\\
$-$edge at $g_j=0.2$: bond-crystal (6,8) ratio & $2.33\pm0.12$, $2.17\pm0.10$ (crossover) & K3.5\\
$+$edge $(2.41,0.1)$: $V\langle D_\parallel^2\rangle$ at $L=4,6,8$ & $5.11\pm0.07$, $15.65\pm0.11$, $39.19\pm0.16$ & K3.6\\
same point, independent replica chains (exponents $0.71$ / $0.80$) & $4.97\pm0.10$, $15.72\pm0.11$, $39.28\pm0.17$ & K3.1\\
\end{tabular}
\end{ruledtabular}
\end{table*}

\paragraph{Composite takeover.} Across SYM $\to\Pc\to$ SMG the response to a $\Z_4$-odd field moves from the elementary fermion to its $\varepsilon$-vertex partner $\Psi^a=\varepsilon_{abcd}\chi^b\chi^c\chi^d$ (Fig.~\ref{fig:takeover}). Two sources are used. With a nodal source on one flavour ($4^4$, $6^4$, production boundary conditions) the composite-to-elementary ratio is $\le0.0055$ in SYM and $\ge5.8$ in SMG, where the composite responds linearly, $\phi_T/h=-0.0877$ ($4^4$) and $-0.0896$ ($6^4$). With $C_\chi$ on all-antiperiodic $8^4$, $6^4$ and $6^3\times12$ (the stage-1 rows, tabulated in Appendix~\ref{app:rows}) the SMG elementary amplitude is $\phi_R/h\le0.011\times$ free, while the composite is $\phi_{T_R}/h\approx-0.00195$, $h$-independent within 0.9\,\% and lattice-independent within 1.9\,\%; the ratio $\rho=|\phi_{T_R}/\phi_R|$ exceeds $4$ on all nine SMG rows. At $\Pc$ on all-antiperiodic boxes $\rho=0.004$--$0.012$, the elementary side. The $\Pc$ ratio is boundary-condition dependent ($0.0378\pm0.0034$ on $6^4$ pppa against $0.0087\pm0.0004$ aaaa at $h=0.5$) while the light momentum readout is not; this is a finite-box fact. The measurement is, in lattice flavour, the hand-over described by the Eichten--Preskill picture and by Golterman--Shamir's conjectured composite interpolating fields~\cite{Eichten:1985ft,Golterman:2025boq}. Two differences from the continuum picture: the source is a nodal pairing field rather than a Majorana mass, and $\Psi^a$ is an on-site three-fermion operator in the lattice $\bar{\mathbf4}$, not the $\bar\psi\bar\psi\psi\in\overline{\mathbf{16}}$ ($U(1)_\varepsilon\neq U(1)_\psi$). It is a small-quasiparticle-weight statement at $L\le8$, not a proof that the composite partner exists as a state.

\begin{figure*}[!t]
\includegraphics[width=\textwidth]{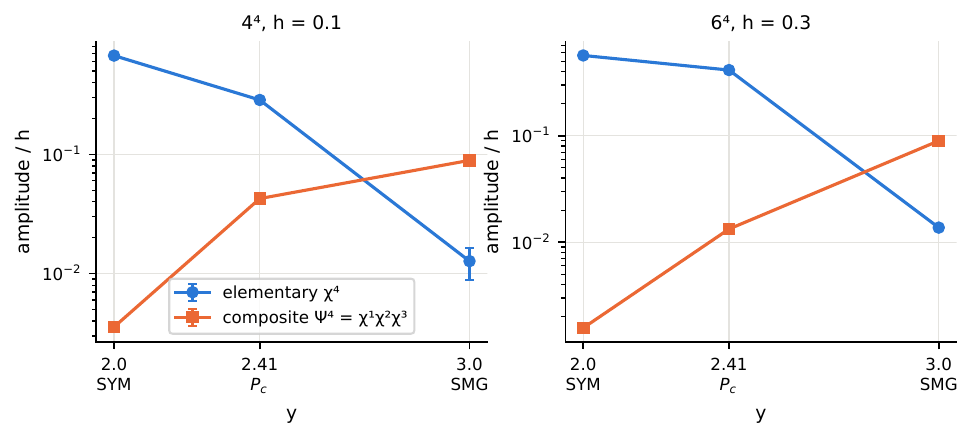}
\caption{Composite takeover with a nodal $\Z_4$-odd source on one flavour: elementary ($\chi^4$) and composite ($\Psi^4=\chi^1\chi^2\chi^3$) amplitudes per unit $h$ in the source direction at $y=2.0$ (SYM), $2.41$ ($\Pc$) and $3.0$ (SMG) on $4^4$ ($h=0.1$) and $6^4$ ($h=0.3$), production boundary conditions; sign-reweighted means, 8-block errors.}
\label{fig:takeover}
\end{figure*}

\paragraph{The zero readout.} In the $\tau_2K$ class the frequency winding of the corner-block propagator vanishes identically (every loop step within $4.1\times10^{-14}$ of $\{0,\pi\}$), and its quantisation would need gaps above the bandwidth at $L_t\le8$: no integer order parameter here, in line with the Green's-function-zero literature~\cite{Gurarie:2011,You:2014cka,You:2017ltx,Xu:2021ztz,Lu:2022qtc,Lu:2023cev,Golterman:2023zqf} and with the Nielsen--Ninomiya constraints on SMG spectra derived from propagator zeros~\cite{Golterman:2026sox}. The odd-part exponent $\alpha$ separates SMG from SYM configuration by configuration at $L=6$ and $8$, and, taste-projected, reads the light half in SMG at $\alpha_L=0.928$--$0.977$ at every $h$ and lattice (free $0.05$--$0.58$). This is the SMG-phase value of either doublet, not an identification of a zero specific to the light half: $\alpha_R=\alpha_L$ within $1.8\sigma$ on every SMG row, the sourced doublet included (claim K6.3). On all-antiperiodic boxes $\alpha$ cannot tell a pole from a zero of the same gap below the loop resolution ($0.77$ at $L_t=8$), so at $\Pc$ it reads ``pole-side, moving toward free with $h$'', not a mass.

\section{An explicit Majorana mass on half of the generation}
\label{sec:halfmass}

Model N1 is the first lattice study of an explicit Majorana bilinear on half of an SMG multiplet; the nearest prior settings gap part of a multiplet by interactions (mirror decoupling~\cite{Eichten:1985ft,Chen:2012di,Poppitz:2010at,Kikukawa:2017gvk,Catterall:2020fep,Zeng:2022grc,Araki:2026mqi}), split a three-dimensional SMG point by a symmetric four-fermion coupling~\cite{Maiti:2025ley,Maiti:2026log}, or set the light scale of a near-conformal system by heavy flavours~\cite{Brower:2015owo,Hasenfratz:2016gut}. In the free theory the gate holds exactly: $G_L(p_{\min})=1$ on every all-antiperiodic box at every $h$ while $G_R(p_{\min})$ falls ($0.9279$, $0.4457$ at $h=0.5$, $2$ on $8^4$); every interacting light number is read as a ratio to its free value at the same $(L,L_t,{\rm bc},h)$. Stage 0 ($4^4$, $4^3\times8$) and stage 1 ($8^4$, $6^4$, $6^3\times12$; $y\in\{2.41,3.0\}$, $h\in\{0,0.5,1,2\}$) are reported as written in Table~\ref{tab:prereg}: stage 0 is not a verdict and stage 1 returns none; the outcomes come from stages 1b and 1c.

\subsection{The SMG phase}
\label{sec:smg}

At $y=3.0$ the stage-1b clauses on the new samples and, on two fresh $8^4$ seeds per $h$, read class (a) by the ledger: gap readouts unchanged, with a small drift toward free, and the SSB clause at its $2\sigma$ edge under the re-specified threshold (Table~\ref{tab:smg}). In this phase the source's own linear response is screened: the elementary heavy amplitude is $\phi_R/h\le0.011\times$free on every SMG row (largest $0.011$ at $h=2$; claim K4.3), $\phi_R=7.5\times10^{-4}$ against $0.072$ free at $h=2$ on $8^4$, while the composite carries the response (Sec.~\ref{sec:channel}). The gap statement rests on the single-fermion readouts. The momentum readout is $g=0.0030$--$0.0034$ on $8^4$, $0.3$\,\% of free ($1.3$--$1.5$\,\% on $6^4$, $\dagger$), and the odd-part exponent on the new $8^4$ configurations is $\alpha_L=0.955$--$0.959$, its SMG-phase value, equal to $\alpha_R$ within $1.8\sigma$. The pair susceptibility moves while $g$ does not: on the $8^4$ stage-1 chains $g$ goes from $0.0028$ to $0.0033$ and $c=\chi_L/{\rm free}$ from $0.0022\pm0.0011$ to $0.0121\pm0.0015$ for $h=0\to2$ (claim K8.5), from $0.2\,\%$ to $1.2\,\%$ of free ($1.7\,\%$ pooled over the stage-1c seeds). The light pair-channel effective mass relative to $h=0$ and to the free ratio, $r_L^{\cosh}=0.977\pm0.002$ and $0.921\pm0.002$ at $h=1$ and $2$ on $8^4$, is supporting evidence only: the raw midpoint mass itself moves by $-0.5\,\%$ on $8^4$ ($0.7163$ to $0.7125$ for $h=0\to2$), the $8\,\%$ of $r_L^{\cosh}$ being relative to the free mass, which rises with $h$ on $8^4$ through the source's light block (Sec.~\ref{sec:methods}); it is an effective mass without a plateau, and at $L=6$, $h=0$ it changes from $0.763$ to $0.423$ between $L_t=6$ and $12$.

Every SSB clause of the ledger reads one channel: $\chi_L=V\langle\phi_L^2\rangle$ with $\phi_L$ the amplitude along $P_LC_{\rm asd}(0,3)P_L$, the anti-self-dual $(0,3)$ light Majorana channel, the one light mass direction lying in the planes of the source (Table~\ref{tab:strings}). Symmetry breaking at $(3.0,2)$ is not resolved at $L\le8$ by any of those clauses; every inequivalent light mass channel is read on the $(4^4,8^4)$ pair below (Table~\ref{tab:channels}). The stage-1 SSB clause fired by its letter at this point: $e=0.786\pm0.129$ on $(6^4,8^4)$ against a free value $0.39$, so $e-2\sigma=0.527>0.5$. The stage-1b clause, re-specified relative to free, was fixed before the new trajectories and after that row had been read (Table~\ref{tab:prereg}). The re-specified $(6,8)$ growth is $e-e_{\rm free}=0.600\pm0.104$, pooled over the two independent stage-1c seeds; the growth itself is $e=0.99\pm0.10$, consistent with $e=1$ and $5.8\sigma$ above free, with its $2\sigma$ edge at $0.39$, below the clause's $0.5$. The stage-1c seeds replicate this number; they do not test the criterion. Beside it, $R_{\rm peak}$ is below free on $8^4$ ($0.880\pm0.083$) and rises from $6^4$ ($0.614\pm0.049$) by $2.8\sigma$, and $\chi_L$ is at $1.7$\,\% of free.

The $(6^4,8^4)$ pair compares the IR-most fifth of the $6^4$ modes with every mode of $8^4$, and the $(4^4,8^4)$ pair a source-free light doublet with a partly dressed one (Sec.~\ref{sec:methods}). On $(4^4,8^4)$ the same quantity is $e-e_{\rm free}=-0.32\pm0.04$ at $h=2$ and $-0.44\pm0.07$ at $h=1$ (claim K5.13, whose reader gives every number of this paragraph), $8.9\sigma$ from the $(6,8)$ value at $h=2$ read with the same reader, $+0.513\pm0.085$ ($6^4$ stage 1b against the pooled $8^4$; the stage-1c reader gives the $0.600\pm0.104$ above). Over $4^4$, $4^3\times8$ and $8^4$ the susceptibility $\chi_L/$free falls monotonically ($0.037$, $0.026$, $0.015$ at $h=2$), and the $6^4$ value ($0.008$) lies below all three. The like-for-like control (Table~\ref{tab:modeset}, claim K5.15) recomputes $\chi_L$ on 90 stored $8^4$ configurations at $(3.0,2)$, every second configuration from trajectory 100 of the stage-1c seeds rA and rB, with $P_L$ restricted to the 1536 source-free momenta and, separately, to the 256 IR-most momenta (all $|\cos p_\mu|=\cos(\pi/8)$, the $8^4$ analogue of the $6^4$ mode set), each read against the free value of the same restricted observable ($0.541$ and $0.448$ of the free $\chi_L$, $h$-independent) with the error model of K5.13. On the same 90 configurations the full $P_L$ gives $\chi_L/{\rm free}=0.0156\pm0.0021$, $e-e_{\rm free}=-0.31\pm0.06$ on $(4,8)$ and $+0.54\pm0.13$ on $(6,8)$; the source-free class gives $0.0056\pm0.0006$, $-0.68\pm0.05$ and $-0.35\pm0.11$; the IR class $0.0024\pm0.0003$, $-0.98\pm0.05$ and $-1.08\pm0.11$. On like mode sets $\chi_L/$free falls with the volume on both pairs at $h=2$: the $(6,8)$ growth is a mode-set effect of the $6^4$ projector, measured at $h=2$ (the $h=1$ and $h=0$ rows of the control are not measured), and the partial outcome ``growth relative to free'' of stages 1b and 1c is a comparison across unlike mode sets. The source-dressed modes carry more of the interacting $\chi_L$ than their free share (per configuration the IR class carries $5\,\%$ of $\chi_L$ against $45\,\%$ in the free theory), and the $8^4$ IR class at $p=\pi/8$ sits closer to the corner than the $6^4$ modes at $\pi/6$, in the direction of stronger screening, so $-1.08\pm0.11$ bounds the like-for-like $(6,8)$ fall rather than measuring it. No SSB-scale growth appears on any like-mode-set pair. The $4^4$ box is three correlation lengths across; a $12^4$ point at $(3.0,2)$ would settle it.

Table~\ref{tab:channels} (claim K5.16) gives the susceptibility of every inequivalent light mass-type channel of Table~\ref{tab:strings}, computed from the exact propagator on the stored configurations with the full $8^4$ light projector $P_L$: the three anti-self-dual Majorana channels (one orbit each under the stabiliser), the four one-link $LR$ channels (the Dirac light--heavy mixing a seesaw needs; one orbit), the on-site sextet (the fermionic light part of the AFM order parameter), the four-link sextet and the four three-link $LR$ sextet channels (one orbit), each sextet in the $(\mathbf3,\mathbf1)$ and the $(\mathbf1,\mathbf3)$ of $SO(4)$. On $(4^4,8^4)$ at $(3.0,2)$ no channel grows faster than free: the three Majorana channels fall ($-0.33\pm0.05$, $-0.50\pm0.07$, $-0.38\pm0.07$; $1.5$, $0.5$, $0.6\,\%$ of free on $8^4$), the one-link $LR$ channels are free-like ($0.34$--$0.43$ of free on both boxes, $-0.055\pm0.015$), the on-site sextet is at $0.97$ and $0.93$ of free ($-0.016\pm0.030$), the four-link sextet gives $-0.22\pm0.03$ and the three-link sextets $-0.16\pm0.01$ and $-0.10\pm0.03$. At $(\Pc,2)$ on $(4^4,8^4)$ every channel lies within $+0.008$ to $+0.079$ of free scaling, the largest excess in the on-site sextet ($+0.079\pm0.026$), the AFM channel, whose $S(\pi)$ at $(\Pc,2)$ is at its SMG value (Sec.~\ref{sec:critical}). On $(6^4,8^4)$ at $(3.0,2)$ the one-link, on-site and three-link channels read $+1.08\pm0.04$, $+1.76\pm0.06$, $+0.64\pm0.02$ and $+0.85\pm0.05$, numbers of SSB size, because the $6^4$ value lies $4$--$8\times$ below both other boxes in each of them (on-site sextet $0.97$, $0.12$, $0.93$ of free on $4^4$, $6^4$, $8^4$; one-link $0.39$, $0.10$, $0.34$): the mode-set effect of the $6^4$ projector, channel by channel, so the $(6,8)$ clauses carry no information on symmetry breaking at $y=3.0$. At $\Pc$ the $(6,8)$ excess is $+0.02$ to $+0.14$, and $+0.54\pm0.06$ in the on-site sextet. These are $h=2$, two-volume statements; the free normalisation removes the free counting of the $6^4$ mode set, not its interacting weight.

\begin{table*}[!t]
\caption{Susceptibility of every inequivalent light mass-type channel at $h=2$ (claim K5.16), as a ratio to the free value on the same box, and $e-e_{\rm free}$ on $(4^4,8^4)$ and on $(6^4,8^4)$ ($\dagger$: unlike mode sets). Exact propagators on the stored configurations, full $8^4$ light projector $P_L$; orbit members averaged (they agree within errors). [$\chi_L$]: the channel of the pre-registered clauses. Samples: at $y=3.0$, $4^4$ 225 (stage 0, trajectories $\ge150$), $6^4$ 250 (stage-1b trajectories $\ge1000$), $8^4$ 139 (stage-1c seeds rA, rB, 45 each, and the stage-1 chain, 49); at $\Pc$, $4^4$ 225, $6^4$ 225 (stage 1, trajectories $\ge100$), $8^4$ 45 (stage 1). Error model of K5.13. The $(\mathbf1,\mathbf3)$ on-site and four-link sextet channels are not positive at $y=3.0$ (their interacting connected part exceeds the free value with the opposite sign) and are not read; at $\Pc$ they give $+0.001\pm0.002$ and $+0.008\pm0.002$ on $(4,8)$. Controls: the self-dual planes, with zero light part on $4^4$ and $6^4$, carry on $8^4$ $0.008$, $0.008$, $0.000$ ($y=3.0$) and $0.803$, $0.802$, $0.666$ ($\Pc$) of their free light block, the source's own light block.}
\label{tab:channels}
\begin{ruledtabular}
\begin{tabular}{llllll}
light channel & $4^4$ & $6^4$ $\dagger$ & $8^4$ & $e-e_{\rm free}$ (4,8) & $e-e_{\rm free}$ (6,8) $\dagger$\\
\colrule
\multicolumn{6}{l}{$y=3.0$ (SMG), $h=2$}\\
anti-self-dual $(0,3)$ [$\chi_L$] & $0.0365\pm0.0030$ & $0.00839\pm0.00058$ & $0.0148\pm0.0015$ & $-0.325\pm0.048$ & $+0.493\pm0.108$\\
anti-self-dual $(0,1)$ & $0.0214\pm0.0023$ & $0.00406\pm0.00039$ & $0.00529\pm0.00087$ & $-0.504\pm0.070$ & $+0.231\pm0.166$\\
anti-self-dual $(0,2)$ & $0.0166\pm0.0018$ & $0.00453\pm0.00070$ & $0.00575\pm0.00090$ & $-0.382\pm0.069$ & $+0.208\pm0.191$\\
one-link $LR$, $\mu=0$--$3$ & $0.3903\pm0.0116$ & $0.0970\pm0.0028$ & $0.3353\pm0.0104$ & $-0.055\pm0.015$ & $+1.077\pm0.036$\\
on-site sextet $(\mathbf3,\mathbf1)$, AFM & $0.9721\pm0.0572$ & $0.1235\pm0.0052$ & $0.9301\pm0.0530$ & $-0.016\pm0.030$ & $+1.755\pm0.062$\\
four-link sextet $(\mathbf3,\mathbf1)$ & $0.1410\pm0.0076$ & $0.0756\pm0.0031$ & $0.0763\pm0.0037$ & $-0.221\pm0.026$ & $+0.009\pm0.055$\\
three-link $LR$ sextet $(\mathbf3,\mathbf1)$ & $0.0294\pm0.0008$ & $0.00913\pm0.00015$ & $0.0192\pm0.0003$ & $-0.155\pm0.011$ & $+0.644\pm0.020$\\
three-link $LR$ sextet $(\mathbf1,\mathbf3)$ & $0.0172\pm0.0015$ & $0.00485\pm0.00021$ & $0.0129\pm0.0003$ & $-0.101\pm0.033$ & $+0.853\pm0.045$\\
\colrule
\multicolumn{6}{l}{$\Pc=(2.41,-0.01)$, $h=2$}\\
anti-self-dual $(0,3)$ [$\chi_L$] & $0.7942\pm0.0053$ & $0.8006\pm0.0031$ & $0.8508\pm0.0090$ & $+0.025\pm0.004$ & $+0.053\pm0.010$\\
anti-self-dual $(0,1)$ & $0.7808\pm0.0055$ & $0.7821\pm0.0028$ & $0.8184\pm0.0119$ & $+0.017\pm0.006$ & $+0.039\pm0.013$\\
anti-self-dual $(0,2)$ & $0.7815\pm0.0052$ & $0.7794\pm0.0025$ & $0.7995\pm0.0080$ & $+0.008\pm0.004$ & $+0.022\pm0.009$\\
one-link $LR$, $\mu=0$--$3$ & $1.1260\pm0.0093$ & $1.0229\pm0.0053$ & $1.2023\pm0.0208$ & $+0.024\pm0.007$ & $+0.140\pm0.016$\\
on-site sextet $(\mathbf3,\mathbf1)$, AFM & $1.9584\pm0.0624$ & $1.3019\pm0.0268$ & $2.4346\pm0.1565$ & $+0.079\pm0.026$ & $+0.544\pm0.059$\\
four-link sextet $(\mathbf3,\mathbf1)$ & $0.9322\pm0.0081$ & $1.0885\pm0.0164$ & $1.0564\pm0.0234$ & $+0.045\pm0.009$ & $-0.026\pm0.023$\\
three-link $LR$ sextet $(\mathbf3,\mathbf1)$ & $0.8069\pm0.0020$ & $0.8527\pm0.0010$ & $0.8408\pm0.0012$ & $+0.015\pm0.001$ & $-0.012\pm0.002$\\
three-link $LR$ sextet $(\mathbf1,\mathbf3)$ & $0.8000\pm0.0020$ & $0.8329\pm0.0008$ & $0.8264\pm0.0011$ & $+0.012\pm0.001$ & $-0.007\pm0.001$\\
\end{tabular}
\end{ruledtabular}
\end{table*}

The $8^4$ value of $\chi_L/$free at $h=2$ is $0.0121\pm0.0015$ on the stage-1 chain (stage-1b reader, claim K8.5, trajectories 100--1000), $0.0142\pm0.0020$ on the stage-1b extension (same reader), $0.0167\pm0.0016$ pooled over the stage-1c seeds (stage-1c reader, K5.11) and $0.0151$ pooled over all four samples (K5.13 reader, which gives $0.0116$ for the stage-1 chain alone), a $30\,\%$ spread across samples and readers. Inside SMG $g$ and $c$ rise with $h$ on all three lattices, while the pair-channel effective mass is flat on $8^4$ within $0.008$, falls on $6^4$ and rises on $6^3\times12$: the drift toward free is carried by $g$ and $c$.

\begin{table*}[!t]
\caption{Restricted-mode control at $(3.0,h{=}2)$ on $8^4$ (claim K5.15): $\chi_L$ recomputed on 90 stored stage-1c configurations (45 per seed, every second one from trajectory 100) with the light pattern $P_LC_{\rm asd}P_L$ restricted to a momentum class, as a ratio to the free value of the same restricted observable, and the resulting $e-e_{\rm free}$ against the $4^4$ stage-0 chain and the $6^4$ stage-1b trajectories. The source-free class carries $0.541$ and the IR class $0.448$ of the free $\chi_L$, independently of $h$. Error model of K5.13.}
\label{tab:modeset}
\begin{ruledtabular}
\begin{tabular}{llllll}
$8^4$ light pattern & $\chi_L/{\rm free}$, rA & rB & pooled & $e-e_{\rm free}$ (4,8) & $e-e_{\rm free}$ (6,8)\\
\colrule
full $P_L$ (stored values, same configurations) & $0.0186\pm0.0033$ & $0.0108\pm0.0027$ & $0.0156\pm0.0021$ & $-0.306\pm0.058$ & $+0.540\pm0.134$\\
source-free class (1536 momenta) & $0.0065\pm0.0009$ & $0.0041\pm0.0007$ & $0.0056\pm0.0006$ & $-0.676\pm0.049$ & $-0.352\pm0.110$\\
IR class (256 momenta, $|\cos p_\mu|=\cos\pi/8$) & $0.0029\pm0.0004$ & $0.0017\pm0.0003$ & $0.0024\pm0.0003$ & $-0.976\pm0.050$ & $-1.076\pm0.113$\\
\end{tabular}
\end{ruledtabular}
\end{table*}

\begin{table*}[!t]
\caption{SMG phase, $y=3.0$: the pre-registered clauses on new samples. Stage 1c pools two independent $8^4$ seeds per $h$ (90 measurements each); the $6^4$ side is the stage-1b extension. $r_L^{\cosh}=[m(h)/m(0)]/[m_{\rm free}(h)/m_{\rm free}(0)]$. $\dagger$: $(6^4,8^4)$ light-sector ratio, unlike mode sets (Sec.~\ref{sec:methods}). The $(4,8)$ rows use the stage-0 $4^4$ chains and all four $8^4$ samples (K5.13 reader); the $(4,8)$ pair compares a source-free light doublet with a partly dressed one, Table~\ref{tab:modeset} the like-for-like restriction.}
\label{tab:smg}
\begin{ruledtabular}
\begin{tabular}{lll}
clause & $h=1$ & $h=2$\\
\colrule
$r_L^{\cosh}(8^4)$, stage 1c pooled & $0.977\pm0.002$ & $0.921\pm0.002$\\
$r_L^{\cosh}(8^4)$, stage 1b & $0.976\pm0.002$ & $0.920\pm0.002$\\
$r_L^{\cosh}(6^4)$, stage 1b & $0.986\pm0.004$ & $0.923\pm0.004$\\
$g$, $8^4$ / $6^4$ & $0.0030$ / $0.013$ & $0.0034$ / $0.015$\\
$\alpha_L$, new $8^4$ configurations & $0.955\pm0.006$ & $0.959\pm0.006$\\
$e-e_{\rm free}$ on (6,8), 1c pooled $\dagger$ & $0.296\pm0.167$ & $0.600\pm0.104$\\
$R_{\rm peak}$, $8^4$ / $6^4$ $\dagger$ & $0.93\pm0.26$ (1b) / --- & $0.880\pm0.083$ / $0.614\pm0.049$ (1c)\\
$\chi_L/{\rm free}$, $8^4$ vs $6^4$ (1c) $\dagger$ & --- & $0.0167\pm0.0016$ vs $0.0084\pm0.0006$\\
$e-e_{\rm free}$ on (4,8), split on every momentum & $-0.44\pm0.07$ & $-0.32\pm0.04$\\
$\chi_L/{\rm free}$ on $4^4$ / $4^3\times8$ / $8^4$ & $0.0152$ / --- / $0.0045$ & $0.0365$ / $0.0255$ / $0.0151$\\
class & (a) & (a)\\
\end{tabular}
\end{ruledtabular}
\end{table*}

\subsection{The critical point}
\label{sec:critical}

At $\Pc$ the sourced half's mass detunes the critical point (Table~\ref{tab:pc}, Fig.~\ref{fig:pc}). Six facts are measured on two spatial sizes (three boxes: $8^4$, $6^4$, $6^3\times12$) for $h\le3$.
(1) The staggered $\sigma$ structure factor falls monotonically with $h$ on every lattice (every step from $h\ge0.5$ a fall at $>2.5\sigma$; $8^4$: $26.0\to2.13$ for $h=0\to3$), to $27$--$31$\,\% below the SMG ($y=3.0$, $h=0$) value of the same lattice; $|\Sigma_{\rm stag}|$ likewise.
(2) The light pair-channel effective mass rises monotonically (every step $5$--$24\sigma$) and stays below the free value on the same box, $L_t$, boundary condition and $h$ at every $h$: $m/m_{\rm free}$ rises from $0.129$ to $0.442$ on $8^4$ and from $0.407$ to $0.670$ on $6^4$. The $h=0$ value itself halves from $6^4$ to $8^4$ ($0.392\to0.197$), so the fraction of free is strongly $L$-dependent at small $h$; the $8^4/6^4$ ratio of the $t=2$ cosh mass is within $15\,\%$ of $1$ only at $h\ge1.5$ ($1.065$, $0.95$, $1.15$ at $h=2$, $1.5$, $3$; $0.50$ at $h=0$).
(3) $g$ and $c$ rise monotonically toward $1$ on all three lattices ($8^4$: $0.730\to0.867$, $0.517\to0.920$), never exceed it, and never move toward their SMG values ($0.003$, $0.002$).
(4) No symmetry-breaking signal in any inequivalent light mass channel on the $(4^4,8^4)$ pair at $h=2$ (Table~\ref{tab:channels}: $e-e_{\rm free}$ between $+0.008$ and $+0.079$, the largest in the on-site sextet), and none in the $(0,3)$ channel of the pre-registered clauses on either pair: on $(6^4,8^4)$ ($\dagger$) $e-e_{\rm free}\in[-0.09,+0.05]$ at every $h$, on $(4^4,8^4)$ (stage-0 $4^4$ chains, 300 measurements each; K5.13 reader) $e-e_{\rm free}$ between $-0.114\pm0.014$ at $h=0$ and $+0.025\pm0.004$ at $h=2$, within $0.05$ of the $(6,8)$ value at every $h$ (Table~\ref{tab:pc48}), and $R_{\rm peak}$ is within $2\sigma$ of $1$ except at $h=3$ on $8^4$ ($1.022\pm0.008$, a 2\,\% effect, recorded). The single-fermion readout loses more from $4^4$ to $8^4$ at $\Pc$ than at the symmetric point: $g(8^4)/g(4^4)=0.878$--$0.933$ against $0.955$ at $(y,h)=(2.0,0)$, while $g(8^4)/g(6^4)=0.97$--$1.01$ (symmetric point $0.985$).

\begin{table}[!htb]
\caption{Volume pairs at $\Pc$ (claim K5.15; K5.13 reader): $e-e_{\rm free}$ of $\chi_L$ in the $(0,3)$ light Majorana channel and the ratio of $g$, larger over smaller box, against the stage-0 $4^4$ chains (300 measurements) and the $6^4$ chains; symmetric-point $g$ ratios $0.955$ (8/4), $0.985$ (8/6). $\dagger$: unlike mode sets.}
\label{tab:pc48}
\begin{ruledtabular}
\begin{tabular}{lllll}
$h$ & $e-e_{\rm free}$ (4,8) & (6,8) $\dagger$ & $g$ 8/4 & 8/6 $\dagger$\\
\colrule
$0$ & $-0.114\pm0.014$ & $-0.086\pm0.034$ & $0.878$ & $0.992$\\
$0.5$ & $-0.080\pm0.008$ & $-0.049\pm0.022$ & $0.920$ & $1.008$\\
$1$ & $-0.024\pm0.004$ & $+0.021\pm0.010$ & $0.926$ & $1.006$\\
$2$ & $+0.025\pm0.004$ & $+0.053\pm0.010$ & $0.933$ & $0.972$\\
\end{tabular}
\end{ruledtabular}
\end{table}
(5) Neither response is a single power of $h$. For $m$, the pre-registered form $m^2=m_0^2+(ch^\kappa)^2$ on $h\in[0.5,3]$ gives no acceptable single power on any lattice: $\kappa=0.44\pm0.06$ on $6^3\times12$, outside the window $[0.7,1.3]$ with $\chi^2/\mathrm{dof}=1.2$, and $\kappa=0.64\pm0.03$ and $0.75\pm0.02$ with $\chi^2/\mathrm{dof}=13$ and $27$ on $6^4$ and $8^4$ (by the letter of T2: FAIL on $6^3\times12$, UNDECIDED on the two symmetric boxes). The $6^3\times12$ result is the one on the anisotropic box, where the free midpoint has no cosh solution and the $(0,3)$ light one-point function is a shape term (Sec.~\ref{sec:nomass}). The local quadrature exponent falls monotonically, $1.8\to1.3\to0.9\to0.6$ on $8^4$. For $S(\pi)$, the power form $S-S_\infty\propto h^{-x}$ has $\chi^2/\mathrm{dof}=12$--$15$ on every lattice (T3 UNDECIDED). Any saturating crossover with one scale $h_0\approx2$ fits instead: a Pad\'e form $1/S=s_0+bh^2/(1+h^2/h_0^2)$ gives $\chi^2=0.34$, $1.86$, $0.92$ on 3 dof with $h_0=2.10$, $2.21$, $2.05$ on $8^4$, $6^4$, $6^3\times12$, and $\tanh^2$ and Gaussian forms do as well. The window $h\in[0.5,3]$ is not a scaling window: the heavy doublet's mass crosses the lattice bandwidth there.
(6) The small-$h$ onset on $6^4$, from independent chains at $h=0.25$, $0.5$, $0.75$ and the stored $h=1$ point, is $p=1.14\pm0.37$ in $m^2-m_0^2=(ch^p)^2$ ($\chi^2/\mathrm{dof}=0.17$). This is the onset required by symmetry: $\varepsilon D(h,\sigma)\varepsilon=-D(-h,-\sigma)$, so every observable is even in $h$ and analytic in $h^2$ on a gapped finite box, which is $p=1$ when $\delta m\ll m_0$. The test has no power beyond that.

\begin{table*}[!t]
\caption{Model N1 at $\Pc=(2.41,-0.01)$, all-antiperiodic boxes, stored stage-1 rows ($h\in\{0,0.5,1,2\}$) and stage-1b rows ($h=1.5,3$). $m$: light pair-channel midpoint cosh mass at $t=L_t/2-1$, with the free value on the same box and $h$; $m_{\rm SMG}$: the same quantity at $y=3.0$ and the same $h$; $S(\pi)$: staggered $\sigma$ structure factor, with the SMG ($y=3.0$, $h=0$) value of the same lattice in the last column; $g=G_L(p_{\min})/$free; $c=\chi_L/$free. Errors: blocked jackknife with the error-model inflation; the $6^3\times12$ free midpoint has no cosh solution. $\dagger$: light-sector numbers on $L=6$ boxes rest on the partial taste split (Sec.~\ref{sec:methods}).}
\label{tab:pc}
\begin{ruledtabular}
\begin{tabular}{llllllllll}
lattice & & $h=0$ & $0.5$ & $1$ & $1.5$ & $2$ & $3$ & SMG\\
\colrule
$8^4$ ($t=3$) & $m$ & $0.197\pm0.009$ & $0.211\pm0.013$ & $0.332\pm0.016$ & $0.487\pm0.010$ & $0.614\pm0.005$ & $0.764\pm0.004$ & $0.716\pm0.002$\\
 & $m_{\rm free}$ & $1.519$ & $1.529$ & $1.556$ & $1.596$ & $1.643$ & $1.731$ & \\
 & $m_{\rm SMG}$ & $0.716$ & $0.716$ & $0.720$ & --- & $0.713$ & --- & \\
 & $S(\pi)$ & $25.99\pm2.34$ & $14.51\pm1.43$ & $6.39\pm0.44$ & $3.88\pm0.17$ & $2.91\pm0.12$ & $2.13\pm0.07$ & $3.02\pm0.12$\\
 & $g$ & $0.730\pm0.034$ & $0.770\pm0.029$ & $0.784\pm0.015$ & $0.790\pm0.005$ & $0.821\pm0.002$ & $0.867\pm0.002$ & $0.0028$\\
 & $c$ & $0.517$ & --- & $0.690$ & --- & $0.850$ & $0.920$ & $0.0022$\\
$6^4$ ($t=2$) $\dagger$ & $m$ & $0.392\pm0.008$ & $0.407\pm0.009$ & $0.453\pm0.004$ & $0.511\pm0.003$ & $0.569\pm0.002$ & $0.645\pm0.002$ & $0.763\pm0.002$\\
 & $m_{\rm free}$ & $0.962$ & $0.962$ & $0.962$ & $0.962$ & $0.962$ & $0.962$ & \\
 & $m_{\rm SMG}$ & $0.763$ & $0.761$ & $0.752$ & --- & $0.713$ & $0.610$ & \\
 & $S(\pi)$ & $10.20\pm0.83$ & $8.90\pm0.62$ & $5.53\pm0.40$ & $3.69\pm0.14$ & $2.87\pm0.13$ & $2.23\pm0.08$ & $3.23\pm0.14$\\
 & $g$ & $0.735\pm0.008$ & $0.765\pm0.008$ & $0.781\pm0.004$ & $0.815\pm0.003$ & $0.843\pm0.002$ & $0.887\pm0.001$ & $0.0125$\\
 & $c$ & $0.568$ & --- & $0.674$ & --- & $0.800$ & $0.880$ & $0.0019$\\
$6^3\times12$ ($t=5$) $\dagger$ & $m$ & $0.327\pm0.022$ & $0.359\pm0.011$ & $0.416\pm0.006$ & $0.449\pm0.002$ & $0.475\pm0.001$ & $0.526\pm0.002$ & $0.423\pm0.003$\\
 & $m_{\rm SMG}$ & $0.423$ & $0.421$ & $0.429$ & --- & $0.458$ & --- & \\
 & $S(\pi)$ & $12.12\pm0.85$ & $8.65\pm0.74$ & $5.15\pm0.30$ & $3.71\pm0.16$ & $2.75\pm0.11$ & $2.16\pm0.07$ & $2.95\pm0.14$\\
 & $g$ & $0.677\pm0.014$ & $0.690\pm0.010$ & $0.730\pm0.006$ & $0.776\pm0.002$ & $0.817\pm0.003$ & $0.867\pm0.002$ & $0.0059$\\
 & $c$ & $0.563$ & --- & $0.637$ & --- & $0.698$ & $0.813$ & \\
\end{tabular}
\end{ruledtabular}
\end{table*}

\begin{figure*}[!t]
\includegraphics[width=\textwidth]{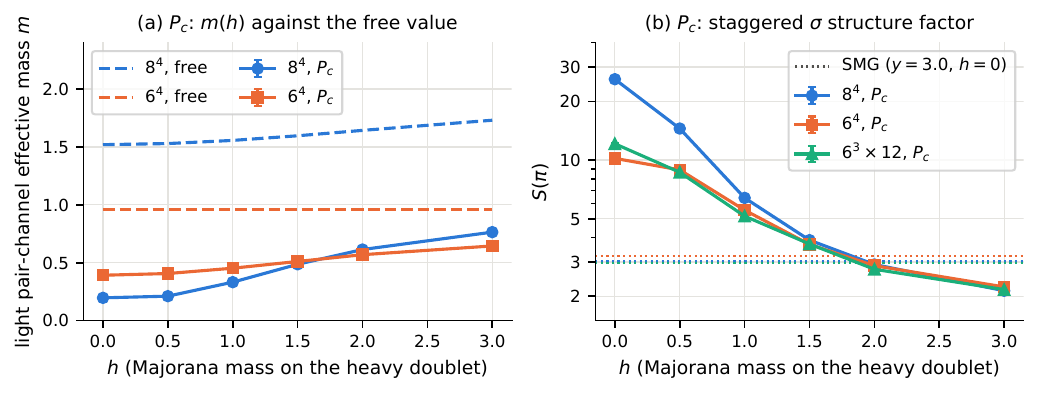}
\caption{Model N1 at $\Pc$. (a) The light pair-channel midpoint effective mass $m$ on $8^4$ and $6^4$ against $h$ (solid) and the free value on the same box, $L_t$, boundary condition and $h$ (dashed): $m$ rises toward free and never reaches it. (b) The staggered $\sigma$ structure factor $S(\pi)$ on three lattices (log scale); dotted lines are the SMG ($y=3.0$, $h=0$) values of the same lattices. Errors: blocked jackknives with the error-model inflation; generated from the committed table \texttt{results/K8\_gap\_vs\_h.csv} of Ref.~\cite{repo}.}
\label{fig:pc}
\end{figure*}

\paragraph{Reading.} The pair-channel effective mass at $\Pc$, $h=0$ is low ($0.13\times$ free on $8^4$) because the critical $\sigma$ mode dominates the two-fermion channel; its rise with $h$ is the loss of that mode. On $6^4$ it is the same fraction of free on either side of the transition ($0.82$ at $y=2.0$, $0.79$ at $y=3.0$) while $g$ is $0.914$ against $0.0125$; on $8^4$ the fractions do not coincide, $0.63$ (SYM, $h=0$), $0.65$ (SYM, $h=3$) and $0.47$ (SMG), and $\Pc$ at $h=3$ sits at $0.44$, below both phases (Table~\ref{tab:direction}). $m$ measures distance from criticality; its rise with $h$ (``anti-seesaw'' in the companion repository) is the light half going free-like, not a gap beyond free. At fixed $\sigma$ (stored $8^4$ $h=0$ configurations re-measured with $h$ in the operator) $h=2$ gives $0.34\pm0.09$ of the annealed rise, $h=1$ none.

\paragraph{Direction of the shift.} In the large-$N$ (Gaussian) approximation the fermion loop drives the staggered $\sigma$ instability, and because $\varepsilon$ anticommutes with $K$ and commutes with the same-parity $C_\chi$, the staggered curvature is $V[1-y^2\Pi(h)]$ with
\begin{align}
\Pi(h)&=\frac{2}{V}\Bigl[\mathrm{Tr}\,A^{-1}+2h^2\,\mathrm{Tr}\bigl(A^{-1}C_\chi\varepsilon\,A^{-1}C_\chi\varepsilon\bigr)\Bigr],\nonumber\\
A&=-(K-hC_\chi)^2,
\end{align}
whose second term is $-\|A^{-1/2}C_\chi\varepsilon A^{-1/2}\|^2\le0$. The closed form equals the finite-difference curvature of $\tfrac12\log\det(D^\dagger D)$ to $2\times10^{-7}$. $\Pi$ decreases strictly with $h$ on every box, so the large-$N$ critical coupling $y_{\rm MF}=\Pi^{-1/2}$ rises: $\times1.056$, $1.159$, $1.271$ at $h=1,2,3$ on $8^4$. The same sign follows from the production operator (ratios equal to the closed form within the asserted $2\times10^{-4}$, largest deviation $4.6\times10^{-5}$) and from the full finite-$\sigma_0$ mean-field potential, in which the staggered transition stays continuous and moves from $y^*=2.02$ to $2.36$ on $4^4$ ($1.88\to2.34$ on $6^4$) between $h=0$ and $3$. The explicit mass removes the heavy doublet from the loop, so a larger $y$ is needed to reach the instability. One assumption is stated, not derived: that the SYM/SMG boundary follows the $\sigma$-instability line for $h>0$ as it does at $h=0$. The large-$N$ line overshoots: it places the moved critical point at $y=2.41\times1.271=3.06$ at $h=3$, while $(y,h)=(3.0,3)$ is deep in SMG ($g=0.0204\pm0.0007$, $c=0.0164\pm0.0014$) and $(2.41,3)$ reads $g=0.887$, $c=0.880$ against $0.952$, $0.951$ at $(2.0,3)$; $y_c(3)$ itself is not located.

The direction is then read from data with a pre-registered three-clause test (P1). With $f_O=(O_{\Pc}-O_{\rm SYM})/(O_{\rm SMG}-O_{\rm SYM})$ the position of $\Pc$ between the SYM ($y=2.0$) and SMG ($y=3.0$) values of an observable, the clauses ask whether $f_g$ and $f_c$ fall from $h=0$ to $3$ and whether $S(\pi)$ at $(\Pc,h{=}3)$ lies closer to its SYM than to its SMG value. On $6^4$ the position moves toward SYM in $g$ ($f_g$: $0.198\to0.070$, $15\sigma$) and in $c$ ($f_c$: $0.363\to0.074$, $29\sigma$). The test returns MIXED because its third clause, $S(\pi)$, peaks at the critical point and falls on both sides, so it is direction-blind ($1.9\sigma$). Two further readings point the same way: with $h$-matched references at both ends the $\varepsilon$-vertex position $f_{O_4}$ falls from $0.250\pm0.005$ to $0.056\pm0.001$ (the heavy doublet's mass alone removes 46\,\% of $O_4$ at $y=2.0$), and inside SMG $g$ and $c$ rise with $h$ (Sec.~\ref{sec:smg}). On $8^4$ the same test (a $y=2.0$ pair at $h=0$ and $3$, 500 trajectories each) repeats the pattern: $g$ $0.190\to0.084$ ($2.8\sigma$), $c$ $0.419\to0.051$ ($15\sigma$), $S(\pi)$ clause $+0.30\pm0.20$ ($1.5\sigma$). MIXED by the letter on both volumes.

The volume trend at $h=3$ is split ($\dagger$; Table~\ref{tab:direction}): the $8^4/6^4$ ratio of $g$ at $\Pc$, $0.9776\pm0.0027$, lies below its SYM value $0.9936\pm0.0019$ (difference $-0.0159\pm0.0033$, $4.8\sigma$: a loss beyond the SYM dressing), while that of $c$, $1.0455\pm0.0114$ against $1.0204\pm0.0062$ (difference $+0.0251\pm0.0129$), is free-like. A light gap below the box scale $\pi/L$ is not excluded at $L\le8$. These tests are consistent with a critical coupling rising with the mass; they do not locate a critical line $y_c(h)$, and $h\in[0.25,3]$ is not a scaling window on $L\le8$.

\begin{table}[!htb]
\caption{Direction readouts at $\Pc$ on $6^4$ and $8^4$ (claims K8.5, K8.9). $m/m_{\rm free}$: light pair-channel midpoint effective mass over its free value on the same box and $h$; $f_O$: position of $\Pc$ between the SYM ($y=2.0$) and SMG ($y=3.0$) values of $O$, from $h=0$ to $3$; $S(\pi)$ clause: $|S_{\Pc}(3)-S_{\rm SMG}(0)|-|S_{\Pc}(3)-S_{\rm SYM}(3)|$; $v_O=O(8^4)/O(6^4)$ at $h=3$. $\dagger$: $6^4$ light-sector numbers rest on the partial taste split (Sec.~\ref{sec:methods}).}
\label{tab:direction}
\begin{ruledtabular}
\begin{tabular}{lll}
readout & $6^4$ $\dagger$ & $8^4$\\
\colrule
$m/m_{\rm free}$: SYM, $h{=}0$ / $3$ & $0.82$ / $0.86$ & $0.63$ / $0.65$\\
\quad SMG, $h{=}0$ & $0.79$ & $0.47$\\
\quad $\Pc$, $h{=}3$ & $0.67$ & $0.44$\\
$f_g$, $h{=}0\to3$ & $0.198\to0.070$ & $0.190\to0.084$\\
\quad significance & $15\sigma$ & $2.8\sigma$\\
$f_c$, $h{=}0\to3$ & $0.363\to0.074$ & $0.419\to0.051$\\
\quad significance & $29\sigma$ & $15\sigma$\\
$S(\pi)$ clause & $1.9\sigma$ & $+0.30\pm0.20$ ($1.5\sigma$)\\
P1 / P1$'$ as written & MIXED & MIXED\\
\colrule
$v_g$ at $\Pc$ / at SYM & \multicolumn{2}{l}{$0.9776\pm0.0027$ / $0.9936\pm0.0019$}\\
\quad difference & \multicolumn{2}{l}{$-0.0159\pm0.0033$ ($4.8\sigma$)}\\
$v_c$ at $\Pc$ / at SYM & \multicolumn{2}{l}{$1.0455\pm0.0114$ / $1.0204\pm0.0062$}\\
\quad difference & \multicolumn{2}{l}{$+0.0251\pm0.0129$}\\
\end{tabular}
\end{ruledtabular}
\end{table}

\paragraph{What the fit does not select.} With the parameter-free distance $\delta_L(h)=1-\Pi_L(h)/\Pi_L(0)$ on the same box, $1/S(\pi)=s_0+b\,\delta_L(h)$ fits the three lattices with $\chi^2/\mathrm{dof}=0.55$, $1.11$, $1.42$ and a lattice-independent $b$ within $0.15$, where a single power of $h$ gives $6.4$, $5.2$, $3.2$. This fit was formed after the data were known and carries little discriminating weight: generic saturating curves with $h_0\approx2$ fit as well ($\Delta\mathrm{AIC}=+0.0$, $-1.4$, $-2.8$), the shape of a mass on both doublets fits better on $6^4$ and $6^3\times12$, $\delta_L$ transfers across lattices better than a generic curve (5 of 6 directions) but not within errors (hold-out $\chi^2=8$--$28$ for 5 points), and the Gaussian $\sigma$-mode form $m^2=m_0^2+a\,\delta_L$ holds on $6^4$ ($\chi^2/\mathrm{dof}=1.3$) and not on $8^4$ ($24.8$). What survives of the large-$N$ description is the sign of the shift and the exclusion of a single power.

\subsection{The moonshot null}

A critical seesaw---the free-normalised light ratio falling with $h$ at $\Pc$ only---was pre-registered as a trigger and is not met at $L\le8$: $r_L$ rises monotonically on both volumes ($8^4$: $1.15\pm0.17$, $2.30\pm0.23$, $5.84\pm0.52$ at $h=0.5,1,2$ on the $t^*$ log-ratio), the stage-1b fall alert stays silent, an injected fall triggers both, and the free null triggers neither. No larger volume was run.

\section{Discussion}
\label{sec:discussion}

\paragraph{What the $\Z_4$ means for neutrinos.} On Standard-Model matter the $\Z_4$ of one generation is generated by $X=5(B-L)-4Y \bmod 4$, $X$ being the charge of the $U(1)$ commuting with $SU(5)$ in $SO(10)$~\cite{Wilczek:1981iz} and the $\Z_{4,X}$ reading that of Refs.~\cite{Garcia-Etxebarria:2018ajm,Wang:2020xyo}. Table~\ref{tab:neutrino} lists the charges: $X\equiv1\bmod4$ on every field of the $\mathbf{16}$, and after electroweak breaking $X'\equiv5(B-L)-4Q\bmod4$ stays exact with $X'^2=(-1)^F$. The selection rules follow by charge arithmetic. While $X'$ is exact there is no Majorana mass for $\nu_L$ or $\nu^c$, no Weinberg operator and no $|\Delta L|=2$ amplitude, so neutrinoless double-beta decay vanishes. A type-I seesaw needs a condensate with $B-L=\pm2$ and $Q=0$, a re-Higgsing of the $\nu^c$ sector that SMG with the $\Z_4$ unbroken cannot supply (Proposition~\ref{prop:noseesaw}). $\nu^c$ alone cannot be gapped symmetrically ($\nu=1\bmod16$), so a $\Z_4$-preserving heavy $\nu^c$ mass is a Dirac pairing with an $X'\equiv3$ partner $S$ (a sterile fermion or the SMG composite); the neutral mass matrix of $\{\nu_L,S\}$ against $\{\nu^c\}$ then has eigenvalues $0,\pm\sqrt{m_D^2+M^2}$, one exactly massless neutral fermion per generation, mostly $\nu_L$. These are Wang's $\Z_{4,X}$ scenarios~\cite{Wang:2020xyo,Wang:2020gqr,Wang:2020mra}, the Spin-$\Z_4$ reading of Razamat and Tong~\cite{Razamat:2020kyf,Tong:2021phe} and the $B-L\bmod4$ conservation implied by Ref.~\cite{Garcia-Etxebarria:2018ajm}, reduced to rules: consequences of the exact $\Z_4$, not results of this paper. Extra $X'\equiv1$ singlets would allow further Dirac masses, but still no Majorana mass and no $|\Delta L|=2$. An observed $|\Delta L|=2$ signal excludes an exact $\Z_4$ in the neutrino sector; a null does not establish it.

\begin{table}[!htb]
\caption{$\Z_4$ charges and selection rules of one generation in Standard-Model language (claim N.1, proven). $X=5(B-L)-4Y$ on left-handed Weyl fields; $X'=5(B-L)-4Q$ after electroweak breaking.}
\label{tab:neutrino}
\begin{ruledtabular}
\begin{tabular}{p{0.54\columnwidth}p{0.40\columnwidth}}
item & charge, rule\\
\colrule
$X$ on $Q,u^c,d^c,L,e^c,\nu^c$ & $1,1,-3,-3,1,5$ (all $\equiv1$)\\
$X$ on $H$; $X'$ on $\langle H^0\rangle$ & $-2$; $0$\\
$X'$ on left-handed / right-handed fields & $\equiv3$ / $\equiv1$\\
$\nu_L\nu_L$, $\nu^c\nu^c$, Weinberg operator & $X'\equiv2$: forbidden\\
$\nu_L\nu^c$ (Dirac) & $X'\equiv0$: allowed\\
$|\Delta L|=2$, $\Delta B=\Delta Q=0$ amplitudes & $\Delta X'\equiv2$: zero ($0\nu\beta\beta$)\\
condensate permitting $M_R$ & $B-L=\pm2$, $Q=0$ ($\Delta_R$-like)\\
$\nu^c$ alone & $\nu=1\bmod16$: no symmetric gap\\
$\{\nu_L,S\}$ against $\{\nu^c\}$, $X'(S)\equiv3$ & eigenvalues $0,\pm\sqrt{m_D^2+M^2}$\\
\end{tabular}
\end{ruledtabular}
\end{table}

\paragraph{What $L\le8$ does not determine.} (i) Any exponent: every volume comparison is a $(6,8)$ or $(4,8)$ pair, $h\in[0.25,3]$ is not a scaling window, and the onset is the analytic $h^2$ of a finite gapped box. (ii) Whether the light half at $(2.41,h>0)$ is gapless or has a small symmetric gap in infinite volume: the light ``gap'' is a pair-channel effective mass without a plateau, and $g$ and $\alpha_L$ read at $p_{\min}=\pi/L$, where a gap $\ll\pi/L$ leaves them free-like. (iii) The existence and location of a critical line $y_c(h)$: the direction tests on two volumes are MIXED by their letter, and the data show detuning toward the symmetric side, not a line. (iv) Symmetry breaking at $(3.0,h{=}2)$: on the $(4^4,8^4)$ pair no inequivalent light mass channel grows faster than free, and on like mode sets $\chi_L/$free falls with the volume on both pairs; on its stored mode sets the $(6^4,8^4)$ pair reads growth of SSB size in several channels, the $6^4$ value lying $4$--$8\times$ below both other boxes; all of this at $h=2$ on two volumes; the pole-versus-zero character of the light propagator at $p\to0$ is not measured; a $12^4$ point would settle the volume dependence. (v) The order of the merged transition beyond ``consistent with continuous'': the first-order exclusion and the $\xi_2/L$ crossing need a third volume, and a thin AFM sliver at $\Pc$, ordered only beyond $L=8$, is not excluded.

\paragraph{Open points of the setup.} (vi) The taste split exists on $19.8\,\%$ of the momenta on $6^4$ and $29.6\,\%$ on $6^3\times12$, so every $(6^4,8^4)$ light-sector ratio compares unlike mode sets, and the $(4^4,8^4)$ pair has a $4^4$ box three correlation lengths across and compares a source-free light doublet with a partly dressed one; the restricted-mode control exists at $h=2$ only. (vii) The $8^4$ $\Pc$ anchor chain has 45 measurements. (viii) $\kappa_c=-0.01\pm0.006$ is transferred from BCHH, our lines bracketing it in $(-0.04,+0.02)$ with the upper edge resting on the scored $8^4$ chain at $(+0.02,2.452)$ and the $y=2.41$ pair. (ix) The lattice $\Z_4$ equals the $U(1)_\psi$ $\Z_4$ only up to a Pati--Salam-central Spin(10) element, and $U(1)_\varepsilon$ is not $U(1)_\psi$.

\paragraph{Prior art and what is ours.} The ingredients of the framework are known: the Spin(10) Majorana channels, the anomaly-free $\Z_4$ with $\Z_4^2=(-1)^F$~\cite{Garcia-Etxebarria:2018ajm,Hsieh:2018ifc,Razamat:2020kyf,Tong:2021phe}, SMG as the pairing of an elementary fermion with a three-fermion composite~\cite{Eichten:1985ft,Golterman:1992yha,Golterman:2025boq,Wang:2022ucy,You:2017ltx} and the lattice $SO(10)$ and $16n$ constructions built on it~\cite{Wen:2013ppa,You:2014oaa,Kikukawa:2017ngf,Wang:2018ugf,Wang:2018jkc}, the Pati--Salam reading of reduced staggered fermions~\cite{Butt:2021brl,Catterall:2020fep}, the positivity theorems~\cite{Wu:2005zzb,Wei:2016sgb,Li:2016gte}, the Green's-function-zero picture~\cite{Gurarie:2011,You:2014cka,Xu:2021ztz,Lu:2022qtc,Lu:2023cev,Golterman:2023zqf,Golterman:2026sox}, the three- and four-dimensional staggered SMG studies~\cite{Ayyar:2014eua,Catterall:2015zua,Ayyar:2015lrd,Ayyar:2016lxq,Catterall:2016dzf,Butt:2018nkn,Butt:2021koj,Hasenfratz:2022qan,Butt:2024kxi,Butt:2026lkt}, and the Higgs-language selection rule~\cite{Weinberg:1979sa}. A mass at a critical point is a relevant perturbation, so some light response is generic; the generic argument does not give the sign (toward free), the outcome at $L\le8$ (no symmetry-breaking signal in any inequivalent light mass channel at $\Pc$ on the $(4^4,8^4)$ pair, unlike mass-split systems~\cite{Brower:2015owo,Hasenfratz:2016gut} and the three-dimensional split with an intermediate broken phase~\cite{Maiti:2025ley}), or the shape (a crossover at the lattice bandwidth). Ours are the lattice statement that the discrete remnant of $SU(2)_L$-taste, the shifts and $SO(4)$ forbid every light mass term, the selection rule tested inside an un-Higgsed phase with composite partners present, the flavour-democracy corollary, the $\mathbf{10}/\mathbf{126}$ content of the lattice channels and the null of the two-site channel, the first-order test at the merged point, the composite-takeover measurement, and the half-multiplet explicit-mass study; the partner matrix and the neutrino rules are the frame and the consequence. We found no earlier study of an explicit bilinear on part of an SMG multiplet in any dimension, including the one-dimensional chains~\cite{Fidkowski:2009dba,Fidkowski:2010jmn}, the 3-4-5-0 model~\cite{Zeng:2022grc}, the bilayer honeycomb and Hubbard models~\cite{Slagle:2014vma,He:2016sbs,He:2026oaf}, and the perturbative treatment~\cite{Martin:2025hic}; the closest earlier discussion of a right-handed neutrino in the Eichten--Preskill setting is its decoupling in chiral lattice theories~\cite{Golterman:1989df}, which concerns a singlet that is removed, not an explicit mass on part of the multiplet.

\paragraph{Data and code availability.} The companion repository~\cite{repo} (github.com/PiotrTopa/mass-as-pairing, tag \texttt{v1.0.1}, archived on Zenodo with DOI 10.5281/zenodo.23178329) holds the lattice code, the derived data, every claim as a directory with a self-asserting \texttt{check.py}, and the scripts that produce the tables and figures; \texttt{make check} runs every check. The raw chain archive is deposited on Zenodo (DOI 10.5281/zenodo.23117210); code under the MIT license, data under CC~BY~4.0.

\begin{acknowledgments}
Derivations, code, verification protocols and drafts were prepared with
substantial assistance from the AI system Claude (Anthropic); all results
were machine-verified as described, and the author takes full
responsibility for the content.
\end{acknowledgments}

\appendix

\begin{table*}[!t]
\caption{Stage-1 rows of model N1, all-antiperiodic boxes (chains cut at trajectory 100). $r_L$: light ratio, $t^*$ log-ratio estimator of the stage-1 letter (raw ratio in brackets on $6^3\times12$, where the free normalisation is void); $g=G_L(p_{\min})/$free; $\phi_R/h$ in units of free; $\phi_{T_R}/h$ in units of $10^{-3}$; $\rho=|\phi_{T_R}/\phi_R|$ (lower bound where $|\phi_R|<2\sigma$); $\alpha_L$ [free value]. $\dagger$: partial taste split.}
\label{tab:rows}
\begin{ruledtabular}
\begin{tabular}{lllllllll}
lattice & $y$ & $h$ & $r_L$ & $g$ & $\phi_R/h$ & $\phi_{T_R}/h$ & $\rho$ & $\alpha_L$ [free]\\
\colrule
$8^4$ & 2.41& 0.5 & $1.15\pm0.17$ & $0.770\pm0.029$ & 0.593 & $-0.359\pm0.025$ & 0.012 & $0.211\pm0.011$ [0.098]\\
 & & 1 & $2.30\pm0.23$ & $0.784\pm0.015$ & 0.716 & $-0.214\pm0.013$ & 0.007 & $0.138\pm0.008$ [0.088]\\
 & & 2 & $5.84\pm0.52$ & $0.821\pm0.002$ & 0.844 & $-0.117\pm0.001$ & 0.004 & $0.070\pm0.002$ [0.048]\\
$6^4$ $\dagger$ & 2.41& 0.5 & $1.04\pm0.04$ & $0.765\pm0.009$ & 0.633 & $-0.249\pm0.011$ & 0.009 & $0.245\pm0.009$ [0.185]\\
 & & 1 & $1.24\pm0.04$ & $0.781\pm0.004$ & 0.703 & $-0.183\pm0.005$ & 0.006 & $0.214\pm0.003$ [0.164]\\
 & & 2 & $1.84\pm0.06$ & $0.843\pm0.002$ & 0.833 & $-0.118\pm0.002$ & 0.004 & $0.115\pm0.003$ [0.081]\\
$6^3\times12$ $\dagger$ & 2.41& 0.5 & $1.15\pm0.06$ [1.12] & $0.690\pm0.010$ & 0.625 & $-0.245\pm0.017$ & 0.008 & $0.632\pm0.005$ [0.577]\\
 & & 1 & $1.48\pm0.07$ [1.31] & $0.730\pm0.006$ & 0.712 & $-0.169\pm0.005$ & 0.005 & $0.615\pm0.006$ [0.576]\\
 & & 2 & $2.94\pm0.13$ [1.64] & $0.817\pm0.003$ & 0.840 & $-0.119\pm0.002$ & 0.004 & $0.587\pm0.003$ [0.575]\\
\colrule
$8^4$ & 3.0& 0.5 & $0.989\pm0.005$ & $0.0030\pm0.0001$ & 0.001 & $-1.952\pm0.014$ & $>12.9$ & $0.944\pm0.009$ [0.098]\\
 & & 1 & $0.973\pm0.004$ & $0.0032\pm0.0001$ & 0.004 & $-1.969\pm0.016$ & $11.5\pm3.2$ & $0.936\pm0.007$ [0.088]\\
 & & 2 & $0.879\pm0.004$ & $0.0033\pm0.0001$ & 0.011 & $-1.955\pm0.015$ & $5.1\pm0.6$ & $0.957\pm0.011$ [0.048]\\
$6^4$ $\dagger$ & 3.0& 0.5 & $0.998\pm0.007$ & $0.0126\pm0.0001$ & $-0.001$ & $-1.945\pm0.012$ & $>18.1$ & $0.943\pm0.009$ [0.185]\\
 & & 1 & $0.976\pm0.006$ & $0.0128\pm0.0001$ & 0.005 & $-1.956\pm0.010$ & $9.4\pm1.8$ & $0.931\pm0.011$ [0.164]\\
 & & 2 & $0.884\pm0.005$ & $0.0146\pm0.0002$ & 0.011 & $-1.938\pm0.010$ & $4.9\pm0.4$ & $0.928\pm0.007$ [0.081]\\
$6^3\times12$ $\dagger$ & 3.0& 0.5 & $1.03\pm0.03$ [1.00] & $0.0060\pm0.0001$ & 0.002 & $-1.949\pm0.021$ & $>9.7$ & $0.965\pm0.003$ [0.577]\\
 & & 1 & $1.17\pm0.03$ [1.04] & $0.0060\pm0.0001$ & 0.004 & $-1.932\pm0.016$ & $11.9\pm5.7$ & $0.968\pm0.004$ [0.576]\\
 & & 2 & $2.10\pm0.05$ [1.17] & $0.0070\pm0.0001$ & 0.010 & $-1.945\pm0.016$ & $5.4\pm0.6$ & $0.968\pm0.005$ [0.575]\\
\end{tabular}
\end{ruledtabular}
\end{table*}

\section{Stage-1 rows of model N1}
\label{app:rows}

Table~\ref{tab:rows} lists, per stored stage-1 chain on the all-antiperiodic boxes, the numbers that Secs.~\ref{sec:channel} and \ref{sec:halfmass} quote as ranges: $r_L$ and $g$ (check K5.5), $\phi_R$, $\phi_{T_R}$ and $\rho$ (check K4.3), and $\alpha_L$ with its free reference (check K6.3).

\section{Claim map}
\label{app:claims}

Every machine-checked or measured claim is an assertion in a runnable script; the tolerances quoted in the text are the assertion thresholds, and each script fails loudly rather than reporting. Tables~\ref{tab:claims} and \ref{tab:claims2} list every claim behind Secs.~\ref{sec:model}--\ref{sec:discussion} with its evidence class.

\begin{table*}[!t]
\caption{Claim map, part 1: instruments (I), structure (K1), positivity (K2), the $\Z_4$-odd channel (K3), the composite takeover (K4). Each ID has a self-asserting check script \texttt{claims/<ID>/check.py} in Ref.~\cite{repo}, run by \texttt{make check}. Class: P proven, M machine-checked, X measured.}
\label{tab:claims}
\begin{ruledtabular}
\footnotesize
\begin{tabular}{llp{0.88\textwidth}}
ID & class & statement\\
\colrule
I.1 & M & RHMC simulates the measure exactly: free limit, forces, reversibility, $\delta\tau^2$, exact Metropolis\\
I.2 & M & Hasenbusch split and nested integrator agree with exact Metropolis; sabotage rejected at $6.9\sigma$\\
I.3 & X & calibration to BCHH, $y=\sqrt2\,y_{\rm ref}$; $\kappa_c=-0.01\pm0.006$ transferred, $\Pc=(2.41,-0.01)$; $\kappa$ scan: $\kappa_c\in(-0.03,0.05)$\\
I.4 & M & link-field HMC exact against the all-orders fermion-bag expansion ($Z$ to $9.4\times10^{-15}$)\\
I.5 & M & link-field HMC resonance at $\omega_j\tau\approx n\pi$; trajectory-length jitter removes it\\
I.6 & M & complete $\mathbf{10}/\mathbf{6}$-channel estimator: dense and stochastic agree; fermionic boundary sign\\
I.7 & M & same-parity pairing source: sign-free, exact HMC, linear response\\
I.8 & M & flavour-selective source: $|\mathrm{Pf}|$ RHMC and its observables exact\\
I.9 & M & Pfaffian sign of the flavour-selective source by three routes; the matrix pencil locates the flips\\
I.10 & M & N1 RHMC against exact Metropolis on $4^4$ (20 observables); phase-flipped pattern rejected at $107\sigma$\\
I.11 & M & light/heavy taste observables of N1 agree with dense exact computation\\
K1.1 & M & $\mathrm{Sym}^2(\mathbf{16})=\mathbf{10}\oplus\mathbf{126}$, $\Lambda^2=\mathbf{120}$; ranks of the $\mathbf{126}$ and $\mathbf{10}$ mass matrices\\
K1.2 & M & unique quartic $\Phi_{10}^2$; $\Phi_{126}^2\equiv0$; no charge-6, one charge-8 invariant\\
K1.3 & M & Fierz basis $\{|\Phi_{10}|^2,|\Phi_{126}|^2\}$ of $\psi\psi\bar\psi\bar\psi$; $J_1\!\cdot\!J_1$ decomposed\\
K1.4 & M & $\Z_4$ the unique anomaly-free remnant of $U(1)_\psi$ ($\eta$ invariants; $\psi^4$ vertex)\\
K1.5 & M & Pati--Salam embedding; the on-site $\varepsilon$ channel is half $\mathbf{10}$, half $\mathbf{126}$\\
K1.6 & M & lattice channel algebra and symmetries; $U(1)_\varepsilon\neq U(1)_\psi$; lattice $\Z_4$ = $U(1)_\psi$ $\Z_4$ up to a central element\\
K1.7 & P, M & $\Z_4^2=(-1)^F$; partner-matrix algebra; the composite partner $\bar\psi\bar\psi\psi$ and its two quartic invariants\\
K2.1 & P, M & $\tau_2K$ class: $\det D_c>0$ configuration by configuration; on-site block a quaternion\\
K2.2 & M & pure $\mathrm{Sym}^2$ term: no sign-free route found; its sum-of-squares completion is sign-free\\
K2.3 & M & structure of the completed term; the wedge $|g_6|\le g_{10}$ is its sign-free region; a sign-free same-parity source\\
K2.4 & P, M & the $\tau_2K$ class is quaternion-linear, hence flavour-democratic; class distances of the 1-of-4 and 2+2 splits\\
K2.5 & M & $C_0$ a nodal pairing (corner zero modes at every $h$); $C_\chi$ the taste-chiral Majorana mass of one doublet\\
K2.6 & P, M & fermion-bag expansion; the 2+2 split is positive in the bag and conjugate-pair representations\\
K3.1 & X & $\chi_{10}$ does not grow with $V$ at five wedge points, $L\le8$; the dimer positive control grows\\
K3.2 & X & $\chi_{10}$ at $\Pc$ free-like and suppressed ($\le0.5\times$ free)\\
K3.3 & X & response to the pure $\mathbf{10}$ term flat in $V$: $R(8)/R(6)=1.02\pm0.05$, $0.97\pm0.02$\\
K3.4 & X & no SSB under the same-parity pairing source, $L\le8$; finite, volume-independent susceptibility\\
K3.5 & X & the $U(1)_\varepsilon$-exact family exits through a bond crystal or a gapless fermion, never the $(\mathbf{10},\mathbf3,\mathbf1)$\\
K3.6 & X & columnar bond crystal on the $+$edge of the wedge\\
K4.1 & X & composite takeover SYM $\to\Pc\to$ SMG under a nodal source on one flavour ($4^4$, $6^4$)\\
K4.2 & X & stage 0: SMG screens the elementary $C_\chi$ response; the composite carries it\\
K4.3 & X & takeover with $C_\chi$ on three boxes: composite linear in $h$, lattice-independent to 2\,\%; elementary $\le1.1\,\%$ of free\\
K4.4 & X & production-bc continuity at $\Pc$: $6^4$ pppa against aaaa; $\rho$ bc-dependent, $g$ not\\
\end{tabular}
\end{ruledtabular}
\end{table*}
\begin{table*}[!t]
\caption{Claim map, part 2: model N1 (K5), the corner-block readout (K6), the order of the transition (K7), the critical point under the explicit mass (K8), the moonshot (M), the neutrino rules (N); conventions as in Table~\ref{tab:claims}.}
\label{tab:claims2}
\begin{ruledtabular}
\footnotesize
\begin{tabular}{llp{0.88\textwidth}}
ID & class & statement\\
\colrule
K5.1 & M & N1 instrument: taste projectors, free baselines, the light gate $G_L(p_{\min})=1$\\
K5.2 & P, M & symmetry zeros of the seven flavour-symmetric light masses; the box-shape domain\\
K5.3 & X & stage 0 gates G1--G4 as written, with controls; G2, G3 mis-specified\\
K5.4 & X & stage 1 integrity: cadence, error model, free baselines\\
K5.5 & X & stage 1 by the letter: no verdict (gap floor on the $t^*$ proxy); clause (b) fired at $(3.0,2)$\\
K5.6 & X & control S3: the phase-flipped pattern gaps neither doublet in the free theory\\
K5.7 & X & stage 1b integrity; S7$'$ drifts; no plateau in $C_L(t)$ (E8); box-shape one-point function (E9)\\
K5.8 & X & stage 1b at $y=3.0$: class (a) at $h=1,2$; partial outcome ``growth relative to free'' on (6,8) recorded\\
K5.9 & X & stage 1c integrity; replica consistency on independent seeds ($p\ge0.16$ at $8^4$, $\ge0.07$ at $6^4$)\\
K5.10 & M & $0{+}1$D instance of Prop.~\ref{prop:noseesaw} by exact diagonalisation: exact zeros; pre-registered verdict VOID\\
K5.11 & X & stage 1c: two independent $8^4$ seeds per $h$ reproduce class (a); $e-e_{\rm free}$ on (6,8) pooled\\
K5.12 & M & every light mass-type string vanishes under the full group of the sourced action; killing symmetry per string\\
K5.13 & X & taste coverage $19.8\,\%$ on $6^4$, split on every momentum of $4^4$, $4^3\times8$, $8^4$; on $(4^4,8^4)$ $\chi_L/$free falls with $V$ in SMG\\
K5.14 & M & every stabiliser-invariant light bilinear of the corner algebra enumerated on $4^4$: (A4) fails under the point-group stabiliser (20 invariants on the periodic box, 36 all-antiperiodic; twelve under its $\Z_4$ extension), holds with the shifts (identity, chirality string, eight odd-in-$p$ strings)\\
K5.15 & M, X & light block of the source per box; restricted-mode control at $(3.0,2)$ on $8^4$: $\chi_L/$free falls with $V$ on like mode sets on (4,8) and (6,8); the $(4^4,8^4)$ pair at $\Pc$\\
K5.16 & X & every inequivalent light mass-type channel at $(3.0,2)$ and $(\Pc,2)$ from exact propagators on stored configurations: on $(4^4,8^4)$ no channel grows faster than free; the $(6,8)$ excess is the $6^4$ mode set, channel by channel\\
K6.1 & P, M & no frequency winding of the corner-block propagator in the $\tau_2K$ class; quantisation scale above the bandwidth\\
K6.2 & X & odd-part exponent $\alpha$ validated on free references; SYM and SMG separated per configuration\\
K6.3 & X & $\alpha_L$ on the stage-1 configurations: at its SMG-phase value at every $h$, equal to $\alpha_R$; pole-side at $\Pc$, moving toward free\\
K7.1 & X & no first-order signature at $\Pc$, $L\le8$ (pilot set, 39 chains); the $\kappa$ lines\\
K7.2 & X & $8^4$ with three starts, pooled $\Pc$: $\gamma/\nu_{\rm eff}=1.58\pm0.16$, $\xi_2/L$ crossing $+2.4\sigma$\\
K7.3 & M & $\xi(y)$ at $L\le12$ cannot separate walking from a power law ($\Delta\chi^2\le1.2$)\\
K8.1 & P, X & the bound $\lambda(M_p)\le\Delta$ of the seesaw corner (bookkeeping, not a result); at $\Pc$ the $\sigma$ channel collapses and $m$ rises with $h$; quench at fixed $\sigma$\\
K8.2 & X & tier 2: $m$ and $S(\pi)$ monotone in $h$; no acceptable single power of $h$ on any lattice (T2, T3)\\
K8.3 & X & small-$h$ onset on $6^4$: $p=1.14\pm0.37$, consistent with the symmetry-required $h^2$\\
K8.4 & X & free comparison at $\Pc$: $m<m_{\rm free}$; $g,c\to1$, never above; $m$ measures distance from criticality\\
K8.5 & X & direction toward SYM on $6^4$ in $g$, $c$ ($15\sigma$, $29\sigma$); P1 MIXED by the $S(\pi)$ clause; P3 at $(3.0,3)$ deep SMG\\
K8.6 & M, X & large-$N$ sign of the shift; same sign from the production operator and at finite $\sigma_0$; the $\delta_L$ fit\\
K8.7 & X & $\varepsilon$-vertex position at $\Pc$ with $h$-matched references: $f_{O_4}$ $0.250\to0.056$\\
K8.8 & X & weight of the large-$N$ fit: generic crossovers with $h_0\approx2$ fit as well; $\delta_L$ does not transfer within errors\\
K8.9 & X & $8^4$ direction pair (P1$'$): $g$, $c$ toward SYM; $S(\pi)$ clause at $1.5\sigma$; MIXED as on $6^4$\\
M.1 & X & moonshot trigger not met: $r_L$ rises at $\Pc$; an injected fall triggers it, the free null does not\\
N.1 & P & $\Z_4=5(B-L)-4Y\bmod4$ on matter; neutrino selection rules; one massless neutral state per generation\\
\end{tabular}
\end{ruledtabular}
\end{table*}

\bibliographystyle{apsrev4-2}
\bibliography{refs,repo}

\begin{thebibliography}{82}%
\makeatletter
\providecommand \@ifxundefined [1]{%
 \@ifx{#1\undefined}
}%
\providecommand \@ifnum [1]{%
 \ifnum #1\expandafter \@firstoftwo
 \else \expandafter \@secondoftwo
 \fi
}%
\providecommand \@ifx [1]{%
 \ifx #1\expandafter \@firstoftwo
 \else \expandafter \@secondoftwo
 \fi
}%
\providecommand \natexlab [1]{#1}%
\providecommand \enquote  [1]{``#1''}%
\providecommand \bibnamefont  [1]{#1}%
\providecommand \bibfnamefont [1]{#1}%
\providecommand \citenamefont [1]{#1}%
\providecommand \href@noop [0]{\@secondoftwo}%
\providecommand \href [0]{\begingroup \@sanitize@url \@href}%
\providecommand \@href[1]{\@@startlink{#1}\@@href}%
\providecommand \@@href[1]{\endgroup#1\@@endlink}%
\providecommand \@sanitize@url [0]{\catcode `\\12\catcode `\$12\catcode
  `\&12\catcode `\#12\catcode `\^12\catcode `\_12\catcode `\%12\relax}%
\providecommand \@@startlink[1]{}%
\providecommand \@@endlink[0]{}%
\providecommand \url  [0]{\begingroup\@sanitize@url \@url }%
\providecommand \@url [1]{\endgroup\@href {#1}{\urlprefix }}%
\providecommand \urlprefix  [0]{URL }%
\providecommand \Eprint [0]{\href }%
\providecommand \doibase [0]{https://doi.org/}%
\providecommand \selectlanguage [0]{\@gobble}%
\providecommand \bibinfo  [0]{\@secondoftwo}%
\providecommand \bibfield  [0]{\@secondoftwo}%
\providecommand \translation [1]{[#1]}%
\providecommand \BibitemOpen [0]{}%
\providecommand \bibitemStop [0]{}%
\providecommand \bibitemNoStop [0]{.\EOS\space}%
\providecommand \EOS [0]{\spacefactor3000\relax}%
\providecommand \BibitemShut  [1]{\csname bibitem#1\endcsname}%
\let\auto@bib@innerbib\@empty
\bibitem [{\citenamefont {Minkowski}(1977)}]{Minkowski:1977sc}%
  \BibitemOpen
  \bibfield  {author} {\bibinfo {author} {\bibfnamefont {P.}~\bibnamefont
  {Minkowski}},\ }\href {https://doi.org/10.1016/0370-2693(77)90435-X}
  {\bibfield  {journal} {\bibinfo  {journal} {Phys. Lett. B}\ }\textbf
  {\bibinfo {volume} {67}},\ \bibinfo {pages} {421} (\bibinfo {year}
  {1977})}\BibitemShut {NoStop}%
\bibitem [{\citenamefont {Gell-Mann}\ \emph {et~al.}(1979)\citenamefont
  {Gell-Mann}, \citenamefont {Ramond},\ and\ \citenamefont
  {Slansky}}]{Gell-Mann:1979vob}%
  \BibitemOpen
  \bibfield  {author} {\bibinfo {author} {\bibfnamefont {M.}~\bibnamefont
  {Gell-Mann}}, \bibinfo {author} {\bibfnamefont {P.}~\bibnamefont {Ramond}},\
  and\ \bibinfo {author} {\bibfnamefont {R.}~\bibnamefont {Slansky}},\
  }\href@noop {} {\bibfield  {journal} {\bibinfo  {journal} {Conf. Proc. C}\
  }\textbf {\bibinfo {volume} {790927}},\ \bibinfo {pages} {315} (\bibinfo
  {year} {1979})},\ \Eprint {https://arxiv.org/abs/1306.4669} {arXiv:1306.4669
  [hep-th]} \BibitemShut {NoStop}%
\bibitem [{\citenamefont {Yanagida}(1979)}]{Yanagida:1979as}%
  \BibitemOpen
  \bibfield  {author} {\bibinfo {author} {\bibfnamefont {T.}~\bibnamefont
  {Yanagida}},\ }\href@noop {} {\bibfield  {journal} {\bibinfo  {journal}
  {Conf. Proc. C}\ }\textbf {\bibinfo {volume} {7902131}},\ \bibinfo {pages}
  {95} (\bibinfo {year} {1979})}\BibitemShut {NoStop}%
\bibitem [{\citenamefont {Mohapatra}\ and\ \citenamefont
  {Senjanovic}(1980)}]{Mohapatra:1979ia}%
  \BibitemOpen
  \bibfield  {author} {\bibinfo {author} {\bibfnamefont {R.~N.}\ \bibnamefont
  {Mohapatra}}\ and\ \bibinfo {author} {\bibfnamefont {G.}~\bibnamefont
  {Senjanovic}},\ }\href {https://doi.org/10.1103/PhysRevLett.44.912}
  {\bibfield  {journal} {\bibinfo  {journal} {Phys. Rev. Lett.}\ }\textbf
  {\bibinfo {volume} {44}},\ \bibinfo {pages} {912} (\bibinfo {year}
  {1980})}\BibitemShut {NoStop}%
\bibitem [{\citenamefont {Eichten}\ and\ \citenamefont
  {Preskill}(1986)}]{Eichten:1985ft}%
  \BibitemOpen
  \bibfield  {author} {\bibinfo {author} {\bibfnamefont {E.}~\bibnamefont
  {Eichten}}\ and\ \bibinfo {author} {\bibfnamefont {J.}~\bibnamefont
  {Preskill}},\ }\href {https://doi.org/10.1016/0550-3213(86)90207-5}
  {\bibfield  {journal} {\bibinfo  {journal} {Nucl. Phys. B}\ }\textbf
  {\bibinfo {volume} {268}},\ \bibinfo {pages} {179} (\bibinfo {year}
  {1986})}\BibitemShut {NoStop}%
\bibitem [{\citenamefont {Wang}\ and\ \citenamefont
  {You}(2022)}]{Wang:2022ucy}%
  \BibitemOpen
  \bibfield  {author} {\bibinfo {author} {\bibfnamefont {J.}~\bibnamefont
  {Wang}}\ and\ \bibinfo {author} {\bibfnamefont {Y.-Z.}\ \bibnamefont {You}},\
  }\href {https://doi.org/10.3390/sym14071475} {\bibfield  {journal} {\bibinfo
  {journal} {Symmetry}\ }\textbf {\bibinfo {volume} {14}},\ \bibinfo {pages}
  {1475} (\bibinfo {year} {2022})},\ \Eprint {https://arxiv.org/abs/2204.14271}
  {arXiv:2204.14271 [cond-mat.str-el]} \BibitemShut {NoStop}%
\bibitem [{\citenamefont {Tong}(2022)}]{Tong:2021phe}%
  \BibitemOpen
  \bibfield  {author} {\bibinfo {author} {\bibfnamefont {D.}~\bibnamefont
  {Tong}},\ }\href {https://doi.org/10.1007/JHEP07(2022)001} {\bibfield
  {journal} {\bibinfo  {journal} {JHEP}\ }\textbf {\bibinfo {volume} {07}},\
  \bibinfo {pages} {001}},\ \Eprint {https://arxiv.org/abs/2104.03997}
  {arXiv:2104.03997 [hep-th]} \BibitemShut {NoStop}%
\bibitem [{\citenamefont {Razamat}\ and\ \citenamefont
  {Tong}(2021)}]{Razamat:2020kyf}%
  \BibitemOpen
  \bibfield  {author} {\bibinfo {author} {\bibfnamefont {S.~S.}\ \bibnamefont
  {Razamat}}\ and\ \bibinfo {author} {\bibfnamefont {D.}~\bibnamefont {Tong}},\
  }\href {https://doi.org/10.1103/PhysRevX.11.011063} {\bibfield  {journal}
  {\bibinfo  {journal} {Phys. Rev. X}\ }\textbf {\bibinfo {volume} {11}},\
  \bibinfo {pages} {011063} (\bibinfo {year} {2021})},\ \Eprint
  {https://arxiv.org/abs/2009.05037} {arXiv:2009.05037 [hep-th]} \BibitemShut
  {NoStop}%
\bibitem [{\citenamefont {Hasenfratz}\ and\ \citenamefont
  {Witzel}(2026)}]{Hasenfratz:2025lti}%
  \BibitemOpen
  \bibfield  {author} {\bibinfo {author} {\bibfnamefont {A.}~\bibnamefont
  {Hasenfratz}}\ and\ \bibinfo {author} {\bibfnamefont {O.}~\bibnamefont
  {Witzel}} (\bibinfo {collaboration} {Lattice Strong Dynamics}),\ }\href
  {https://doi.org/10.22323/1.485.0467} {\bibfield  {journal} {\bibinfo
  {journal} {PoS}\ }\textbf {\bibinfo {volume} {EPS-HEP2025}},\ \bibinfo
  {pages} {467} (\bibinfo {year} {2026})},\ \Eprint
  {https://arxiv.org/abs/2511.22678} {arXiv:2511.22678 [hep-lat]} \BibitemShut
  {NoStop}%
\bibitem [{\citenamefont {Hasenfratz}\ and\ \citenamefont
  {Xu}(2026)}]{Hasenfratz:2026orf}%
  \BibitemOpen
  \bibfield  {author} {\bibinfo {author} {\bibfnamefont {A.}~\bibnamefont
  {Hasenfratz}}\ and\ \bibinfo {author} {\bibfnamefont {C.}~\bibnamefont
  {Xu}},\ }\href {https://doi.org/10.1016/j.aop.2026.170730} {\bibfield
  {journal} {\bibinfo  {journal} {Annals Phys.}\ }\textbf {\bibinfo {volume}
  {495}},\ \bibinfo {pages} {170730} (\bibinfo {year} {2026})},\ \Eprint
  {https://arxiv.org/abs/2604.02424} {arXiv:2604.02424 [hep-lat]} \BibitemShut
  {NoStop}%
\bibitem [{\citenamefont {Golterman}\ \emph {et~al.}(1993)\citenamefont
  {Golterman}, \citenamefont {Petcher},\ and\ \citenamefont
  {Rivas}}]{Golterman:1992yha}%
  \BibitemOpen
  \bibfield  {author} {\bibinfo {author} {\bibfnamefont {M.~F.~L.}\
  \bibnamefont {Golterman}}, \bibinfo {author} {\bibfnamefont {D.~N.}\
  \bibnamefont {Petcher}},\ and\ \bibinfo {author} {\bibfnamefont
  {E.}~\bibnamefont {Rivas}},\ }\href
  {https://doi.org/10.1016/0550-3213(93)90049-U} {\bibfield  {journal}
  {\bibinfo  {journal} {Nucl. Phys. B}\ }\textbf {\bibinfo {volume} {395}},\
  \bibinfo {pages} {596} (\bibinfo {year} {1993})},\ \Eprint
  {https://arxiv.org/abs/hep-lat/9206010} {arXiv:hep-lat/9206010} \BibitemShut
  {NoStop}%
\bibitem [{\citenamefont {Golterman}\ and\ \citenamefont
  {Shamir}(2024)}]{Golterman:2023zqf}%
  \BibitemOpen
  \bibfield  {author} {\bibinfo {author} {\bibfnamefont {M.}~\bibnamefont
  {Golterman}}\ and\ \bibinfo {author} {\bibfnamefont {Y.}~\bibnamefont
  {Shamir}},\ }\href {https://doi.org/10.1103/PhysRevLett.132.081903}
  {\bibfield  {journal} {\bibinfo  {journal} {Phys. Rev. Lett.}\ }\textbf
  {\bibinfo {volume} {132}},\ \bibinfo {pages} {081903} (\bibinfo {year}
  {2024})},\ \Eprint {https://arxiv.org/abs/2311.12790} {arXiv:2311.12790
  [hep-lat]} \BibitemShut {NoStop}%
\bibitem [{\citenamefont {Golterman}\ and\ \citenamefont
  {Shamir}(2026{\natexlab{a}})}]{Golterman:2025boq}%
  \BibitemOpen
  \bibfield  {author} {\bibinfo {author} {\bibfnamefont {M.}~\bibnamefont
  {Golterman}}\ and\ \bibinfo {author} {\bibfnamefont {Y.}~\bibnamefont
  {Shamir}},\ }\href {https://doi.org/10.1103/qd1t-wzy1} {\bibfield  {journal}
  {\bibinfo  {journal} {Phys. Rev. D}\ }\textbf {\bibinfo {volume} {113}},\
  \bibinfo {pages} {014503} (\bibinfo {year} {2026}{\natexlab{a}})},\ \Eprint
  {https://arxiv.org/abs/2505.20436} {arXiv:2505.20436 [hep-lat]} \BibitemShut
  {NoStop}%
\bibitem [{\citenamefont {You}\ \emph {et~al.}(2018)\citenamefont {You},
  \citenamefont {He}, \citenamefont {Xu},\ and\ \citenamefont
  {Vishwanath}}]{You:2017ltx}%
  \BibitemOpen
  \bibfield  {author} {\bibinfo {author} {\bibfnamefont {Y.-Z.}\ \bibnamefont
  {You}}, \bibinfo {author} {\bibfnamefont {Y.-C.}\ \bibnamefont {He}},
  \bibinfo {author} {\bibfnamefont {C.}~\bibnamefont {Xu}},\ and\ \bibinfo
  {author} {\bibfnamefont {A.}~\bibnamefont {Vishwanath}},\ }\href
  {https://doi.org/10.1103/PhysRevX.8.011026} {\bibfield  {journal} {\bibinfo
  {journal} {Phys. Rev. X}\ }\textbf {\bibinfo {volume} {8}},\ \bibinfo {pages}
  {011026} (\bibinfo {year} {2018})},\ \Eprint
  {https://arxiv.org/abs/1705.09313} {arXiv:1705.09313 [cond-mat.str-el]}
  \BibitemShut {NoStop}%
\bibitem [{\citenamefont {Garc{\'\i}a-Etxebarria}\ and\ \citenamefont
  {Montero}(2019)}]{Garcia-Etxebarria:2018ajm}%
  \BibitemOpen
  \bibfield  {author} {\bibinfo {author} {\bibfnamefont {I.}~\bibnamefont
  {Garc{\'\i}a-Etxebarria}}\ and\ \bibinfo {author} {\bibfnamefont
  {M.}~\bibnamefont {Montero}},\ }\href
  {https://doi.org/10.1007/JHEP08(2019)003} {\bibfield  {journal} {\bibinfo
  {journal} {JHEP}\ }\textbf {\bibinfo {volume} {08}},\ \bibinfo {pages}
  {003}},\ \Eprint {https://arxiv.org/abs/1808.00009} {arXiv:1808.00009
  [hep-th]} \BibitemShut {NoStop}%
\bibitem [{\citenamefont {Hsieh}(2018)}]{Hsieh:2018ifc}%
  \BibitemOpen
  \bibfield  {author} {\bibinfo {author} {\bibfnamefont {C.-T.}\ \bibnamefont
  {Hsieh}},\ }\href@noop {} {\bibinfo {title} {{Discrete gauge anomalies
  revisited}}} (\bibinfo {year} {2018}),\ \Eprint
  {https://arxiv.org/abs/1808.02881} {arXiv:1808.02881 [hep-th]} \BibitemShut
  {NoStop}%
\bibitem [{\citenamefont {Davighi}\ \emph {et~al.}(2020)\citenamefont
  {Davighi}, \citenamefont {Gripaios},\ and\ \citenamefont
  {Lohitsiri}}]{Davighi:2019rcd}%
  \BibitemOpen
  \bibfield  {author} {\bibinfo {author} {\bibfnamefont {J.}~\bibnamefont
  {Davighi}}, \bibinfo {author} {\bibfnamefont {B.}~\bibnamefont {Gripaios}},\
  and\ \bibinfo {author} {\bibfnamefont {N.}~\bibnamefont {Lohitsiri}},\ }\href
  {https://doi.org/10.1007/JHEP07(2020)232} {\bibfield  {journal} {\bibinfo
  {journal} {JHEP}\ }\textbf {\bibinfo {volume} {07}},\ \bibinfo {pages}
  {232}},\ \Eprint {https://arxiv.org/abs/1910.11277} {arXiv:1910.11277
  [hep-th]} \BibitemShut {NoStop}%
\bibitem [{\citenamefont {Wan}\ and\ \citenamefont {Wang}(2020)}]{Wan:2019gqr}%
  \BibitemOpen
  \bibfield  {author} {\bibinfo {author} {\bibfnamefont {Z.}~\bibnamefont
  {Wan}}\ and\ \bibinfo {author} {\bibfnamefont {J.}~\bibnamefont {Wang}},\
  }\href {https://doi.org/10.1007/JHEP07(2020)062} {\bibfield  {journal}
  {\bibinfo  {journal} {JHEP}\ }\textbf {\bibinfo {volume} {07}},\ \bibinfo
  {pages} {062}},\ \Eprint {https://arxiv.org/abs/1910.14668} {arXiv:1910.14668
  [hep-th]} \BibitemShut {NoStop}%
\bibitem [{\citenamefont {Wang}(2020{\natexlab{a}})}]{Wang:2020xyo}%
  \BibitemOpen
  \bibfield  {author} {\bibinfo {author} {\bibfnamefont {J.}~\bibnamefont
  {Wang}},\ }\href@noop {} {\bibinfo {title} {{Anomaly and Cobordism
  Constraints Beyond the Standard Model: Topological Force}}} (\bibinfo {year}
  {2020}{\natexlab{a}}),\ \Eprint {https://arxiv.org/abs/2006.16996}
  {arXiv:2006.16996 [hep-th]} \BibitemShut {NoStop}%
\bibitem [{\citenamefont {Weinberg}(1979)}]{Weinberg:1979sa}%
  \BibitemOpen
  \bibfield  {author} {\bibinfo {author} {\bibfnamefont {S.}~\bibnamefont
  {Weinberg}},\ }\href {https://doi.org/10.1103/PhysRevLett.43.1566} {\bibfield
   {journal} {\bibinfo  {journal} {Phys. Rev. Lett.}\ }\textbf {\bibinfo
  {volume} {43}},\ \bibinfo {pages} {1566} (\bibinfo {year}
  {1979})}\BibitemShut {NoStop}%
\bibitem [{\citenamefont {Wang}(2020{\natexlab{b}})}]{Wang:2020gqr}%
  \BibitemOpen
  \bibfield  {author} {\bibinfo {author} {\bibfnamefont {J.}~\bibnamefont
  {Wang}},\ }\href@noop {} {\bibinfo {title} {{Anomaly and Cobordism
  Constraints Beyond Grand Unification: Energy Hierarchy}}} (\bibinfo {year}
  {2020}{\natexlab{b}}),\ \Eprint {https://arxiv.org/abs/2008.06499}
  {arXiv:2008.06499 [hep-th]} \BibitemShut {NoStop}%
\bibitem [{\citenamefont {Butt}\ \emph {et~al.}(2026)\citenamefont {Butt},
  \citenamefont {Catterall}, \citenamefont {Hartshaw},\ and\ \citenamefont
  {Hasenfratz}}]{Butt:2026lkt}%
  \BibitemOpen
  \bibfield  {author} {\bibinfo {author} {\bibfnamefont {N.}~\bibnamefont
  {Butt}}, \bibinfo {author} {\bibfnamefont {S.}~\bibnamefont {Catterall}},
  \bibinfo {author} {\bibfnamefont {G.}~\bibnamefont {Hartshaw}},\ and\
  \bibinfo {author} {\bibfnamefont {A.}~\bibnamefont {Hasenfratz}},\
  }\href@noop {} {\bibinfo {title} {{Searching for symmetric mass generation
  with staggered fermions in four dimensions}}} (\bibinfo {year} {2026}),\
  \Eprint {https://arxiv.org/abs/2608.18239} {arXiv:2608.18239 [hep-lat]}
  \BibitemShut {NoStop}%
\bibitem [{\citenamefont {Pati}\ and\ \citenamefont
  {Salam}(1974)}]{Pati:1974yy}%
  \BibitemOpen
  \bibfield  {author} {\bibinfo {author} {\bibfnamefont {J.~C.}\ \bibnamefont
  {Pati}}\ and\ \bibinfo {author} {\bibfnamefont {A.}~\bibnamefont {Salam}},\
  }\href {https://doi.org/10.1103/PhysRevD.10.275} {\bibfield  {journal}
  {\bibinfo  {journal} {Phys. Rev. D}\ }\textbf {\bibinfo {volume} {10}},\
  \bibinfo {pages} {275} (\bibinfo {year} {1974})},\ \bibinfo {note} {[Erratum:
  Phys.Rev.D 11, 703--703 (1975)]}\BibitemShut {NoStop}%
\bibitem [{\citenamefont {Butt}\ \emph
  {et~al.}(2021{\natexlab{a}})\citenamefont {Butt}, \citenamefont {Catterall},
  \citenamefont {Pradhan},\ and\ \citenamefont {Toga}}]{Butt:2021brl}%
  \BibitemOpen
  \bibfield  {author} {\bibinfo {author} {\bibfnamefont {N.}~\bibnamefont
  {Butt}}, \bibinfo {author} {\bibfnamefont {S.}~\bibnamefont {Catterall}},
  \bibinfo {author} {\bibfnamefont {A.}~\bibnamefont {Pradhan}},\ and\ \bibinfo
  {author} {\bibfnamefont {G.~C.}\ \bibnamefont {Toga}},\ }\href
  {https://doi.org/10.1103/PhysRevD.104.094504} {\bibfield  {journal} {\bibinfo
   {journal} {Phys. Rev. D}\ }\textbf {\bibinfo {volume} {104}},\ \bibinfo
  {pages} {094504} (\bibinfo {year} {2021}{\natexlab{a}})},\ \Eprint
  {https://arxiv.org/abs/2101.01026} {arXiv:2101.01026 [hep-th]} \BibitemShut
  {NoStop}%
\bibitem [{\citenamefont {Catterall}(2021)}]{Catterall:2020fep}%
  \BibitemOpen
  \bibfield  {author} {\bibinfo {author} {\bibfnamefont {S.}~\bibnamefont
  {Catterall}},\ }\href {https://doi.org/10.1103/PhysRevD.104.014503}
  {\bibfield  {journal} {\bibinfo  {journal} {Phys. Rev. D}\ }\textbf {\bibinfo
  {volume} {104}},\ \bibinfo {pages} {014503} (\bibinfo {year} {2021})},\
  \Eprint {https://arxiv.org/abs/2010.02290} {arXiv:2010.02290 [hep-lat]}
  \BibitemShut {NoStop}%
\bibitem [{\citenamefont {Topa}(2026)}]{repo}%
  \BibitemOpen
  \bibfield  {author} {\bibinfo {author} {\bibfnamefont {P.~S.}\ \bibnamefont
  {Topa}},\ }\href@noop {} {\bibinfo {title} {{mass-as-pairing: companion code,
  derived data and machine-checked claims for this paper}}},\ \bibinfo
  {howpublished} {\url{https://github.com/PiotrTopa/mass-as-pairing}, tag
  \texttt{v1.0.1}, doi:10.5281/zenodo.23178329; raw chain archive Zenodo,
  doi:10.5281/zenodo.23117210} (\bibinfo {year} {2026})\BibitemShut {NoStop}%
\bibitem [{\citenamefont {Babu}\ and\ \citenamefont
  {Mohapatra}(1993)}]{Babu:1992ia}%
  \BibitemOpen
  \bibfield  {author} {\bibinfo {author} {\bibfnamefont {K.~S.}\ \bibnamefont
  {Babu}}\ and\ \bibinfo {author} {\bibfnamefont {R.~N.}\ \bibnamefont
  {Mohapatra}},\ }\href {https://doi.org/10.1103/PhysRevLett.70.2845}
  {\bibfield  {journal} {\bibinfo  {journal} {Phys. Rev. Lett.}\ }\textbf
  {\bibinfo {volume} {70}},\ \bibinfo {pages} {2845} (\bibinfo {year}
  {1993})},\ \Eprint {https://arxiv.org/abs/hep-ph/9209215}
  {arXiv:hep-ph/9209215} \BibitemShut {NoStop}%
\bibitem [{\citenamefont {Slansky}(1981)}]{Slansky:1981yr}%
  \BibitemOpen
  \bibfield  {author} {\bibinfo {author} {\bibfnamefont {R.}~\bibnamefont
  {Slansky}},\ }\href {https://doi.org/10.1016/0370-1573(81)90092-2} {\bibfield
   {journal} {\bibinfo  {journal} {Phys. Rept.}\ }\textbf {\bibinfo {volume}
  {79}},\ \bibinfo {pages} {1} (\bibinfo {year} {1981})}\BibitemShut {NoStop}%
\bibitem [{\citenamefont {Georgi}(1975)}]{Georgi:1974my}%
  \BibitemOpen
  \bibfield  {author} {\bibinfo {author} {\bibfnamefont {H.}~\bibnamefont
  {Georgi}},\ }\href {https://doi.org/10.1063/1.2947450} {\bibfield  {journal}
  {\bibinfo  {journal} {AIP Conf. Proc.}\ }\textbf {\bibinfo {volume} {23}},\
  \bibinfo {pages} {575} (\bibinfo {year} {1975})}\BibitemShut {NoStop}%
\bibitem [{\citenamefont {Fritzsch}\ and\ \citenamefont
  {Minkowski}(1975)}]{Fritzsch:1974nn}%
  \BibitemOpen
  \bibfield  {author} {\bibinfo {author} {\bibfnamefont {H.}~\bibnamefont
  {Fritzsch}}\ and\ \bibinfo {author} {\bibfnamefont {P.}~\bibnamefont
  {Minkowski}},\ }\href {https://doi.org/10.1016/0003-4916(75)90211-0}
  {\bibfield  {journal} {\bibinfo  {journal} {Annals Phys.}\ }\textbf {\bibinfo
  {volume} {93}},\ \bibinfo {pages} {193} (\bibinfo {year} {1975})}\BibitemShut
  {NoStop}%
\bibitem [{\citenamefont {Ibanez}\ and\ \citenamefont
  {Ross}(1991)}]{Ibanez:1991hv}%
  \BibitemOpen
  \bibfield  {author} {\bibinfo {author} {\bibfnamefont {L.~E.}\ \bibnamefont
  {Ibanez}}\ and\ \bibinfo {author} {\bibfnamefont {G.~G.}\ \bibnamefont
  {Ross}},\ }\href {https://doi.org/10.1016/0370-2693(91)91614-2} {\bibfield
  {journal} {\bibinfo  {journal} {Phys. Lett. B}\ }\textbf {\bibinfo {volume}
  {260}},\ \bibinfo {pages} {291} (\bibinfo {year} {1991})}\BibitemShut
  {NoStop}%
\bibitem [{\citenamefont {Chandrasekharan}(2010)}]{Chandrasekharan:2009wc}%
  \BibitemOpen
  \bibfield  {author} {\bibinfo {author} {\bibfnamefont {S.}~\bibnamefont
  {Chandrasekharan}},\ }\href {https://doi.org/10.1103/PhysRevD.82.025007}
  {\bibfield  {journal} {\bibinfo  {journal} {Phys. Rev. D}\ }\textbf {\bibinfo
  {volume} {82}},\ \bibinfo {pages} {025007} (\bibinfo {year} {2010})},\
  \Eprint {https://arxiv.org/abs/0910.5736} {arXiv:0910.5736 [hep-lat]}
  \BibitemShut {NoStop}%
\bibitem [{\citenamefont {Chandrasekharan}(2013)}]{Chandrasekharan:2013rpa}%
  \BibitemOpen
  \bibfield  {author} {\bibinfo {author} {\bibfnamefont {S.}~\bibnamefont
  {Chandrasekharan}},\ }\href {https://doi.org/10.1140/epja/i2013-13090-y}
  {\bibfield  {journal} {\bibinfo  {journal} {Eur. Phys. J. A}\ }\textbf
  {\bibinfo {volume} {49}},\ \bibinfo {pages} {90} (\bibinfo {year} {2013})},\
  \Eprint {https://arxiv.org/abs/1304.4900} {arXiv:1304.4900 [hep-lat]}
  \BibitemShut {NoStop}%
\bibitem [{\citenamefont {Wu}\ and\ \citenamefont {Zhang}(2005)}]{Wu:2005zzb}%
  \BibitemOpen
  \bibfield  {author} {\bibinfo {author} {\bibfnamefont {C.}~\bibnamefont
  {Wu}}\ and\ \bibinfo {author} {\bibfnamefont {S.-C.}\ \bibnamefont {Zhang}},\
  }\href {https://doi.org/10.1103/PhysRevB.71.155115} {\bibfield  {journal}
  {\bibinfo  {journal} {Phys. Rev. B}\ }\textbf {\bibinfo {volume} {71}},\
  \bibinfo {pages} {155115} (\bibinfo {year} {2005})},\ \Eprint
  {https://arxiv.org/abs/cond-mat/0407272} {arXiv:cond-mat/0407272}
  \BibitemShut {NoStop}%
\bibitem [{\citenamefont {Wei}\ \emph {et~al.}(2016)\citenamefont {Wei},
  \citenamefont {Wu}, \citenamefont {Li}, \citenamefont {Zhang},\ and\
  \citenamefont {Xiang}}]{Wei:2016sgb}%
  \BibitemOpen
  \bibfield  {author} {\bibinfo {author} {\bibfnamefont {Z.~C.}\ \bibnamefont
  {Wei}}, \bibinfo {author} {\bibfnamefont {C.}~\bibnamefont {Wu}}, \bibinfo
  {author} {\bibfnamefont {Y.}~\bibnamefont {Li}}, \bibinfo {author}
  {\bibfnamefont {S.}~\bibnamefont {Zhang}},\ and\ \bibinfo {author}
  {\bibfnamefont {T.}~\bibnamefont {Xiang}},\ }\href
  {https://doi.org/10.1103/PhysRevLett.116.250601} {\bibfield  {journal}
  {\bibinfo  {journal} {Phys. Rev. Lett.}\ }\textbf {\bibinfo {volume} {116}},\
  \bibinfo {pages} {250601} (\bibinfo {year} {2016})},\ \Eprint
  {https://arxiv.org/abs/1601.01994} {arXiv:1601.01994 [cond-mat.str-el]}
  \BibitemShut {NoStop}%
\bibitem [{\citenamefont {Li}\ \emph {et~al.}(2016)\citenamefont {Li},
  \citenamefont {Jiang},\ and\ \citenamefont {Yao}}]{Li:2016gte}%
  \BibitemOpen
  \bibfield  {author} {\bibinfo {author} {\bibfnamefont {Z.-X.}\ \bibnamefont
  {Li}}, \bibinfo {author} {\bibfnamefont {Y.-F.}\ \bibnamefont {Jiang}},\ and\
  \bibinfo {author} {\bibfnamefont {H.}~\bibnamefont {Yao}},\ }\href
  {https://doi.org/10.1103/PhysRevLett.117.267002} {\bibfield  {journal}
  {\bibinfo  {journal} {Phys. Rev. Lett.}\ }\textbf {\bibinfo {volume} {117}},\
  \bibinfo {pages} {267002} (\bibinfo {year} {2016})},\ \Eprint
  {https://arxiv.org/abs/1601.05780} {arXiv:1601.05780 [cond-mat.str-el]}
  \BibitemShut {NoStop}%
\bibitem [{\citenamefont {Catterall}(2016)}]{Catterall:2015zua}%
  \BibitemOpen
  \bibfield  {author} {\bibinfo {author} {\bibfnamefont {S.}~\bibnamefont
  {Catterall}},\ }\href {https://doi.org/10.1007/JHEP01(2016)121} {\bibfield
  {journal} {\bibinfo  {journal} {JHEP}\ }\textbf {\bibinfo {volume} {01}},\
  \bibinfo {pages} {121}},\ \Eprint {https://arxiv.org/abs/1510.04153}
  {arXiv:1510.04153 [hep-lat]} \BibitemShut {NoStop}%
\bibitem [{\citenamefont {Golterman}\ and\ \citenamefont
  {Smit}(1984)}]{Golterman:1984cy}%
  \BibitemOpen
  \bibfield  {author} {\bibinfo {author} {\bibfnamefont {M.~F.~L.}\
  \bibnamefont {Golterman}}\ and\ \bibinfo {author} {\bibfnamefont
  {J.}~\bibnamefont {Smit}},\ }\href
  {https://doi.org/10.1016/0550-3213(84)90424-3} {\bibfield  {journal}
  {\bibinfo  {journal} {Nucl. Phys. B}\ }\textbf {\bibinfo {volume} {245}},\
  \bibinfo {pages} {61} (\bibinfo {year} {1984})}\BibitemShut {NoStop}%
\bibitem [{\citenamefont {Kilcup}\ and\ \citenamefont
  {Sharpe}(1987)}]{Kilcup:1986dg}%
  \BibitemOpen
  \bibfield  {author} {\bibinfo {author} {\bibfnamefont {G.~W.}\ \bibnamefont
  {Kilcup}}\ and\ \bibinfo {author} {\bibfnamefont {S.~R.}\ \bibnamefont
  {Sharpe}},\ }\href {https://doi.org/10.1016/0550-3213(87)90285-9} {\bibfield
  {journal} {\bibinfo  {journal} {Nucl. Phys. B}\ }\textbf {\bibinfo {volume}
  {283}},\ \bibinfo {pages} {493} (\bibinfo {year} {1987})}\BibitemShut
  {NoStop}%
\bibitem [{\citenamefont {Fidkowski}\ and\ \citenamefont
  {Kitaev}(2010)}]{Fidkowski:2009dba}%
  \BibitemOpen
  \bibfield  {author} {\bibinfo {author} {\bibfnamefont {L.}~\bibnamefont
  {Fidkowski}}\ and\ \bibinfo {author} {\bibfnamefont {A.}~\bibnamefont
  {Kitaev}},\ }\href {https://doi.org/10.1103/PhysRevB.81.134509} {\bibfield
  {journal} {\bibinfo  {journal} {Phys. Rev. B}\ }\textbf {\bibinfo {volume}
  {81}},\ \bibinfo {pages} {134509} (\bibinfo {year} {2010})},\ \Eprint
  {https://arxiv.org/abs/0904.2197} {arXiv:0904.2197 [cond-mat.str-el]}
  \BibitemShut {NoStop}%
\bibitem [{\citenamefont {Fidkowski}\ and\ \citenamefont
  {Kitaev}(2011)}]{Fidkowski:2010jmn}%
  \BibitemOpen
  \bibfield  {author} {\bibinfo {author} {\bibfnamefont {L.}~\bibnamefont
  {Fidkowski}}\ and\ \bibinfo {author} {\bibfnamefont {A.}~\bibnamefont
  {Kitaev}},\ }\href {https://doi.org/10.1103/PhysRevB.83.075103} {\bibfield
  {journal} {\bibinfo  {journal} {Phys. Rev. B}\ }\textbf {\bibinfo {volume}
  {83}},\ \bibinfo {pages} {075103} (\bibinfo {year} {2011})},\ \Eprint
  {https://arxiv.org/abs/1008.4138} {arXiv:1008.4138 [cond-mat.str-el]}
  \BibitemShut {NoStop}%
\bibitem [{\citenamefont {Tong}\ and\ \citenamefont
  {Turner}(2020)}]{Tong:2019bbk}%
  \BibitemOpen
  \bibfield  {author} {\bibinfo {author} {\bibfnamefont {D.}~\bibnamefont
  {Tong}}\ and\ \bibinfo {author} {\bibfnamefont {C.}~\bibnamefont {Turner}},\
  }\href {https://doi.org/10.21468/SciPostPhysLectNotes.14} {\bibfield
  {journal} {\bibinfo  {journal} {SciPost Phys. Lect. Notes}\ }\textbf
  {\bibinfo {volume} {14}},\ \bibinfo {pages} {1} (\bibinfo {year} {2020})},\
  \Eprint {https://arxiv.org/abs/1906.07199} {arXiv:1906.07199 [hep-th]}
  \BibitemShut {NoStop}%
\bibitem [{\citenamefont {Catterall}(2023)}]{Catterall:2022jky}%
  \BibitemOpen
  \bibfield  {author} {\bibinfo {author} {\bibfnamefont {S.}~\bibnamefont
  {Catterall}},\ }\href {https://doi.org/10.1103/PhysRevD.107.014501}
  {\bibfield  {journal} {\bibinfo  {journal} {Phys. Rev. D}\ }\textbf {\bibinfo
  {volume} {107}},\ \bibinfo {pages} {014501} (\bibinfo {year} {2023})},\
  \Eprint {https://arxiv.org/abs/2209.03828} {arXiv:2209.03828 [hep-lat]}
  \BibitemShut {NoStop}%
\bibitem [{\citenamefont {Butt}\ \emph {et~al.}(2018)\citenamefont {Butt},
  \citenamefont {Catterall},\ and\ \citenamefont {Schaich}}]{Butt:2018nkn}%
  \BibitemOpen
  \bibfield  {author} {\bibinfo {author} {\bibfnamefont {N.}~\bibnamefont
  {Butt}}, \bibinfo {author} {\bibfnamefont {S.}~\bibnamefont {Catterall}},\
  and\ \bibinfo {author} {\bibfnamefont {D.}~\bibnamefont {Schaich}},\ }\href
  {https://doi.org/10.1103/PhysRevD.98.114514} {\bibfield  {journal} {\bibinfo
  {journal} {Phys. Rev. D}\ }\textbf {\bibinfo {volume} {98}},\ \bibinfo
  {pages} {114514} (\bibinfo {year} {2018})},\ \Eprint
  {https://arxiv.org/abs/1810.06117} {arXiv:1810.06117 [hep-lat]} \BibitemShut
  {NoStop}%
\bibitem [{\citenamefont {Catterall}\ and\ \citenamefont
  {Schaich}(2017)}]{Catterall:2016dzf}%
  \BibitemOpen
  \bibfield  {author} {\bibinfo {author} {\bibfnamefont {S.}~\bibnamefont
  {Catterall}}\ and\ \bibinfo {author} {\bibfnamefont {D.}~\bibnamefont
  {Schaich}},\ }\href {https://doi.org/10.1103/PhysRevD.96.034506} {\bibfield
  {journal} {\bibinfo  {journal} {Phys. Rev. D}\ }\textbf {\bibinfo {volume}
  {96}},\ \bibinfo {pages} {034506} (\bibinfo {year} {2017})},\ \Eprint
  {https://arxiv.org/abs/1609.08541} {arXiv:1609.08541 [hep-lat]} \BibitemShut
  {NoStop}%
\bibitem [{\citenamefont {Schaich}\ and\ \citenamefont
  {Catterall}(2018)}]{Schaich:2017czc}%
  \BibitemOpen
  \bibfield  {author} {\bibinfo {author} {\bibfnamefont {D.}~\bibnamefont
  {Schaich}}\ and\ \bibinfo {author} {\bibfnamefont {S.}~\bibnamefont
  {Catterall}},\ }\href {https://doi.org/10.1051/epjconf/201817503004}
  {\bibfield  {journal} {\bibinfo  {journal} {EPJ Web Conf.}\ }\textbf
  {\bibinfo {volume} {175}},\ \bibinfo {pages} {03004} (\bibinfo {year}
  {2018})},\ \Eprint {https://arxiv.org/abs/1710.08137} {arXiv:1710.08137
  [hep-lat]} \BibitemShut {NoStop}%
\bibitem [{\citenamefont {Ayyar}\ and\ \citenamefont
  {Chandrasekharan}(2016{\natexlab{a}})}]{Ayyar:2016lxq}%
  \BibitemOpen
  \bibfield  {author} {\bibinfo {author} {\bibfnamefont {V.}~\bibnamefont
  {Ayyar}}\ and\ \bibinfo {author} {\bibfnamefont {S.}~\bibnamefont
  {Chandrasekharan}},\ }\href {https://doi.org/10.1007/JHEP10(2016)058}
  {\bibfield  {journal} {\bibinfo  {journal} {JHEP}\ }\textbf {\bibinfo
  {volume} {10}},\ \bibinfo {pages} {058}},\ \Eprint
  {https://arxiv.org/abs/1606.06312} {arXiv:1606.06312 [hep-lat]} \BibitemShut
  {NoStop}%
\bibitem [{\citenamefont {Catterall}\ and\ \citenamefont
  {Butt}(2018)}]{Catterall:2017ogi}%
  \BibitemOpen
  \bibfield  {author} {\bibinfo {author} {\bibfnamefont {S.}~\bibnamefont
  {Catterall}}\ and\ \bibinfo {author} {\bibfnamefont {N.}~\bibnamefont
  {Butt}},\ }\href {https://doi.org/10.1103/PhysRevD.97.094502} {\bibfield
  {journal} {\bibinfo  {journal} {Phys. Rev. D}\ }\textbf {\bibinfo {volume}
  {97}},\ \bibinfo {pages} {094502} (\bibinfo {year} {2018})},\ \Eprint
  {https://arxiv.org/abs/1708.06715} {arXiv:1708.06715 [hep-lat]} \BibitemShut
  {NoStop}%
\bibitem [{\citenamefont {Butt}\ \emph {et~al.}(2025)\citenamefont {Butt},
  \citenamefont {Catterall},\ and\ \citenamefont {Hasenfratz}}]{Butt:2024kxi}%
  \BibitemOpen
  \bibfield  {author} {\bibinfo {author} {\bibfnamefont {N.}~\bibnamefont
  {Butt}}, \bibinfo {author} {\bibfnamefont {S.}~\bibnamefont {Catterall}},\
  and\ \bibinfo {author} {\bibfnamefont {A.}~\bibnamefont {Hasenfratz}},\
  }\href {https://doi.org/10.1103/PhysRevLett.134.031602} {\bibfield  {journal}
  {\bibinfo  {journal} {Phys. Rev. Lett.}\ }\textbf {\bibinfo {volume} {134}},\
  \bibinfo {pages} {031602} (\bibinfo {year} {2025})},\ \Eprint
  {https://arxiv.org/abs/2409.02062} {arXiv:2409.02062 [hep-lat]} \BibitemShut
  {NoStop}%
\bibitem [{\citenamefont {Gorbenko}\ \emph {et~al.}(2018)\citenamefont
  {Gorbenko}, \citenamefont {Rychkov},\ and\ \citenamefont
  {Zan}}]{Gorbenko:2018ncu}%
  \BibitemOpen
  \bibfield  {author} {\bibinfo {author} {\bibfnamefont {V.}~\bibnamefont
  {Gorbenko}}, \bibinfo {author} {\bibfnamefont {S.}~\bibnamefont {Rychkov}},\
  and\ \bibinfo {author} {\bibfnamefont {B.}~\bibnamefont {Zan}},\ }\href
  {https://doi.org/10.1007/JHEP10(2018)108} {\bibfield  {journal} {\bibinfo
  {journal} {JHEP}\ }\textbf {\bibinfo {volume} {10}},\ \bibinfo {pages}
  {108}},\ \Eprint {https://arxiv.org/abs/1807.11512} {arXiv:1807.11512
  [hep-th]} \BibitemShut {NoStop}%
\bibitem [{\citenamefont {Kaplan}\ \emph {et~al.}(2009)\citenamefont {Kaplan},
  \citenamefont {Lee}, \citenamefont {Son},\ and\ \citenamefont
  {Stephanov}}]{Kaplan:2009kr}%
  \BibitemOpen
  \bibfield  {author} {\bibinfo {author} {\bibfnamefont {D.~B.}\ \bibnamefont
  {Kaplan}}, \bibinfo {author} {\bibfnamefont {J.-W.}\ \bibnamefont {Lee}},
  \bibinfo {author} {\bibfnamefont {D.~T.}\ \bibnamefont {Son}},\ and\ \bibinfo
  {author} {\bibfnamefont {M.~A.}\ \bibnamefont {Stephanov}},\ }\href
  {https://doi.org/10.1103/PhysRevD.80.125005} {\bibfield  {journal} {\bibinfo
  {journal} {Phys. Rev. D}\ }\textbf {\bibinfo {volume} {80}},\ \bibinfo
  {pages} {125005} (\bibinfo {year} {2009})},\ \Eprint
  {https://arxiv.org/abs/0905.4752} {arXiv:0905.4752 [hep-th]} \BibitemShut
  {NoStop}%
\bibitem [{\citenamefont {Gurarie}(2011)}]{Gurarie:2011}%
  \BibitemOpen
  \bibfield  {author} {\bibinfo {author} {\bibfnamefont {V.}~\bibnamefont
  {Gurarie}},\ }\href {https://doi.org/10.1103/PhysRevB.83.085426} {\bibfield
  {journal} {\bibinfo  {journal} {Phys. Rev. B}\ }\textbf {\bibinfo {volume}
  {83}},\ \bibinfo {pages} {085426} (\bibinfo {year} {2011})},\ \Eprint
  {https://arxiv.org/abs/1011.2273} {arXiv:1011.2273 [cond-mat.mes-hall]}
  \BibitemShut {NoStop}%
\bibitem [{\citenamefont {You}\ \emph {et~al.}(2014{\natexlab{a}})\citenamefont
  {You}, \citenamefont {Wang}, \citenamefont {Oon},\ and\ \citenamefont
  {Xu}}]{You:2014cka}%
  \BibitemOpen
  \bibfield  {author} {\bibinfo {author} {\bibfnamefont {Y.-Z.}\ \bibnamefont
  {You}}, \bibinfo {author} {\bibfnamefont {Z.}~\bibnamefont {Wang}}, \bibinfo
  {author} {\bibfnamefont {J.}~\bibnamefont {Oon}},\ and\ \bibinfo {author}
  {\bibfnamefont {C.}~\bibnamefont {Xu}},\ }\href
  {https://doi.org/10.1103/PhysRevB.90.060502} {\bibfield  {journal} {\bibinfo
  {journal} {Phys. Rev. B}\ }\textbf {\bibinfo {volume} {90}},\ \bibinfo
  {pages} {060502} (\bibinfo {year} {2014}{\natexlab{a}})},\ \Eprint
  {https://arxiv.org/abs/1403.4938} {arXiv:1403.4938 [cond-mat.str-el]}
  \BibitemShut {NoStop}%
\bibitem [{\citenamefont {Xu}\ and\ \citenamefont {Xu}(2021)}]{Xu:2021ztz}%
  \BibitemOpen
  \bibfield  {author} {\bibinfo {author} {\bibfnamefont {Y.}~\bibnamefont
  {Xu}}\ and\ \bibinfo {author} {\bibfnamefont {C.}~\bibnamefont {Xu}},\
  }\href@noop {} {\bibinfo {title} {{Green's function Zero and Symmetric Mass
  Generation}}} (\bibinfo {year} {2021}),\ \Eprint
  {https://arxiv.org/abs/2103.15865} {arXiv:2103.15865 [cond-mat.str-el]}
  \BibitemShut {NoStop}%
\bibitem [{\citenamefont {Lu}\ \emph {et~al.}(2023{\natexlab{a}})\citenamefont
  {Lu}, \citenamefont {Zeng}, \citenamefont {Wang},\ and\ \citenamefont
  {You}}]{Lu:2022qtc}%
  \BibitemOpen
  \bibfield  {author} {\bibinfo {author} {\bibfnamefont {D.-C.}\ \bibnamefont
  {Lu}}, \bibinfo {author} {\bibfnamefont {M.}~\bibnamefont {Zeng}}, \bibinfo
  {author} {\bibfnamefont {J.}~\bibnamefont {Wang}},\ and\ \bibinfo {author}
  {\bibfnamefont {Y.-Z.}\ \bibnamefont {You}},\ }\href
  {https://doi.org/10.1103/PhysRevB.107.195133} {\bibfield  {journal} {\bibinfo
   {journal} {Phys. Rev. B}\ }\textbf {\bibinfo {volume} {107}},\ \bibinfo
  {pages} {195133} (\bibinfo {year} {2023}{\natexlab{a}})},\ \Eprint
  {https://arxiv.org/abs/2210.16304} {arXiv:2210.16304 [cond-mat.str-el]}
  \BibitemShut {NoStop}%
\bibitem [{\citenamefont {Lu}\ \emph {et~al.}(2023{\natexlab{b}})\citenamefont
  {Lu}, \citenamefont {Zeng},\ and\ \citenamefont {You}}]{Lu:2023cev}%
  \BibitemOpen
  \bibfield  {author} {\bibinfo {author} {\bibfnamefont {D.-C.}\ \bibnamefont
  {Lu}}, \bibinfo {author} {\bibfnamefont {M.}~\bibnamefont {Zeng}},\ and\
  \bibinfo {author} {\bibfnamefont {Y.-Z.}\ \bibnamefont {You}},\ }\href
  {https://doi.org/10.1103/PhysRevB.108.205117} {\bibfield  {journal} {\bibinfo
   {journal} {Phys. Rev. B}\ }\textbf {\bibinfo {volume} {108}},\ \bibinfo
  {pages} {205117} (\bibinfo {year} {2023}{\natexlab{b}})},\ \Eprint
  {https://arxiv.org/abs/2307.12223} {arXiv:2307.12223 [cond-mat.str-el]}
  \BibitemShut {NoStop}%
\bibitem [{\citenamefont {Golterman}\ and\ \citenamefont
  {Shamir}(2026{\natexlab{b}})}]{Golterman:2026sox}%
  \BibitemOpen
  \bibfield  {author} {\bibinfo {author} {\bibfnamefont {M.}~\bibnamefont
  {Golterman}}\ and\ \bibinfo {author} {\bibfnamefont {Y.}~\bibnamefont
  {Shamir}},\ }\href {https://doi.org/10.22323/1.518.0156} {\bibfield
  {journal} {\bibinfo  {journal} {PoS}\ }\textbf {\bibinfo {volume}
  {LATTICE2025}},\ \bibinfo {pages} {156} (\bibinfo {year}
  {2026}{\natexlab{b}})},\ \Eprint {https://arxiv.org/abs/2603.15985}
  {arXiv:2603.15985 [hep-lat]} \BibitemShut {NoStop}%
\bibitem [{\citenamefont {Chen}\ \emph {et~al.}(2013)\citenamefont {Chen},
  \citenamefont {Giedt},\ and\ \citenamefont {Poppitz}}]{Chen:2012di}%
  \BibitemOpen
  \bibfield  {author} {\bibinfo {author} {\bibfnamefont {C.}~\bibnamefont
  {Chen}}, \bibinfo {author} {\bibfnamefont {J.}~\bibnamefont {Giedt}},\ and\
  \bibinfo {author} {\bibfnamefont {E.}~\bibnamefont {Poppitz}},\ }\href
  {https://doi.org/10.1007/JHEP04(2013)131} {\bibfield  {journal} {\bibinfo
  {journal} {JHEP}\ }\textbf {\bibinfo {volume} {04}},\ \bibinfo {pages}
  {131}},\ \Eprint {https://arxiv.org/abs/1211.6947} {arXiv:1211.6947
  [hep-lat]} \BibitemShut {NoStop}%
\bibitem [{\citenamefont {Poppitz}\ and\ \citenamefont
  {Shang}(2010)}]{Poppitz:2010at}%
  \BibitemOpen
  \bibfield  {author} {\bibinfo {author} {\bibfnamefont {E.}~\bibnamefont
  {Poppitz}}\ and\ \bibinfo {author} {\bibfnamefont {Y.}~\bibnamefont
  {Shang}},\ }\href {https://doi.org/10.1142/S0217751X10049852} {\bibfield
  {journal} {\bibinfo  {journal} {Int. J. Mod. Phys. A}\ }\textbf {\bibinfo
  {volume} {25}},\ \bibinfo {pages} {2761} (\bibinfo {year} {2010})},\ \Eprint
  {https://arxiv.org/abs/1003.5896} {arXiv:1003.5896 [hep-lat]} \BibitemShut
  {NoStop}%
\bibitem [{\citenamefont {Kikukawa}(2019{\natexlab{a}})}]{Kikukawa:2017gvk}%
  \BibitemOpen
  \bibfield  {author} {\bibinfo {author} {\bibfnamefont {Y.}~\bibnamefont
  {Kikukawa}},\ }\href {https://doi.org/10.1093/ptep/ptz055} {\bibfield
  {journal} {\bibinfo  {journal} {PTEP}\ }\textbf {\bibinfo {volume} {2019}},\
  \bibinfo {pages} {073B02} (\bibinfo {year} {2019}{\natexlab{a}})},\ \Eprint
  {https://arxiv.org/abs/1710.11101} {arXiv:1710.11101 [hep-lat]} \BibitemShut
  {NoStop}%
\bibitem [{\citenamefont {Zeng}\ \emph {et~al.}(2022)\citenamefont {Zeng},
  \citenamefont {Zhu}, \citenamefont {Wang},\ and\ \citenamefont
  {You}}]{Zeng:2022grc}%
  \BibitemOpen
  \bibfield  {author} {\bibinfo {author} {\bibfnamefont {M.}~\bibnamefont
  {Zeng}}, \bibinfo {author} {\bibfnamefont {Z.}~\bibnamefont {Zhu}}, \bibinfo
  {author} {\bibfnamefont {J.}~\bibnamefont {Wang}},\ and\ \bibinfo {author}
  {\bibfnamefont {Y.-Z.}\ \bibnamefont {You}},\ }\href
  {https://doi.org/10.1103/PhysRevLett.128.185301} {\bibfield  {journal}
  {\bibinfo  {journal} {Phys. Rev. Lett.}\ }\textbf {\bibinfo {volume} {128}},\
  \bibinfo {pages} {185301} (\bibinfo {year} {2022})},\ \Eprint
  {https://arxiv.org/abs/2202.12355} {arXiv:2202.12355 [cond-mat.str-el]}
  \BibitemShut {NoStop}%
\bibitem [{\citenamefont {Araki}\ \emph {et~al.}(2026)\citenamefont {Araki},
  \citenamefont {Fukaya}, \citenamefont {Onogi},\ and\ \citenamefont
  {Yamaguchi}}]{Araki:2026mqi}%
  \BibitemOpen
  \bibfield  {author} {\bibinfo {author} {\bibfnamefont {S.}~\bibnamefont
  {Araki}}, \bibinfo {author} {\bibfnamefont {H.}~\bibnamefont {Fukaya}},
  \bibinfo {author} {\bibfnamefont {T.}~\bibnamefont {Onogi}},\ and\ \bibinfo
  {author} {\bibfnamefont {S.}~\bibnamefont {Yamaguchi}},\ }\href@noop {}
  {\bibinfo {title} {{Symmetric Mass Generation for Domain-Wall Fermions}}}
  (\bibinfo {year} {2026}),\ \Eprint {https://arxiv.org/abs/2608.29963}
  {arXiv:2608.29963 [hep-lat]} \BibitemShut {NoStop}%
\bibitem [{\citenamefont {Maiti}\ \emph {et~al.}(2025)\citenamefont {Maiti},
  \citenamefont {Banerjee}, \citenamefont {Chandrasekharan},\ and\
  \citenamefont {Marinkovic}}]{Maiti:2025ley}%
  \BibitemOpen
  \bibfield  {author} {\bibinfo {author} {\bibfnamefont {S.}~\bibnamefont
  {Maiti}}, \bibinfo {author} {\bibfnamefont {D.}~\bibnamefont {Banerjee}},
  \bibinfo {author} {\bibfnamefont {S.}~\bibnamefont {Chandrasekharan}},\ and\
  \bibinfo {author} {\bibfnamefont {M.~K.}\ \bibnamefont {Marinkovic}},\
  }\href@noop {} {\bibinfo {title} {{Symmetric mass generation as a
  multicritical point with enhanced symmetry}}} (\bibinfo {year} {2025}),\
  \Eprint {https://arxiv.org/abs/2512.24836} {arXiv:2512.24836 [hep-lat]}
  \BibitemShut {NoStop}%
\bibitem [{\citenamefont {Maiti}\ \emph {et~al.}(2026)\citenamefont {Maiti},
  \citenamefont {Banerjee}, \citenamefont {Chandrasekharan},\ and\
  \citenamefont {Marinkovic}}]{Maiti:2026log}%
  \BibitemOpen
  \bibfield  {author} {\bibinfo {author} {\bibfnamefont {S.}~\bibnamefont
  {Maiti}}, \bibinfo {author} {\bibfnamefont {D.}~\bibnamefont {Banerjee}},
  \bibinfo {author} {\bibfnamefont {S.}~\bibnamefont {Chandrasekharan}},\ and\
  \bibinfo {author} {\bibfnamefont {M.~K.}\ \bibnamefont {Marinkovic}},\
  }\href@noop {} {\bibinfo {title} {{Phase diagram of a lattice fermion model
  with symmetric mass generation}}} (\bibinfo {year} {2026}),\ \Eprint
  {https://arxiv.org/abs/2602.18360} {arXiv:2602.18360 [hep-lat]} \BibitemShut
  {NoStop}%
\bibitem [{\citenamefont {Brower}\ \emph {et~al.}(2016)\citenamefont {Brower},
  \citenamefont {Hasenfratz}, \citenamefont {Rebbi}, \citenamefont {Weinberg},\
  and\ \citenamefont {Witzel}}]{Brower:2015owo}%
  \BibitemOpen
  \bibfield  {author} {\bibinfo {author} {\bibfnamefont {R.~C.}\ \bibnamefont
  {Brower}}, \bibinfo {author} {\bibfnamefont {A.}~\bibnamefont {Hasenfratz}},
  \bibinfo {author} {\bibfnamefont {C.}~\bibnamefont {Rebbi}}, \bibinfo
  {author} {\bibfnamefont {E.}~\bibnamefont {Weinberg}},\ and\ \bibinfo
  {author} {\bibfnamefont {O.}~\bibnamefont {Witzel}},\ }\href
  {https://doi.org/10.1103/PhysRevD.93.075028} {\bibfield  {journal} {\bibinfo
  {journal} {Phys. Rev. D}\ }\textbf {\bibinfo {volume} {93}},\ \bibinfo
  {pages} {075028} (\bibinfo {year} {2016})},\ \Eprint
  {https://arxiv.org/abs/1512.02576} {arXiv:1512.02576 [hep-ph]} \BibitemShut
  {NoStop}%
\bibitem [{\citenamefont {Hasenfratz}\ \emph {et~al.}(2017)\citenamefont
  {Hasenfratz}, \citenamefont {Rebbi},\ and\ \citenamefont
  {Witzel}}]{Hasenfratz:2016gut}%
  \BibitemOpen
  \bibfield  {author} {\bibinfo {author} {\bibfnamefont {A.}~\bibnamefont
  {Hasenfratz}}, \bibinfo {author} {\bibfnamefont {C.}~\bibnamefont {Rebbi}},\
  and\ \bibinfo {author} {\bibfnamefont {O.}~\bibnamefont {Witzel}},\ }\href
  {https://doi.org/10.1016/j.physletb.2017.07.058} {\bibfield  {journal}
  {\bibinfo  {journal} {Phys. Lett. B}\ }\textbf {\bibinfo {volume} {773}},\
  \bibinfo {pages} {86} (\bibinfo {year} {2017})},\ \Eprint
  {https://arxiv.org/abs/1609.01401} {arXiv:1609.01401 [hep-ph]} \BibitemShut
  {NoStop}%
\bibitem [{\citenamefont {Wilczek}\ and\ \citenamefont
  {Zee}(1982)}]{Wilczek:1981iz}%
  \BibitemOpen
  \bibfield  {author} {\bibinfo {author} {\bibfnamefont {F.}~\bibnamefont
  {Wilczek}}\ and\ \bibinfo {author} {\bibfnamefont {A.}~\bibnamefont {Zee}},\
  }\href {https://doi.org/10.1103/PhysRevD.25.553} {\bibfield  {journal}
  {\bibinfo  {journal} {Phys. Rev. D}\ }\textbf {\bibinfo {volume} {25}},\
  \bibinfo {pages} {553} (\bibinfo {year} {1982})}\BibitemShut {NoStop}%
\bibitem [{\citenamefont {Wang}(2021)}]{Wang:2020mra}%
  \BibitemOpen
  \bibfield  {author} {\bibinfo {author} {\bibfnamefont {J.}~\bibnamefont
  {Wang}},\ }\href {https://doi.org/10.1103/PhysRevD.103.105024} {\bibfield
  {journal} {\bibinfo  {journal} {Phys. Rev. D}\ }\textbf {\bibinfo {volume}
  {103}},\ \bibinfo {pages} {105024} (\bibinfo {year} {2021})},\ \Eprint
  {https://arxiv.org/abs/2012.15860} {arXiv:2012.15860 [hep-th]} \BibitemShut
  {NoStop}%
\bibitem [{\citenamefont {Wen}(2013)}]{Wen:2013ppa}%
  \BibitemOpen
  \bibfield  {author} {\bibinfo {author} {\bibfnamefont {X.-G.}\ \bibnamefont
  {Wen}},\ }\href {https://doi.org/10.1088/0256-307X/30/11/111101} {\bibfield
  {journal} {\bibinfo  {journal} {Chin. Phys. Lett.}\ }\textbf {\bibinfo
  {volume} {30}},\ \bibinfo {pages} {111101} (\bibinfo {year} {2013})},\
  \Eprint {https://arxiv.org/abs/1305.1045} {arXiv:1305.1045 [hep-lat]}
  \BibitemShut {NoStop}%
\bibitem [{\citenamefont {You}\ \emph {et~al.}(2014{\natexlab{b}})\citenamefont
  {You}, \citenamefont {BenTov},\ and\ \citenamefont {Xu}}]{You:2014oaa}%
  \BibitemOpen
  \bibfield  {author} {\bibinfo {author} {\bibfnamefont {Y.}~\bibnamefont
  {You}}, \bibinfo {author} {\bibfnamefont {Y.}~\bibnamefont {BenTov}},\ and\
  \bibinfo {author} {\bibfnamefont {C.}~\bibnamefont {Xu}},\ }\href@noop {}
  {\bibinfo {title} {{Interacting Topological Superconductors and possible
  Origin of $16n$ Chiral Fermions in the Standard Model}}} (\bibinfo {year}
  {2014}{\natexlab{b}}),\ \Eprint {https://arxiv.org/abs/1402.4151}
  {arXiv:1402.4151 [cond-mat.str-el]} \BibitemShut {NoStop}%
\bibitem [{\citenamefont {Kikukawa}(2019{\natexlab{b}})}]{Kikukawa:2017ngf}%
  \BibitemOpen
  \bibfield  {author} {\bibinfo {author} {\bibfnamefont {Y.}~\bibnamefont
  {Kikukawa}},\ }\href {https://doi.org/10.1093/ptep/ptz115} {\bibfield
  {journal} {\bibinfo  {journal} {PTEP}\ }\textbf {\bibinfo {volume} {2019}},\
  \bibinfo {pages} {113B03} (\bibinfo {year} {2019}{\natexlab{b}})},\ \Eprint
  {https://arxiv.org/abs/1710.11618} {arXiv:1710.11618 [hep-lat]} \BibitemShut
  {NoStop}%
\bibitem [{\citenamefont {Wang}\ and\ \citenamefont
  {Wen}(2018)}]{Wang:2018ugf}%
  \BibitemOpen
  \bibfield  {author} {\bibinfo {author} {\bibfnamefont {J.}~\bibnamefont
  {Wang}}\ and\ \bibinfo {author} {\bibfnamefont {X.-G.}\ \bibnamefont {Wen}},\
  }\href {https://doi.org/10.1103/PhysRevD.99.111501} {\bibfield  {journal}
  {\bibinfo  {journal} {Phys. Rev. D}\ }\textbf {\bibinfo {volume} {99}},\
  \bibinfo {pages} {111501} (\bibinfo {year} {2018})},\ \Eprint
  {https://arxiv.org/abs/1807.05998} {arXiv:1807.05998 [hep-lat]} \BibitemShut
  {NoStop}%
\bibitem [{\citenamefont {Wang}\ and\ \citenamefont
  {Wen}(2020)}]{Wang:2018jkc}%
  \BibitemOpen
  \bibfield  {author} {\bibinfo {author} {\bibfnamefont {J.}~\bibnamefont
  {Wang}}\ and\ \bibinfo {author} {\bibfnamefont {X.-G.}\ \bibnamefont {Wen}},\
  }\href {https://doi.org/10.1103/PhysRevResearch.2.023356} {\bibfield
  {journal} {\bibinfo  {journal} {Phys. Rev. Res.}\ }\textbf {\bibinfo {volume}
  {2}},\ \bibinfo {pages} {023356} (\bibinfo {year} {2020})},\ \Eprint
  {https://arxiv.org/abs/1809.11171} {arXiv:1809.11171 [hep-th]} \BibitemShut
  {NoStop}%
\bibitem [{\citenamefont {Ayyar}\ and\ \citenamefont
  {Chandrasekharan}(2015)}]{Ayyar:2014eua}%
  \BibitemOpen
  \bibfield  {author} {\bibinfo {author} {\bibfnamefont {V.}~\bibnamefont
  {Ayyar}}\ and\ \bibinfo {author} {\bibfnamefont {S.}~\bibnamefont
  {Chandrasekharan}},\ }\href {https://doi.org/10.1103/PhysRevD.91.065035}
  {\bibfield  {journal} {\bibinfo  {journal} {Phys. Rev. D}\ }\textbf {\bibinfo
  {volume} {91}},\ \bibinfo {pages} {065035} (\bibinfo {year} {2015})},\
  \Eprint {https://arxiv.org/abs/1410.6474} {arXiv:1410.6474 [hep-lat]}
  \BibitemShut {NoStop}%
\bibitem [{\citenamefont {Ayyar}\ and\ \citenamefont
  {Chandrasekharan}(2016{\natexlab{b}})}]{Ayyar:2015lrd}%
  \BibitemOpen
  \bibfield  {author} {\bibinfo {author} {\bibfnamefont {V.}~\bibnamefont
  {Ayyar}}\ and\ \bibinfo {author} {\bibfnamefont {S.}~\bibnamefont
  {Chandrasekharan}},\ }\href {https://doi.org/10.1103/PhysRevD.93.081701}
  {\bibfield  {journal} {\bibinfo  {journal} {Phys. Rev. D}\ }\textbf {\bibinfo
  {volume} {93}},\ \bibinfo {pages} {081701} (\bibinfo {year}
  {2016}{\natexlab{b}})},\ \Eprint {https://arxiv.org/abs/1511.09071}
  {arXiv:1511.09071 [hep-lat]} \BibitemShut {NoStop}%
\bibitem [{\citenamefont {Butt}\ \emph
  {et~al.}(2021{\natexlab{b}})\citenamefont {Butt}, \citenamefont {Catterall},\
  and\ \citenamefont {Toga}}]{Butt:2021koj}%
  \BibitemOpen
  \bibfield  {author} {\bibinfo {author} {\bibfnamefont {N.}~\bibnamefont
  {Butt}}, \bibinfo {author} {\bibfnamefont {S.}~\bibnamefont {Catterall}},\
  and\ \bibinfo {author} {\bibfnamefont {G.~C.}\ \bibnamefont {Toga}},\ }\href
  {https://doi.org/10.3390/sym13122276} {\bibfield  {journal} {\bibinfo
  {journal} {Symmetry}\ }\textbf {\bibinfo {volume} {13}},\ \bibinfo {pages}
  {2276} (\bibinfo {year} {2021}{\natexlab{b}})},\ \Eprint
  {https://arxiv.org/abs/2111.01001} {arXiv:2111.01001 [hep-lat]} \BibitemShut
  {NoStop}%
\bibitem [{\citenamefont {Hasenfratz}(2022)}]{Hasenfratz:2022qan}%
  \BibitemOpen
  \bibfield  {author} {\bibinfo {author} {\bibfnamefont {A.}~\bibnamefont
  {Hasenfratz}},\ }\href {https://doi.org/10.1103/PhysRevD.106.014513}
  {\bibfield  {journal} {\bibinfo  {journal} {Phys. Rev. D}\ }\textbf {\bibinfo
  {volume} {106}},\ \bibinfo {pages} {014513} (\bibinfo {year} {2022})},\
  \Eprint {https://arxiv.org/abs/2204.04801} {arXiv:2204.04801 [hep-lat]}
  \BibitemShut {NoStop}%
\bibitem [{\citenamefont {Slagle}\ \emph {et~al.}(2015)\citenamefont {Slagle},
  \citenamefont {You},\ and\ \citenamefont {Xu}}]{Slagle:2014vma}%
  \BibitemOpen
  \bibfield  {author} {\bibinfo {author} {\bibfnamefont {K.}~\bibnamefont
  {Slagle}}, \bibinfo {author} {\bibfnamefont {Y.-Z.}\ \bibnamefont {You}},\
  and\ \bibinfo {author} {\bibfnamefont {C.}~\bibnamefont {Xu}},\ }\href
  {https://doi.org/10.1103/PhysRevB.91.115121} {\bibfield  {journal} {\bibinfo
  {journal} {Phys. Rev. B}\ }\textbf {\bibinfo {volume} {91}},\ \bibinfo
  {pages} {115121} (\bibinfo {year} {2015})},\ \Eprint
  {https://arxiv.org/abs/1409.7401} {arXiv:1409.7401 [cond-mat.str-el]}
  \BibitemShut {NoStop}%
\bibitem [{\citenamefont {He}\ \emph {et~al.}(2016)\citenamefont {He},
  \citenamefont {Wu}, \citenamefont {You}, \citenamefont {Xu}, \citenamefont
  {Meng},\ and\ \citenamefont {Lu}}]{He:2016sbs}%
  \BibitemOpen
  \bibfield  {author} {\bibinfo {author} {\bibfnamefont {Y.-Y.}\ \bibnamefont
  {He}}, \bibinfo {author} {\bibfnamefont {H.-Q.}\ \bibnamefont {Wu}}, \bibinfo
  {author} {\bibfnamefont {Y.-Z.}\ \bibnamefont {You}}, \bibinfo {author}
  {\bibfnamefont {C.}~\bibnamefont {Xu}}, \bibinfo {author} {\bibfnamefont
  {Z.~Y.}\ \bibnamefont {Meng}},\ and\ \bibinfo {author} {\bibfnamefont
  {Z.-Y.}\ \bibnamefont {Lu}},\ }\href
  {https://doi.org/10.1103/PhysRevB.94.241111} {\bibfield  {journal} {\bibinfo
  {journal} {Phys. Rev. B}\ }\textbf {\bibinfo {volume} {94}},\ \bibinfo
  {pages} {241111} (\bibinfo {year} {2016})},\ \Eprint
  {https://arxiv.org/abs/1603.08376} {arXiv:1603.08376 [cond-mat.str-el]}
  \BibitemShut {NoStop}%
\bibitem [{\citenamefont {He}\ \emph {et~al.}(2026)\citenamefont {He},
  \citenamefont {You},\ and\ \citenamefont {Xu}}]{He:2026oaf}%
  \BibitemOpen
  \bibfield  {author} {\bibinfo {author} {\bibfnamefont {C.-H.}\ \bibnamefont
  {He}}, \bibinfo {author} {\bibfnamefont {Y.-Z.}\ \bibnamefont {You}},\ and\
  \bibinfo {author} {\bibfnamefont {X.~Y.}\ \bibnamefont {Xu}},\ }\href
  {https://doi.org/10.1103/dn3v-f483} {\bibfield  {journal} {\bibinfo
  {journal} {Phys. Rev. B}\ }\textbf {\bibinfo {volume} {114}},\ \bibinfo
  {pages} {165125} (\bibinfo {year} {2026})},\ \Eprint
  {https://arxiv.org/abs/2603.18278} {arXiv:2603.18278 [cond-mat.str-el]}
  \BibitemShut {NoStop}%
\bibitem [{\citenamefont {Martin}\ and\ \citenamefont
  {Grover}(2025)}]{Martin:2025hic}%
  \BibitemOpen
  \bibfield  {author} {\bibinfo {author} {\bibfnamefont {S.}~\bibnamefont
  {Martin}}\ and\ \bibinfo {author} {\bibfnamefont {T.}~\bibnamefont
  {Grover}},\ }\href@noop {} {\bibinfo {title} {{A Perturbative Approach to
  Symmetric Mass Generation}}} (\bibinfo {year} {2025}),\ \Eprint
  {https://arxiv.org/abs/2507.23032} {arXiv:2507.23032 [cond-mat.str-el]}
  \BibitemShut {NoStop}%
\bibitem [{\citenamefont {Golterman}\ and\ \citenamefont
  {Petcher}(1989)}]{Golterman:1989df}%
  \BibitemOpen
  \bibfield  {author} {\bibinfo {author} {\bibfnamefont {M.~F.~L.}\
  \bibnamefont {Golterman}}\ and\ \bibinfo {author} {\bibfnamefont {D.~N.}\
  \bibnamefont {Petcher}},\ }\href
  {https://doi.org/10.1016/0370-2693(89)91028-9} {\bibfield  {journal}
  {\bibinfo  {journal} {Phys. Lett. B}\ }\textbf {\bibinfo {volume} {225}},\
  \bibinfo {pages} {159} (\bibinfo {year} {1989})}\BibitemShut {NoStop}%
\end{thebibliography}%

\end{document}